\documentclass{lmcs}
\pdfoutput=1

\usepackage{lastpage}
\lmcsdoi{17}{4}{6}
\lmcsheading{}{\pageref{LastPage}}{}{}%
{Dec.~02,~2020}{Nov.~03,~2021}{}

\usepackage[utf8]{inputenc}

\keywords{Timed automata, interval Markov chains, probabilistic model checking}

\usepackage{tikz,pgf}
\usetikzlibrary{matrix}
\usepackage{array,longtable}
\newcolumntype{L}{>{$ \displaystyle}l<{$}}
\newcolumntype{C}{>{$ \displaystyle}c<{$}}
\newcolumntype{R}{>{$ \displaystyle}r<{$}}

\usetikzlibrary{automata,calc,arrows,shapes,backgrounds}

\definecolor{ly}{RGB}{250,250,200}
\definecolor{lg}{RGB}{200,250,200}

\def\Nset{\mathbb{N}}

\def\Qset{\mathbb{Q}}
\def\Rset{\mathbb{R}}

\newcommand{\true}{\mbox{\tt true}}

\def\ra{\rightarrow}

\newcommand{\sat}{\models}

\newcommand{\sectref}[1]{Section~\ref{#1}}

\newcommand{\propref}[1]{Proposition~\ref{#1}}

\newcommand{\lemref}[1]{Lemma~\ref{#1}}
\newcommand{\figref}[1]{Figure~\ref{#1}}

\newcommand{\sem}[1]{{[\hspace{-0.05cm}[#1]\hspace{-0.05cm}]}}

\newcommand{\comment}[1]{}

\newcommand{\setsep}{\mid}

\newcommand{\IGNORE}[1]{}

\newcommand{\nnr}{\Rset_{\geq 0}}

\newcommand{\dist}{\mathsf{Dist}}
\newcommand{\adist}{\mu}

\newcommand{\support}{\mathsf{support}}

\newcommand{\anMDP}{\mathcal{M}}

\newcommand{\ptsstates}{S}
\newcommand{\astate}{s}

\newcommand{\anaction}{a}

\newcommand{\afinrun}{r}

\newcommand{\last}{\mathit{last}}

\newcommand{\astrat}{\sigma}

\newcommand{\probl}[2]{\mathrm{Pr}^{#1}_{#2}}

\newcommand{\aclock}{x}
\newcommand{\anotherclock}{y}

\newcommand{\aval}{v}

\newcommand{\aclockset}{X}

\newcommand{\pedges}{\mathit{prob}}
\newcommand{\apedge}{p}
\newcommand{\pd}{\wp}

\newcommand{\locs}{L}

\newcommand{\inv}{\mathit{inv}}

\newcommand{\aloc}{l}

\newcommand{\g}{g}

\newcommand{\vinz}{\sat}

\newcommand{\ccons}{\mathit{\Psi}}
\newcommand{\ccl}{\psi}

\newcommand{\maxval}[3]{\mathbb{P}^\mathrm{max}_{{#1},{#3}}(\Diamond {#2})}
\newcommand{\minval}[3]{\mathbb{P}^\mathrm{min}_{{#1},{#3}}(\Diamond {#2})}

\newcommand{\Sup}[1]{\sup_{#1}}
\newcommand{\Inf}[1]{\inf_{#1}}

\newcommand{\adelay}{t}

\newcommand{\pthresh}{\lambda}

\newcommand{\adisttemp}{\pd}
\newcommand{\adtpara}[2]{{#1}[{#2}]}
\newcommand{\disttemps}{\mathsf{Temp}}

\newcommand{\acdpta}{\mathcal{P}}

\newcommand{\interval}[1]{I^{#1}}

\newcommand{\outcome}{e}

\newcommand{\fc}[2]{c^{#1}_{#2}}
\newcommand{\fd}[2]{d^{#1}_{#2}}

\newcommand{\rsuccact}{\tau}

\newcommand{\set}[1]{\{{#1}\}}
\newcommand{\ldot}{\, . \,} 

\newcommand{\intervals}{\mathcal{I}}
\newcommand{\anint}{I}
\newcommand{\lep}[1]{\mathsf{left}({#1})}
\newcommand{\rep}[1]{\mathsf{right}({#1})}

\newcommand{\adtmc}{\mathcal{C}}
\newcommand{\dstates}{S}
\newcommand{\dtrans}{\mathbf{P}}
\newcommand{\dpaths}[1]{\mathit{Paths}^{#1}}

\newcommand{\dprob}[2]{\mathrm{Pr}^{#1}_{#2}}
\newcommand{\adtmcparaone}[1]{\adtmc^{#1}}

\newcommand{\apath}{\mathbf{\afinpath}}

\newcommand{\anmdp}{\mathcal{M}}
\newcommand{\mstates}{S}
\newcommand{\mactions}{A}
\newcommand{\mtrans}{\Delta}
\newcommand{\noaction}{\bot}
\newcommand{\mactionsof}[1]{\mactions({#1})}

\newcommand{\mpaths}[1]{\mathit{Paths}^{#1}}
\newcommand{\mfinpaths}[1]{\mathit{Paths}^{#1}_*}
\newcommand{\afinpath}{r}
\newcommand{\afinpathpref}[1]{\afinpath_{#1}}
\newcommand{\mprob}[2]{\mathrm{Pr}^{#1}_{#2}}
\newcommand{\asched}{\sigma}
\newcommand{\ansched}{\pi}

\newcommand{\anoimc}{\mathfrak{C}}
\newcommand{\ocstates}{S}
\newcommand{\octrans}{\mathfrak{P}}

\newcommand{\apa}[1]{\mathbf{a}}

\newcommand{\targetset}{T}

\newcommand{\strats}[1]{\Sigma^{#1}}

\newcommand{\anoimdp}{\mathfrak{M}}
\newcommand{\omstates}{S}
\newcommand{\omactions}{\mathfrak{A}}
\newcommand{\omtrans}{\mathfrak{D}}

\newcommand{\omactionsof}[1]{\omactions({#1})}
\newcommand{\anomact}{\mathfrak{a}}
\newcommand{\omstacts}{\omstates \otimes \omactions}

\newcommand{\cmTOm}{f}
\newcommand{\mTOcm}{g}

\newcommand{\anintdist}{\mathfrak{d}}
\newcommand{\intdist}{\mathfrak{Dist}}
\newcommand{\ana}{\alpha}
\newcommand{\anab}{\beta}
\newcommand{\anac}{\gamma}
\newcommand{\ass}[1]{\mathfrak{G}({#1})}

\newcommand{\anabof}[2]{\anab_{({#1},{#2})}}

\newcommand{\oimcfrom}[1]{\anoimc[{#1}]}
\newcommand{\cfstates}{\tilde{\ocstates}}
\newcommand{\cftrans}{\tilde{\octrans}}

\newcommand{\probfin}[2]{\mathrm{Pr}^{#1}_{*,#2}}

\newcommand{\afincmpath}{\afinpath}
\newcommand{\mfinpathsendsin}[2]{\mathit{Paths}^{#1}_{*,{#2}}}

\newcommand{\pol}[2]{\mathit{pol}({#1},{#2})}

\newcommand{\cyl}{\mathsf{Cyl}}

\newcommand{\semmdp}[1]{\sem{#1}}

\newcommand{\anval}{w}

\newcommand{\targetlocs}{F}
\newcommand{\targetsetof}[1]{\targetset_{#1}}
\newcommand{\targetsetimdpof}[1]{\mathfrak{T}_{#1}}

\newcommand{\closure}[1]{\overline{#1}}

\newcommand{\constants}[1]{\mathit{Cst}({#1})}
\newcommand{\ab}{b}
\newcommand{\borders}{\mathbb{B}}
\newcommand{\intervalsof}[1]{\intervals_{#1}}
\newcommand{\abint}{B}

\newcommand{\oimdpfrom}[1]{\anoimdp[{#1}]}

\newcommand{\omfromstates}[1]{\omstates_{#1}}
\newcommand{\omfromactions}[1]{\omactions_{#1}}
\newcommand{\omfromtrans}[1]{\omtrans_{#1}}
\newcommand{\omfromactionsof}[2]{\omactions_{#1}({#2})}

\newcommand{\mfromstates}[1]{\mstates_{#1}}
\newcommand{\mfromactions}[1]{\mactions_{#1}}
\newcommand{\mfromtrans}[1]{\mtrans_{#1}}

\newcommand{\timeelapse}[1]{\tilde{#1}}

\newcommand{\regof}[1]{\mathsf{reg}({#1})}

\newcommand{\sourceof}[1]{\mathit{source}({#1})}
\newcommand{\guardof}[1]{\mathit{guard}({#1})}

\newcommand{\dir}{\mathsf{ep}}
\newcommand{\lb}{\mathsf{le}}
\newcommand{\rb}{\mathsf{re}}

\newcommand{\regstates}{\omfromstates{\oimdpfrom{\acdpta}}^\mathsf{reg}}
\newcommand{\eis}{\omfromstates{\oimdpfrom{\acdpta}}^\mathsf{end}}

\newcommand{\dummyact}{\tau}
\newcommand{\probof}[2]{\lambda^{#1}_{#2}}

\newcommand{\bstrats}[1]{\Sigma^{#1}_{\borders}}

\newcommand{\abfinpath}{\rho}
\newcommand{\abfinpathprefix}[1]{\abfinpath_{#1}}

\newcommand{\corrtwo}[2]{\Lambda^{#1}({#2})}

\newcommand{\timeelapsespecial}[3]{\tilde{\aval}^\mathrm{*}_{{#1},({#2},{#3})}}

\newcommand{\valnset}[3]{\Xi_{{#1},({#2},{#3})}}

\newcommand{\bactsset}{\Upsilon}

\newcommand{\fcno}{\fc{}{}}
\newcommand{\fdno}{\fd{}{}}

\newcommand{\probfinone}[1]{\mathrm{Pr}^{#1}_{*}}

\newcommand{\anosched}{\hat{\ansched}}

\newcommand{\dummyassof}[1]{\ana^{#1}}

\newcommand{\afinpathrepr}[2]{\hat{\afinpath}} 

\newcommand{\assTOvaln}[2]{h_{#1}}

\newcommand{\futfpaths}[3]{\Omega_{{#1},({#2},{#3})}}
\newcommand{\finalass}[1]{\mathsf{FA}({#1})}
\newcommand{\msuffpathfrom}[2]{\mathit{SuffFrom}_*^{#1}({#2})}
\newcommand{\msuffpath}[1]{\mathit{Suffixes}_*^{#1}}
\newcommand{\critical}[3]{\mathit{Critical}_{#1}^{#2}}
\newcommand{\futass}[3]{\Psi_{{#1},({#2},{#3})}}

\newcommand{\anaof}[3]{\ana^{#1}_{({#2},{#3})}}

\newcommand{\afinpathspecial}{\check{\afinpath}}

\newcommand{\yanotherclock}{z}

\begin{document}

\title[PTA with One Clock and Initialised Clock-Dependent Probabilities]{Probabilistic Timed Automata with One Clock and Initialised Clock-Dependent Probabilities}

\author[Jeremy~Sproston]{Jeremy Sproston}
\address{Dipartimento di Informatica, University of Turin, Italy}
\email{sproston@di.unito.it}


\begin{abstract}
Clock-dependent probabilistic timed automata extend classical timed automata with discrete probabilistic choice,
where the probabilities are allowed to depend on the exact values of the clocks.
Previous work has shown that the quantitative reachability problem
for clock-dependent probabilistic timed automata with at least three clocks is undecidable.
In this paper, we consider the subclass of clock-dependent probabilistic timed automata
that have \emph{one clock},
that have clock dependencies described by affine functions,
and that satisfy an \emph{initialisation} condition
requiring that,
at some point between taking edges with non-trivial clock dependencies,
the clock must have an integer value.
We present an approach for solving in polynomial time
quantitative and qualitative reachability problems of
such one-clock initialised clock-dependent probabilistic timed automata.
Our results are obtained by a transformation to interval Markov decision processes.
\end{abstract}
%

\maketitle

\allowdisplaybreaks%

\section{Introduction}\label{sec:intro}
The diffusion of complex systems with timing requirements
that operate in unpredictable environments
has led to interest in formal modelling and verification techniques for
timed and probabilistic systems.
Model checking~\cite{CGP01,BK08} is an example of a formal verification technique,
and comprises the automatic verification of a system model against formally-specified properties.
%
A well-established modelling formalism for timed systems is \emph{timed automata}~\cite{AD94}.
A timed automaton consists of a finite graph equipped with
a set of real-valued variables called \emph{clocks},
which increase at the same rate as real time and which can be used
to constrain the relative time of events.
To model probabilistic systems formally,
frameworks such as Markov chains or Markov decision processes are used typically.
Model-checking algorithms for these formalisms have been presented in the literature:
for overviews of these techniques see, for example,~\cite{BFLMOW18} for timed automata,
and~\cite{BK08,FKNP11} for Markov chains and Markov decision processes.
Furthermore, timed automata and Markov decision processes
have been combined to obtain
the formalism of \emph{probabilistic timed automata}~\cite{GJ95,KNSS02,NPS13},
which can be viewed as timed automata with probabilities associated with their edges
(or, equivalently, as Markov decision processes equipped with clocks and their associated constraints).
%


\begin{figure}[t]
\centering
%

\centering
\tiny

\begin{tikzpicture}[->,>=stealth',shorten >=1pt,auto, thin] 
  \tikzstyle{state}=[draw=black, text=black, shape=circle, inner sep=3pt, outer sep=0pt, circle, rounded corners] 
  \tikzstyle{final_state}=[draw=black, text=black, shape=circle, inner sep=3pt, outer sep=0pt, circle, rounded corners] 


	\node[state](W){$\mathrm{W}$};
	\node[yshift=-0.55cm,above of = W](inv_W){$\aclock < 3$};

	\node[fill=black, node distance=1.7cm] (nail_W) [below of=W] {};

	\node[state, node distance=1.7cm](S)[left of=nail_W]{$\mathrm{S}$};

	\node[state, node distance=1.7cm](T)[below of=nail_W]{$\mathrm{T}$};

	\node[state, node distance=1.7cm](F)[right of=nail_W]{$\mathrm{F}$};
	\node[yshift=-0.6cm,above of = F](inv_F){$\aclock < 5$};

	\node[fill=black, node distance=1.7cm] (nail_F) [right of=F] {};

	\path[rounded corners]

	(W)
		edge [right] node {$1<\aclock<3$} (nail_W)

	(nail_W)
		edge [above] node {\colorbox{gray!15}{$\frac{3\aclock - 3}{8}$}} (S)
		edge [left] node[yshift=-2mm] {\colorbox{gray!15}{$\frac{11 - 3\aclock}{16}$}} (T)
		edge [below] node {\colorbox{gray!15}{$\frac{11 - 3\aclock}{16}$}} (F)

	(F)
		edge [above] node {$4<\aclock<5$} (nail_F);


\draw[->,rounded corners] (nail_F.north) |- node[above,text=black,pos=.7]{\colorbox{gray!15}{$\frac{\aclock-4}{2}$}} (W.east);
\draw[->,rounded corners] (nail_F.north) |- node[below,text=black,pos=.7]{$\set{\aclock}$} (W.east);
\draw[->,rounded corners] (nail_F.south) |- node[above,text=black,pos=.7]{\colorbox{gray!15}{$\frac{6-\aclock}{2}$}} (T.east);

%

\end{tikzpicture}


\caption{An example of one-clock clock-dependent probabilistic automaton $\acdpta$.}%
\label{fig:one_clock}

\end{figure}


For the modelling of certain systems,
it may be advantageous to model the fact that the probability of some events,
in particular those concerning the environment in which the system is operating,
vary as time passes.
For example, in automotive and aeronautic contexts,
the probability of certain reactions of human operators may depend on factors such as fatigue,
which can increase over time
(see, for example,~\cite{FWHT16});
an increase in the amount of time that an unmanned aerial vehicle spends performing a search and rescue operation
in a hazardous zone may increase the probability that the vehicle incurs damage from the environment;
an increase in the time elapsed before a metro train arrives at a station
can result in an increase in the number of passengers on the station's platform,
which can in turn increase the probability of the doors failing to shut at the station,
due to overcrowding of the train
(see~\cite{BBHMPS19}).
A natural way of representing such a dependency of probability of events on time
is using a continuous function:
for example,
for the case in which a task can be completed between $1$ and $3$ time units in the future,
we could represent the successful completion of the task by probability $\frac{\aclock+1}{4}$,
where the clock variable $\aclock$
(measuring the amount of time elapsed)
ranges over the interval $[1,3]$.
The standard probabilistic timed automaton formalism cannot express such a continuous relationship between
probabilities and time,
being limited to step functions
(where the intervals along which the function is constant must have rational-numbered endpoints).
This limitation led to the development of an extension of probabilistic timed automata
called \emph{clock-dependent probabilistic timed automata}~\cite{Spr21},
in which the probabilities of crossing edges can depend on clock values according to piecewise constant functions.
\figref{fig:one_clock} gives an example of such a clock-dependent probabilistic timed automaton,
using the standard conventions for the graphical representation of (probabilistic) timed automata
(the model has one clock denoted by $\aclock$, and black boxes denote probabilistic choices over outgoing edges).
In location $\mathrm{W}$, the system is \emph{working} on a task,
which is completed after between $1$ and $3$ units of time.
When the task is completed, it is either \emph{successful} (edge to location $\mathrm{S}$),
\emph{fails} (edge to location $\mathrm{F}$)
or leads to system \emph{termination} (edge to location $\mathrm{T}$).
For the case in which the task completion fails,
between $4$ and $5$ time units after work on the task started
the system may either restart the task from the beginning (edge to location $\mathrm{W}$, resetting $\aclock$ to $0$),
or terminate (edge to location $\mathrm{T}$).
The edges corresponding to probabilistic choices are labelled with \emph{expressions} over the clock $\aclock$,
which describe how the probability of those edges changes in accordance with changes in the value of $\aclock$.
For example, the longer the time spent in location $\mathrm{W}$,
the higher the value of $\aclock$ when location $\mathrm{W}$ is left,
and the higher the probability of making a transition to location $\mathrm{S}$,
which corresponds to the successful completion of the task.

Previous work on clock-dependent probabilistic timed automata showed that a basic
\emph{quantitative} reachability problem,
regarding whether there is a scheduler of nondeterministic choice such that the probability of reaching a set of target locations
exceeds some probability threshold,
is undecidable,
but that an approach based on the region graph
(which is a finite-state abstraction used extensively for timed automata)
can be employed to approximate optimal
reachability probabilities~\cite{Spr21}.
The undecidability result relied on the presence of at least three clocks:
in this paper, following similar precedents in the context of (non-probabilistic and probabilistic) variants of timed automata
(for example,~\cite{LMS04,BLM08,JLS08,BBBMBGJ14,BBG14,ABKMT16}),
we restrict our attention to clock-dependent probabilistic timed automata with a \emph{single} clock variable.
As in~\cite{Spr21}, we consider the case in which the dependencies of transition probabilities
on the value of the clock
are described by affine functions.
Furthermore, we assume that,
between any two edges with a non-constant dependence on the clock,
the clock must have a natural-numbered value,
either through being reset to $0$ or by increasing as time passes.
We call this condition \emph{initialisation},
following the precedents of~\cite{ABKMT16} and~\cite{HKPV98},
in which similar conditions are used to obtain decidability results for stochastic timed systems with one clock,
and hybrid automata, respectively;
intuitively, the value of the clock is ``reinitialised''
(either explicitly, through a reset to $0$, or implicitly, through the passage of time)
to a known, natural value between
non-constant dependencies of probability on the value of the clock.
Note that the clock-dependent probabilistic timed automaton of \figref{fig:one_clock}
satisfies this assumption
(although clock $\aclock$ is not reset on the edge to location $\mathrm{F}$, it must take values $3$ and $4$,
i.e., at least one natural-numbered value,
before location $\mathrm{F}$ can be left).
We show that, for such clock-dependent probabilistic timed automata,
quantitative reachability problems can be solved in polynomial time.
Similarly, we can also solve in polynomial time \emph{qualitative} reachability problems,
which ask whether there exists a scheduler of nondeterminism
such that a set of target locations can be reached with probability $1$ (or $0$),
or whether all schedulers of nondeterminism result in the target locations being reached with probability $1$ (or $0$).

These results rely on the construction of an \emph{interval Markov decision process} from
the one-clock clock-dependent probabilistic timed automaton.
Interval Markov decision processes have been well-studied in the verification context
(for example, in~\cite{PLSS13,HHK14,HM18}),
and also in other contexts, such as planning~\cite{GLD00} and control~\cite{NG05,WK08}.
They comprise a finite state space where transitions between states are achieved in the following manner:
for each state, there is a nondeterministic choice between a set of actions,
where each action is associated with a decoration of the set of edges from the state with \emph{intervals} in $[0,1]$;
then a nondeterministic choice as to the exact probabilities associated with each outgoing edge is chosen from
the intervals associated with the action chosen in the first step;
finally, a probabilistic choice is made over the edges according to the probabilities chosen in the second step,
thus determining the next state.
In contrast to the standard formulation of interval Markov decision processes,
we allow edges corresponding to probabilistic choices to be labelled not only with closed intervals,
but also with open and half-open intervals.
While (half-)open intervals have been considered previously in the context of interval Markov chains in~\cite{CK15,Spr18},
we are unaware of any work considering them in the context of interval Markov decision processes.
%
Open intervals in the constructed interval Markov decision process
provide a natural way of representing strict constraints on clocks
in the one-clock clock-dependent probabilistic timed automaton.

We proceed by giving some preliminary concepts in Section~\ref{sec:oimdps}:
this includes a reduction from interval Markov decision processes
to interval Markov chains~\cite{JL91,KU02,SVA06}
with the standard Markov decision process-based semantics,
which may be of independent interest.
The reduction takes open and half-open intervals into account;
while~\cite{CK15} has shown that open interval Markov chains can be reduced to closed Markov chains for the purposes
of quantitative properties,~\cite{Spr18} shows that the open/closed distinction is critical for the evaluation of qualitative properties.
In Section~\ref{sec:onec-cdpta},
we present the definition of one-clock clock-dependent probabilistic timed automata,
and describe the transformation to interval Markov decision processes in Section~\ref{sec:general}.
This paper extends the conference version~\cite{Spr20}
with full proofs of the results.

\section{Interval Markov Decision Processes}\label{sec:oimdps}

In this section, we focus on interval Markov decision processes,
a formalism that we will subsequently use to construct exact finite-state abstractions
from one-clock clock-dependent probabilistic timed automata.
First we present a number of preliminary concepts,
before proceeding to define interval Markov chains and interval Markov decision processes.
At the end of the section, we present a novel result,
namely a translation from interval Markov decision processes
to interval Markov chains,
and prove its correctness with respect to quantitative and qualitative reachability problems.

\subsection{Preliminaries.}
We use $\nnr$ to denote the set of non-negative real numbers,
$\Qset$ to denote the set of rational numbers, and
$\Nset$ to denote the set of natural numbers.
%
%
%
A (discrete) probability \emph{distribution} over a countable set $Q$ is
a function ${\adist: Q \ra [0,1]}$ such that $\sum_{q \in Q} \adist(q) = 1$.
Let $\dist(Q)$ be the set of distributions over $Q$.
For a (possibly uncountable) set $Q$ and a function $\adist : Q \ra [0,1]$,
we define $\support(\adist) = \set{q \in Q \setsep \adist(q) > 0}$.
Then, for an uncountable set $Q$, we define $\dist(Q)$ to be the set of
functions ${\adist : Q \ra [0,1]}$ such that $\support(\adist)$ is a
countable set and $\adist$ restricted to $\support(\adist)$ is a distribution.
Given a binary function $f: Q \times Q \ra [0,1]$ and element $q \in Q$,
we denote by $f(q,\cdot): Q \ra [0,1]$ the unary function
such that $f(q,\cdot)(q')=f(q,q')$ for each $q' \in Q$.
%


A \emph{Markov chain}
(MC) $\adtmc$ is a pair $(\dstates,\dtrans)$
where $\dstates$ is a set of \emph{states}
and $\dtrans: \dstates \times \dstates \ra [0,1]$ is a \emph{transition probability function},
such that $\dtrans(\astate,\cdot) \in \dist(\dstates)$ for each state $\astate \in \dstates$.
A \emph{path} of MC $\adtmc$ is a sequence $\astate_0 \astate_1 \cdots$ of states
such that $\dtrans(\astate_i,\astate_{i+1})>0$ for all $i \geq 0$.
Given a path $\apath = \astate_0 \astate_1 \cdots$ and $i \geq 0$,
we let $\apath(i) = \astate_i$ be the $(i+1)$-th state along $\apath$.
The set of paths of $\adtmc$ starting in state $\astate \in \dstates$ is denoted by $\dpaths{\adtmc}(\astate)$.
In the standard manner (see, for example,~\cite{BK08,FKNP11}),
given a state $\astate \in \dstates$,
we can define a probability measure $\dprob{\adtmc}{\astate}$ over $\dpaths{\adtmc}(\astate)$.

A \emph{Markov decision process} (MDP)
$\anMDP = ( \mstates, \mactions, \mtrans )$ comprises
a set $\mstates$ of \emph{states}, 
a set $\mactions$ of \emph{actions},
and a \emph{probabilistic transition function}
$\mtrans: \mstates \times \mactions \ra \dist(\mstates) \cup \set{ \noaction }$.
The symbol $\noaction$ is used to represent the unavailability of an action in a state,
i.e., $\mtrans(\astate,\anaction) = \noaction$ signifies that
action $\anaction \in \mactions$ is not available in state $\astate \in \mstates$.
For each state $\astate \in \mstates$,
let $\mactionsof{\astate} = \set{ \anaction \in \mactions \setsep \mtrans(\astate,\anaction) \neq \noaction }$,
and assume that $\mactionsof{\astate} \neq \emptyset$,
i.e., there is at least one available action in each state.
Transitions from state to state of an MDP are performed in two
steps: if the current state is $\astate$, the first step concerns a
nondeterministic selection of an action $\anaction \in \mactionsof{\astate}$;
the second step comprises a probabilistic choice, made according to the distribution $\mtrans(\astate,\anaction)$,
as to which state to make the transition
(that is, a transition to a state $\astate' \in \mstates$ is made with probability $\mtrans(\astate,\anaction)(\astate')$).
In general, the sets of states and actions
can be uncountable.
We say that an MDP is \emph{finite} if $\mstates$ and $\mactions$ are finite sets.

A(n infinite) path of an MDP $\anmdp$ is a sequence $\astate_0 \anaction_0 \astate_1 \anaction_1 \cdots$
such that $\anaction_i \in \mactionsof{\astate_i}$ and $\mtrans(\astate_i,\anaction_i)(\astate_{i+1})>0$
for all $i \geq 0$.
Given an infinite path $\apath = \astate_0 \anaction_0 \astate_1 \anaction_1 \cdots$ and $i \geq 0$,
we let $\apath(i) = \astate_i$ be the $(i+1)$-th state along $\apath$.
Let $\mpaths{\anmdp}$ be the set of infinite paths of $\anmdp$.
A finite path is a sequence $\afinpath = \astate_0 \anaction_0 \astate_1 \anaction_1 \cdots \anaction_{n-1} \astate_n$
such that $\anaction_i \in \mactionsof{\astate_i}$ and $\mtrans(\astate_i,\anaction_i)(\astate_{i+1})>0$
for all $0 \leq i < n$.
Let $\last(\afinpath) = \astate_n$ denote the final state of $\afinpath$. 
For $\anaction \in \mactionsof{\astate_n}$ and $\astate \in \mstates$
such that $\mtrans(\astate_n,\anaction)(\astate)>0$,
we use $\afinpath \anaction \astate$ to denote the finite path
$\astate_0 \anaction_0 \astate_1 \anaction_1 \cdots \anaction_{n-1} \astate_n  \anaction \astate$.
Let $\mfinpaths{\anmdp}$ be the set of finite paths of the MDP $\anmdp$.
Let $\mpaths{\anmdp}(\astate)$ and $\mfinpaths{\anmdp}(\astate)$
be the sets of infinite paths and finite paths, respectively,
of $\anmdp$ starting in state $\astate \in \mstates$.

A \emph{scheduler} is a function
$\astrat: \mfinpaths{\anmdp} \ra \bigcup_{\astate \in \mstates} \dist(\mactionsof{\astate})$
such that $\astrat(\afinpath) \in \dist(\mactionsof{\last(\afinpath)})$
for all $\afinpath \in \mfinpaths{\anmdp}$.\footnote{From~\cite[Lemma~4.10]{Hah13},
without loss of generality we can assume henceforth that schedulers
map to distributions assigning positive probability to \emph{finite} sets of actions,
i.e., schedulers $\astrat$ for which $|\support(\astrat(\afinrun))|$ is finite for all $\afinrun \in \mfinpaths{\anmdp}$.}
Let $\strats{\anmdp}$ be the set of schedulers of the MDP $\anmdp$.
We say that infinite path
$\apath = \astate_0 \anaction_0 \astate_1 \anaction_1 \cdots$ is \emph{generated} by $\astrat$
if $\astrat(\astate_0 \anaction_0 \astate_1 \anaction_1 \cdots \anaction_{i-1} \astate_i)(\anaction_i)>0$
for all $i \in \Nset$.
Let $\mpaths{\astrat}$ be the set of paths generated by $\astrat$.
The set $\mfinpaths{\astrat}$ of finite paths generated by $\astrat$ is defined similarly.
Let $\mpaths{\astrat}(\astate) = \mpaths{\astrat} \cap \mpaths{\anmdp}(\astate)$
and $\mfinpaths{\astrat}(\astate) = \mfinpaths{\astrat} \cap \mfinpaths{\anmdp}(\astate)$.
Given a scheduler $\astrat \in \strats{\anmdp}$,
we can define a countably infinite-state MC $\adtmcparaone{\astrat}$
that corresponds to the behaviour of $\astrat$:
we let $\adtmcparaone{\astrat} = (\mfinpaths{\astrat},\dtrans)$,
where, for $\afinpath, \afinpath' \in \mfinpaths{\astrat}$,
we have $\dtrans(\afinpath, \afinpath') = \astrat(\afinpath)(\anaction) \cdot \mtrans(\last(\afinpath),\anaction)(\astate)$
if $\afinpath' = \afinpath \anaction \astate$ and $\anaction \in \mactionsof{\last(\afinpath)}$,
and $\dtrans(\afinpath, \afinpath') = 0$ otherwise.
For $\afinpath = \astate_0 \anaction_0 \astate_1 \anaction_1 \cdots \anaction_{n-1} \astate_n$,
we denote the $(i+1)$-th prefix of $\afinpath$ by $\afinpathpref{i}$,
i.e., $\afinpathpref{i} = \astate_0 \anaction_0 \astate_1 \anaction_1 \cdots \anaction_{i-1} \astate_i$,
for $i \leq n$
(note that $\afinpathpref{0} = \astate_0$).
Given $\astate \in \mstates$ and $\afinpath \in \mfinpaths{\astrat}$,
we let $\probfin{\asched}{\astate}(\afinpath)
= \dtrans(\afinpathpref{0},\afinpathpref{1}) \cdot \dots \cdot \dtrans(\afinpathpref{n-1},\afinpathpref{n})$ 
if $\astate = \afinpathpref{0}$,
and let $\probfin{\asched}{\astate}(\afinpath) = 0$ otherwise.
Let $\cyl(\afinpath) \subseteq \mpaths{\anmdp}(\astate)$ be the set of infinite paths starting in $\astate$
that have the finite path $\afinpath$ as a prefix.
Then we let $\mprob{\astrat}{\astate}$ be the unique probability measure over $\mpaths{\astrat}(\astate)$
such that $\mprob{\astrat}{\astate}(\cyl(\afinpath)) = \probfin{\asched}{\astate}(\afinpath)$
(for more details, see~\cite{BK08,FKNP11}).

Given a set $\targetset \subseteq \mstates$,
we define $\Diamond \targetset
= \{ \apath \in \mpaths{\anmdp} \setsep \exists i \in \Nset \ldot \apath(i) \in \targetset  \}$
as the set of infinite paths of $\anmdp$ such that some state of $\targetset$ is visited along the path.
Let $\astate \in \mstates$.
We define the
\emph{maximum probability of reaching $\targetset$ from $\astate$} as
$\maxval{\anmdp}{\targetset}{\astate}
=
\Sup{\astrat \in \strats{\anmdp}}~ \probl{\astrat}{\astate}(\Diamond \targetset)$.
Similarly,
the \emph{minimum probability of reaching $\targetset$ from $\astate$}
is defined as
$\minval{\anmdp}{\targetset}{\astate}
=
\Inf{\astrat \in \strats{\anmdp}}~ \probl{\astrat}{\astate}(\Diamond \targetset)$.
The \emph{maximal reachability problem}
for $\anmdp$, $\targetset \subseteq \mstates$, $\astate \in \mstates$,
$\unrhd \in \{ \geq, > \}$ and $\pthresh \in [0,1]$
is to decide whether $\maxval{\anmdp}{\targetset}{\astate} \unrhd \pthresh$.
Similarly, the \emph{minimal reachability problem}
for $\anmdp$, $\targetset \subseteq \ptsstates$, $\astate \in \mstates$,
$\unlhd \in \{ \leq, < \}$ and $\pthresh \in [0,1]$
is to decide whether $\minval{\anmdp}{\targetset}{\astate} \unlhd \pthresh$.
The maximal and minimal reachability problems are called \emph{quantitative} problems.
We also consider the following \emph{qualitative} problems:
($\forall 0$)
decide whether $\probl{\astrat}{\astate}(\Diamond \targetset)=0$ for all $\astrat \in \strats{\anmdp}$;
($\exists 0$)
decide whether there exists $\astrat \in \strats{\anmdp}$
such that $\probl{\astrat}{\astate}(\Diamond \targetset)=0$;
($\exists 1$)
decide whether there exists $\astrat \in \strats{\anmdp}$
such that $\probl{\astrat}{\astate}(\Diamond \targetset)=1$;
($\forall 1$)
decide whether $\probl{\astrat}{\astate}(\Diamond \targetset)=1$ for all $\astrat \in \strats{\anmdp}$.


\subsection{Interval Markov Chains.}
We let $\intervals$ denote the set of (open, half-open or closed) intervals
that are subsets of $[0,1]$
and that have rational-numbered endpoints.
Given an interval $\anint \in \intervals$,
we let $\lep{\anint}$ (respectively, $\rep{\anint}$) be the left (respectively, right) endpoint of $\anint$.
%

An \emph{interval distribution} over a finite set $Q$ is a function $\anintdist: Q \ra \intervals$
such that:
\begin{description}
\item[(1)]
$\sum_{q \in Q} \lep{\anintdist(q)} \leq 1
\leq \sum_{q \in Q} \rep{\anintdist(q)}$,
\item[(2a)]
$\sum_{q \in Q} \lep{\anintdist(q)} = 1$ implies that
$\anintdist(q)$ is left-closed for all $q \in Q$,
and
\item[(2b)]
$\sum_{q \in Q} \rep{\anintdist(q)} = 1$ implies that
$\anintdist(q)$ is right-closed for all $q \in Q$.
\end{description}
We define $\intdist(Q)$ as the set of interval distributions over $Q$.
An \emph{assignment for interval distribution $\anintdist$}
is a distribution $\ana \in \dist(Q)$ such that
$\ana(q) \in \anintdist(q)$ for each $q \in Q$.
Note that conditions (1), (2a) and (2b) in the definition of interval distributions
guarantee that there exists at least one assignment for each interval distribution.
Let $\ass{\anintdist}$ be the set of assignments for $\anintdist$.

An (open) \emph{interval Markov chain} (IMC) $\anoimc$ is a pair $(\ocstates,\octrans)$,
where $\ocstates$ is a finite set of \emph{states},
and $\octrans: \ocstates \times \ocstates \ra \intervals$
is a \emph{interval-based transition function}
such that $\octrans(\astate,\cdot)$ is an interval distribution for each $\astate \in \ocstates$
(formally, $\octrans(\astate,\cdot) \in \intdist(\ocstates)$).
An IMC makes a transition from a state $\astate \in \ocstates$ in two steps:
first an assignment $\ana$ 
is chosen from the set $\ass{\octrans(\astate,\cdot)}$ of assignments for $\octrans(\astate,\cdot)$,
then a probabilistic choice over target states is made according to $\ana$.
Let $\ass{\octrans} = \bigcup_{\astate \in \ocstates} \ass{\octrans(\astate,\cdot)}$
be the set of all assignments with respect to $\octrans$.
%
%
%
The semantics of an IMC corresponds to an MDP that has the same state space as the IMC
and the action set $\ass{\octrans}$,
and where selecting assignment $\ana \in \ass{\octrans(\astate,\cdot)}$ from state $\astate$
means that the target state of the transition is chosen probabilistically according to $\ana$.
%
%
%
Formally, the semantics of an IMC $\anoimc = (\ocstates,\octrans)$
is the MDP $\sem{\anoimc} = (\ocstates,\ass{\octrans},\mtrans)$,
where
(1)~$\mtrans(\astate,\ana) = \ana$ for all $\astate \in \ocstates$
and $\ana \in \ass{\octrans(\astate,\cdot)}$,
and (2) $\mtrans(\astate,\ana) = \noaction$ for all $\astate \in \ocstates$
and $\ana \in \ass{\octrans} \setminus \ass{\octrans(\astate,\cdot)}$.
In previous literature (for example,~\cite{SVA06,CSH08,CHK13}),
this semantics is called the ``IMDP semantics''.

Computing $\maxval{\sem{\anoimc}}{\targetset}{\astate}$
and $\minval{\sem{\anoimc}}{\targetset}{\astate}$
can be done for an IMC $\anoimc$ simply by transforming the IMC by closing all of its (half-)open intervals,
then employing a standard maximum/minimum reachability probability computation on the new, ``closed'' IMC
(for example, the algorithms of~\cite{SVA06,CHK13}):
the correctness of this approach is shown in~\cite{CK15}.
Algorithms for qualitative problems of IMCs (with open, half-open and closed intervals)
are given in~\cite{Spr18}.
All of the aforementioned algorithms run in polynomial time in the size of the IMC,
which is obtained as the sum over all states $\astate,\astate' \in \ocstates$
of the binary representation of the endpoints of $\octrans(\astate,\astate')$,
where rational numbers are encoded as the quotient of integers written in binary.

\subsection{Interval Markov Decision Processes.}
An (open) \emph{interval Markov decision process} (IMDP) $\anoimdp = (\omstates,\omactions,\omtrans)$
comprises a finite set $\omstates$ of states,
a finite set $\omactions$ of actions,
and an \emph{interval-based transition function}
$\omtrans: \omstates \times \omactions \ra \intdist(\omstates) \cup \set{ \noaction }$.
Let $\omactionsof{\astate} =
\set{ \anomact \in \omactions \setsep \omtrans(\astate,\anomact) \neq \noaction }$,
and assume that $\omactionsof{\astate} \neq \emptyset$ for each state $\astate \in \omstates$.
In contrast to IMCs,
an IMDP makes a transition from a state $\astate \in \omstates$ in \emph{three} steps:
(1)~an action $\anomact \in \omactionsof{\astate}$ is chosen,
then (2)~an assignment $\ana$ for 
$\omtrans(\astate,\anomact)$ is chosen,
and finally (3)~a probabilistic choice over target states to make the transition to is performed
according to $\ana$.
%
%
%
%
Formally, the semantics of an IMDP $\anoimdp = (\omstates,\omactions,\omtrans)$
is the MDP $\sem{\anoimdp} = (\omstates,\mactions,\mtrans)$
where $\mactions = \omactions \times \dist(\omstates)$,
and where
(1)~$\mtrans(\astate,(\anomact,\ana)) = \ana$
for all $\astate \in \omstates$, $\anomact \in \omactionsof{\astate}$
and $\ana \in \ass{\omtrans(\astate,\anomact)}$,
and
(2)~$\mtrans(\astate,(\anomact,\ana)) = \noaction$
for all $\astate \in \omstates$
for which either (a)~$\anomact \not\in \omactionsof{\astate}$
or (b)~$\anomact \in \omactionsof{\astate}$ and $\ana \not\in \ass{\omtrans(\astate,\anomact)}$.
%
%
%
%
%
Note that
we adopt a \emph{cooperative} resolution of nondeterminism for IMDPs
(as in, for example,~\cite{PLSS13,HHK14,HM18}),
in which the choice of action and assignment (steps~(1) and (2) above)
is combined into a \emph{single} nondeterministic choice in the semantic MDP\@.

Given the cooperative nondeterminism for IMDPs,
we can show that, given an IMDP, an IMC can be constructed in polynomial time
such that the maximal and minimal reachability probabilities for the IMDP and the constructed IMC coincide,
and furthermore qualitative properties agree on the IMDP and the constructed IMC\@.
Formally, given the IMDP $\anoimdp = (\omstates,\omactions,\omtrans)$,
we construct an IMC $\oimcfrom{\anoimdp} = (\cfstates,\cftrans)$ in the following way:
\begin{itemize}
\item
the set of states is defined as $\cfstates =\omstates \cup (\omstacts)$,
where $\omstacts =
\bigcup_{\astate \in \omstates}
\set{ (\astate,\anomact) \in \omstates \times \omactions \setsep \anomact \in \omactionsof{\astate} }$;
\item
for $\astate \in \omstates$ and $\anomact \in \omactionsof{\astate}$,
let $\cftrans(\astate,(\astate,\anomact)) = [0,1]$,
and let $\cftrans((\astate,\anomact),\cdot) = \omtrans(\astate,\anomact)$.
\end{itemize}
%


\begin{figure}[t]
\centering
%

\centering
\tiny

\begin{tikzpicture}[->,>=stealth',shorten >=1pt,auto, thin] 
  \tikzstyle{state}=[draw=black, text=black, shape=circle, inner sep=3pt, outer sep=0pt, circle, rounded corners] 
  \tikzstyle{big_state}=[draw=black, text=black, shape=rectangle, inner sep=3pt, outer sep=0pt, rectangle, rounded corners] 
  \tikzstyle{final_state}=[draw=black, text=black, shape=circle, inner sep=3pt, outer sep=0pt, circle, rounded corners] 



	\node[state](sl){$\astate$};

	\node[node distance=1.5cm] (inv_mid)[right of=sl]{};

	\node[fill=black, node distance=1.5cm] (nail_top) [above of=inv_mid] {};
	\node[xshift=7mm,left of = nail_top](a_nail_top){$\anomact_1$};

	\node[fill=black, node distance=1.5cm] (nail_bottom) [below of=inv_mid] {};
	\node[xshift=7mm,left of = nail_bottom](a_nail_bottom){$\anomact_2$};

	\node[state, node distance=1.5cm](st)[right of=nail_top]{$\astate_1$};
	\node[state, node distance=1.5cm](sm)[right of=inv_mid]{$\astate_2$};
	\node[state, node distance=1.5cm](sb)[right of=nail_bottom]{$\astate_3$};


	\node[state, node distance=3cm](slc)[right of=sm]{$\astate$};

	\node[node distance=1.5cm] (inv_midc)[right of=slc]{};

	\node[big_state, node distance=1.5cm] (state_top) [above of=inv_midc] {$(\astate,\anomact_1)$};
	\node[big_state, node distance=1.5cm] (state_bottom) [below of=inv_midc] {$(\astate,\anomact_2)$};

	\node[state, node distance=1.5cm](stc)[right of=state_top]{$\astate_1$};
	\node[state, node distance=1.5cm](smc)[right of=inv_midc]{$\astate_2$};
	\node[state, node distance=1.5cm](sbc)[right of=state_bottom]{$\astate_3$};


	\path[rounded corners]

	(sl)
		edge (nail_top)
		edge (nail_bottom)

	(nail_top)
		edge [above] node {$(\frac{1}{4},\frac{2}{3}]$} (st)
		edge [right] node {$[\frac{1}{3},\frac{3}{4})$} (sm)

	(nail_bottom)
		edge [right] node {$[\frac{1}{2},\frac{1}{2}]$} (sm)
		edge [below] node {$[\frac{1}{2},\frac{1}{2}]$} (sb)


	(slc)
		edge [left] node {$[0,1]$} (state_top)
		edge [left] node {$[0,1]$} (state_bottom)

	(state_top)
		edge [above] node {$(\frac{1}{4},\frac{2}{3}]$} (stc)
		edge [right] node {$[\frac{1}{3},\frac{3}{4})$} (smc)

	(state_bottom)
		edge [right] node {$[\frac{1}{2},\frac{1}{2}]$} (smc)
		edge [below] node {$[\frac{1}{2},\frac{1}{2}]$} (sbc)
;

\end{tikzpicture}


\caption{Example fragments of IMDP $\anoimdp$ (left) and the constructed IMC $\oimcfrom{\anoimdp}$ (right).}%
\label{fig:IMDP_IMC_example}

\end{figure}


\begin{exa}
In \figref{fig:IMDP_IMC_example}
we illustrate a fragment of an IMDP (left),
and give the corresponding fragment of the constructed IMC $\oimcfrom{\anoimdp}$ (right).
Note that the nondeterministic choice
between actions $\anomact_1$ and $\anomact_2$
from state $\astate$ of $\anoimdp$
is replaced in $\oimcfrom{\anoimdp}$ by a choice between states $(\astate,\anomact_1)$ and $(\astate,\anomact_2)$;
more precisely, the choice of a scheduler of $\anoimdp$ from a path ending in state $\astate$
as to probabilities to assign to actions $\anomact_1$ and $\anomact_2$
is reflected by the choice of assignment made in $\oimcfrom{\anoimdp}$
over states $(\astate,\anomact_1)$ and $(\astate,\anomact_2)$
(note that this latter choice is unconstrained by the interval $[0,1]$ used for these transitions).
\end{exa}

The following proposition states the correctness of the construction of $\oimcfrom{\anoimdp}$
with respect to quantitative and qualitative problems.
%
\begin{prop}\label{prop:imdp2imc_props}
Let $\anoimdp = (\omstates,\omactions,\omtrans)$ be an IMDP,
and let $\astate \in \omstates$, $\targetset \subseteq \omstates$ and $\lambda \in \set{ 0,1 }$.
Then:
\begin{itemize}
\item
$\maxval{\sem{\anoimdp}}{\targetset}{\astate}
= \maxval{\sem{\oimcfrom{\anoimdp}}}{\targetset}{\astate}$
and
$\minval{\sem{\anoimdp}}{\targetset}{\astate}
= \minval{\sem{\oimcfrom{\anoimdp}}}{\targetset}{\astate}$;
\item
there exists $\astrat \in \strats{\sem{\anoimdp}}$
such that $\probl{\astrat}{\astate}(\Diamond \targetset)=\lambda$
if and only if there exists $\astrat' \in \strats{\sem{\oimcfrom{\anoimdp}}}$
such that $\probl{\astrat'}{\astate}(\Diamond \targetset)=\lambda$; 
\item
$\probl{\astrat}{\astate}(\Diamond \targetset)=\lambda$
for all $\astrat \in \strats{\sem{\anoimdp}}$
if and only if
$\probl{\astrat'}{\astate}(\Diamond \targetset)=\lambda$
for all $\astrat' \in \strats{\sem{\oimcfrom{\anoimdp}}}$.
\end{itemize}
\end{prop}

\noindent
In order to show \propref{prop:imdp2imc_props},
we consider two lemmata
that show that a scheduler of $\anoimdp$
can be matched by a ``mimicking'' scheduler of $\oimcfrom{\anoimdp}$ (\lemref{lem:imdp2imc}),
and vice versa (\lemref{lem:imc2imdp}),
such that the probabilities of reaching a set of target states for the two schedulers coincide.
Together, these lemmata suffice to establish \propref{prop:imdp2imc_props}.

\begin{lem}\label{lem:imdp2imc}
Let $\anoimdp = (\omstates,\omactions,\omtrans)$ be an IMDP,
$\astate \in \omstates$ and $\targetset \subseteq \omstates$.
Then, for each $\asched \in \strats{\sem{\anoimdp}}$,
there exists $\ansched \in \strats{\sem{\oimcfrom{\anoimdp}}}$
such that
$\probl{\asched}{\astate}(\Diamond \targetset) = \probl{\ansched}{\astate}(\Diamond \targetset)$.
\end{lem}
\begin{proof}
The proof of the lemma consists of constructing $\ansched \in \strats{\sem{\oimcfrom{\anoimdp}}}$
based on $\asched$,
and showing that
$\probl{\asched}{\astate}(\Diamond \targetset) = \probl{\ansched}{\astate}(\Diamond \targetset)$.
Intuitively, the construction of scheduler $\ansched$
proceeds in the following manner:
each choice made by $\asched$ (which we recall, from the definition of $\sem{\anoimdp}$,
is a distribution over $\anoimdp$-action/assignment pairs)
is mimicked by a sequence of \emph{two} choices made by $\ansched$,
the first of which mimics the choice of $\asched$ over $\anoimdp$-actions,
and where the second mimics the choice of $\asched$ over assignments.
The proof consists of three parts:
first, we recall some useful facts and introduce some notation;
second, we present the construction of $\ansched$;
finally, we show that
$\probl{\asched}{\astate}(\Diamond \targetset) = \probl{\ansched}{\astate}(\Diamond \targetset)$.

\sloppypar{
Let $\anoimdp = (\omstates,\omactions,\omtrans)$.
To avoid ambiguity, we henceforth refer to elements of $\omactions$ as $\anoimdp$-actions.
Recall that
$\sem{\anoimdp} = (\omstates,\mactions,\mtrans_\sem{\anoimdp})$
where
\[\mactionsof{\astate} =
\set{ (\anomact,\ana) \in \omactions \times \dist(\omstates) \setsep \anomact \in \omactionsof{\astate} \mbox{ and } \ana \in \ass{\omtrans(\astate,\anomact)} }\]
and $\mtrans_\sem{\anoimdp}(\astate,(\anomact,\ana)) = \ana$
for each $\astate \in \omstates$ and $(\anomact,\ana) \in \mactionsof{\astate}$.
Note that an action $(\anomact,\ana) \in \mactions$ of $\sem{\anoimdp}$
comprises an $\omactions$-action $\anomact$
and an assignment $\ana$ from $\dist(\omstates)$.
We can observe that finite paths of $\sem{\anoimdp}$ (i.e., elements of $\mfinpaths{\sem{\anoimdp}}$)
have the form
$\astate_0 (\anomact_0,\ana_0) \astate_1 (\anomact_1,\ana_1) \cdots (\anomact_{n-1},\ana_{n-1}) \astate_n$,
where $\astate_i \in \omstates$ for all $0 \leq i \leq n$,
and $(\anomact_j,\ana_j) \in \omactions \times \dist(\omstates)$ for all $0 \leq i < n$.
}

Recall that $\oimcfrom{\anoimdp} = (\omstates \cup (\omstacts),\cftrans)$,
where $\cftrans(\astate,(\astate,\anomact)) = [0,1]$
and $\cftrans((\astate,\anomact),\cdot) = \omtrans(\astate,\anomact)$
for $\astate \in \omstates$ and $\anomact \in \omactionsof{\astate}$,
i.e., $\oimcfrom{\anoimdp}$ alternates between states from $\omstates$ and $\omstacts$,
where the probability of transitions from $\omstates$ to $\omstacts$ are unconstrained,
and the probability of transitions from $\omstacts$ to $\omstates$ are constrained by
intervals defined by $\omtrans$.
Finally, we also recall that
$\sem{\oimcfrom{\anoimdp}} = (\omstates \cup (\omstacts),\ass{\cftrans},\mtrans_\sem{\oimcfrom{\anoimdp}})$,
where $\ass{\cftrans} = \bigcup_{\tilde{\astate} \in \omstates \cup (\omstacts)} \ass{\cftrans(\tilde{\astate},\cdot)}$
and $\mtrans_\sem{\oimcfrom{\anoimdp}}(\tilde{\astate},\ana) = \ana$ for all $\tilde{\astate} \in \omstates \cup (\omstacts)$
and $\ana \in \ass{\cftrans(\tilde{\astate},\cdot)}$.
Note that the actions of $\sem{\oimcfrom{\anoimdp}}$, i.e., elements of the set $\ass{\cftrans}$,
are themselves the assignments
that are used to determine the next state:
given the aforementioned alternation between $\omstates$ and $\omstacts$ in $\oimcfrom{\anoimdp}$,
from states in $\omstates$, only assignments in $\dist(\omstacts)$ are available;
similarly, from states in $\omstacts$, only assignments in $\dist(\omstates)$ are available.
We partition $\mfinpaths{\sem{\oimcfrom{\anoimdp}}}$
into two sets, on the basis of whether the final state of a path is in $\omstates$ or $\omstacts$:
let $\mfinpathsendsin{\sem{\oimcfrom{\anoimdp}}}{\omstates}
\subseteq \mfinpaths{\sem{\oimcfrom{\anoimdp}}}$
be the set of finite paths 
of the form
$\astate_0 \anab_0 (\astate_0,\anomact_0) \anac_0
\astate_1 \anab_1 (\astate_1,\anomact_1) \anac_1  \cdots
\astate_{n-1} \anab_{n-1} (\astate_{n-1},\anomact_{n-1}) \anac_{n-1}
\astate_n$,
and let $\mfinpathsendsin{\sem{\oimcfrom{\anoimdp}}}{\omstacts}
\subseteq \mfinpaths{\sem{\oimcfrom{\anoimdp}}}$
be the set of finite paths 
of the form
$\astate_0 \anab_0 (\astate_0,\anomact_0) \anac_0
\astate_1 \anab_1 (\astate_1,\anomact_1) \anac_1  \cdots
\astate_{m} \anab_{m} (\astate_{m},\anomact_{m})$.
Note that, in the context of paths,
notation such as $\astate_i \anab_i (\astate_i,\anomact_i)$
refers to a transition from state $\astate_i$ to state $(\astate_i,\anomact_i)$ with assignment $\anab_i$
(recall that the actions of $\sem{\oimcfrom{\anoimdp}}$ are assignments).

Let $\astate \in \omstates$, $\targetset \subseteq \omstates$ and $\asched \in \strats{\sem{\anoimdp}}$.
We now proceed to construct $\ansched \in \strats{\sem{\oimcfrom{\anoimdp}}}$
based on $\asched$.
To describe formally the construction of $\ansched$,
we first introduce the function
$\cmTOm: \mfinpathsendsin{\sem{\oimcfrom{\anoimdp}}}{\omstates}(\astate) \ra \mfinpaths{\sem{\anoimdp}}(\astate)$
that associates, for each finite path in $\mfinpathsendsin{\sem{\oimcfrom{\anoimdp}}}{\omstates}(\astate)$,
a finite path in $\mfinpaths{\sem{\anoimdp}}(\astate)$
that (1)~visits the same states of $\omstates$,
and (2)~features the same actions from $\omactions$ and assignments from $\dist(\omstates)$.
Formally, for $\afincmpath \in \mfinpathsendsin{\sem{\oimcfrom{\anoimdp}}}{\omstates}(\astate)$
such that $\afincmpath =
\astate_0 \anab_0 (\astate_0,\anomact_0) \anac_0
\astate_1 \anab_1 (\astate_1,\anomact_1) \anac_1  \cdots
\astate_{n-1} \anab_{n-1} (\astate_{n-1},\anomact_{n-1}) \anac_{n-1}
\astate_n$,
we let $\cmTOm(\afincmpath) =
\astate_0 (\anomact_0,\anac_0) \astate_1 (\anomact_1,\anac_1) \cdots (\anomact_{n-1},\anac_{n-1}) \astate_n$.
That is, $\cmTOm(\afincmpath)$ retains fully the subsequence $\astate_0 \astate_1 \cdots \astate_n$
of states from $\omstates$,
the subsequence $\anomact_0 \anomact_1 \cdots \anomact_{n-1}$ of actions of $\anoimdp$ from $\omactions$, and
the subsequence $\anac_0 \anac_1 \cdots \anac_{n-1}$ of assignments from $\dist(\omstates)$,
but does not retain $\anab_0 \anab_1 \cdots \anab_{n-1}$ of assignments from $\dist(\omstacts)$.
We note that $\cmTOm(\afincmpath)$ is an element of $\mfinpaths{\sem{\anoimdp}}(\astate)$ by the following reasoning:
for each $0 \leq i < n$,
we require that (a)~$\anomact_i \in \omactionsof{\astate_i}$
and (b)~$\anac_i \in \ass{\omtrans(\astate_i,\anomact_i)}$;
both (a) and (b) follow from the definition of $\oimcfrom{\anoimdp}$,
with (a) following because $(\astate_i,\anomact_i) \in \omstacts$ implies that $\anomact_i \in \omactionsof{\astate_i}$,
and (b) following because $\cftrans((\astate_i,\anomact_i),\cdot) = \omtrans(\astate_i,\anomact_i)$,
and hence $\anac_i \in \ass{\cftrans((\astate_i,\anomact_i),\cdot)}$
implies that $\anac_i \in \ass{\omtrans(\astate_i,\anomact_i)}$.
Consider finite path $\afincmpath \in \mfinpathsendsin{\sem{\oimcfrom{\anoimdp}}}{\omstates}(\astate)$.
We define the choice of $\ansched$ after $\afincmpath$
by mimicking the choice of $\asched$ after $\cmTOm(\afincmpath)$.
Recall that $\asched(\cmTOm(\afincmpath)) \in \dist(\mactions)$,
and that $\mactions \subseteq \omactions \times \dist(\omstates)$,
i.e., $\support(\asched(\cmTOm(\afincmpath)))$ contains pairs of the form $(\anomact,\ana)$
for $\anomact \in \mactions$ and $\ana \in \dist(\omstates)$.
We assume w.l.o.g.\ that, for $\anomact \in \omactionsof{\last(\afincmpath)}$,
there exists at most one $\ana \in \ass{\omtrans(\last(\afincmpath),\anomact)}$
such that $(\anomact,\ana) \in \support(\asched(\cmTOm(\afincmpath)))$.
The assumption can be made w.l.o.g.\ because
the set of assignments $\ass{\omtrans(\last(\afincmpath),\anomact)}$
is closed under convex combinations
(as noted in, for example,~\cite{HHK14}).
In the following, for $\anomact \in \omactionsof{\last(\afincmpath)}$
such that $(\anomact,\ana) \in \support(\asched(\cmTOm(\afincmpath)))$
for some $\ana \in \dist(\omstates)$,
we let $\pol{\afincmpath}{\anomact} = \ana$,
i.e., $\pol{\afincmpath}{\anomact}$ denotes the unique assignment
such that $(\anomact,\pol{\afincmpath}{\anomact}) \in \support(\asched(\cmTOm(\afincmpath)))$.
For each $\anomact \in \omactionsof{\last(\afincmpath)}$,
let $\anabof{\last(\afincmpath)}{\anomact}: \omstacts \ra \set{0,1}$ be the function
such that, for each $(\astate',\anomact') \in \omstacts$,
we have $\anabof{\last(\afincmpath)}{\anomact}(\astate',\anomact') = 1$
if $(\astate',\anomact') = (\last(\afincmpath),\anomact)$
and $\anabof{\last(\afincmpath)}{\anomact}(\astate',\anomact') = 0$ otherwise.
Note that $\anabof{\last(\afincmpath)}{\anomact} \in \ass{\cftrans(\last(\afincmpath),\cdot)}$,
i.e., $\anabof{\last(\afincmpath)}{\anomact}$ is an assignment for $\cftrans(\last(\afincmpath),\cdot)$
(recall that, for any $\anomact \in \omactionsof{\last(\afincmpath)}$,
we have $\cftrans(\astate,(\last(\afincmpath),\anomact)) = [0,1]$).
We then let $\ansched(\afincmpath)$ be such that
$\ansched(\afincmpath)(\anabof{\last(\afincmpath)}{\anomact})
= \asched(\cmTOm(\afincmpath))(\anomact,\pol{\afincmpath}{\anomact})$
for each $\anomact \in \omactionsof{\last(\afincmpath)}$ such that
$(\anomact,\pol{\afincmpath}{\anomact}) \in \support(\asched(\cmTOm(\afincmpath)))$.\footnote{We
note that the scheduler $\ansched(\afincmpath)$ can be defined in an alternative way, as follows.
Define $\anab: \omstacts \ra [0,1]$ to be the function
such that $\anab(\last(\afincmpath),\anomact) =
\asched(\cmTOm(\afincmpath))(\anomact,\pol{\afincmpath}{\anomact})$
for each $\anomact \in \omactionsof{\last(\afincmpath)}$ such that
$(\anomact,\pol{\afincmpath}{\anomact}) \in \support(\asched(\cmTOm(\afincmpath)))$,
and let $\ansched(\afincmpath)(\anab) = 1$,
i.e., the probabilistic choice of the distribution $\asched(\cmTOm(\afincmpath))$
is encoded in $\anab$, which is chosen by $\ansched(\afincmpath)$ with probability $1$.
We do not use this alternative definition of $\ansched$.}
Now consider finite path
$\afincmpath' \in \mfinpathsendsin{\sem{\oimcfrom{\anoimdp}}}{\omstacts}(\astate)$,
where $\afincmpath' = \afincmpath \anab (\last(\afincmpath),\anomact)$ for some
$\afincmpath \in \mfinpathsendsin{\sem{\oimcfrom{\anoimdp}}}{\omstates}(\astate)$,
$\anab \in \ass{\cftrans(\last(\afincmpath),\cdot)}$ and
$(\last(\afincmpath),\anomact) \in \omstacts$
(note that we use the notation $\afincmpath \anab (\last(\afincmpath),\anomact)$
to denote the finite path with prefix $\afincmpath$
and suffix $\last(\afincmpath) \anab (\last(\afincmpath),\anomact)$;
we use similar notation throughout this and subsequent proofs).
We define the choice of $\ansched$ after $\afincmpath'$ as follows.
Consider the case in which $(\anomact,\pol{\afincmpath}{\anomact}) \in \support(\asched(\cmTOm(\afincmpath)))$;
then we let $\ansched(\afincmpath')$ be such that $\ansched(\afincmpath')(\pol{\afincmpath}{\anomact}) = 1$.
Instead, in the case in which there does not exist any
$(\anomact,\ana) \in \support(\asched(\cmTOm(\afincmpath)))$,
we let $\ansched(\afincmpath')$ be an arbitrary distribution.
We now proceed to the third part of the proof,
which consists in showing that
$\probl{\asched}{\astate}(\Diamond \targetset) = \probl{\ansched}{\astate}(\Diamond \targetset)$.
Recall that $\mfinpaths{\asched}(\astate)$ and $\mfinpaths{\ansched}(\astate)$
are the sets of finite paths from state $\astate$ induced by $\asched$ and $\ansched$, respectively.
Let $\mfinpathsendsin{\ansched}{\omstates}(\astate)
= \mfinpaths{\ansched}(\astate) \cap \mfinpathsendsin{\sem{\oimcfrom{\anoimdp}}}{\omstates}(\astate)$
and
$\mfinpathsendsin{\ansched}{\omstacts}(\astate)
= \mfinpaths{\ansched}(\astate) \cap \mfinpathsendsin{\sem{\oimcfrom{\anoimdp}}}{\omstacts}(\astate)$.
Let $\probfin{\asched}{\astate}$ (respectively, $\probl{\ansched}{\astate}$)
be the probability measure over finite paths induced by $\asched$ (respectively, $\ansched$),
defined in the standard manner.
In particular, we note that, for finite paths
$\afincmpath, \afincmpath' \in \mfinpathsendsin{\ansched}{\omstates}(\astate)$,
if $\afincmpath' = \afincmpath \anab (\last(\afincmpath),\anomact) \anac \astate'$,
then $\probfin{\ansched}{\astate}(\afincmpath') =
\probfin{\ansched}{\astate}(\afincmpath) \cdot \ansched(\afincmpath)(\anab) \cdot \anab[(\last(\afincmpath),\anomact)] \cdot
\ansched(\afincmpath \anab (\last(\afincmpath),\anomact))(\anac) \cdot
\anac(\astate')$
(where $\probfin{\ansched}{\astate}(\astate) = 1$,
and where we write $\anab[(\last(\afincmpath),\anomact)]$ to denote the probability
of $(\last(\afincmpath),\anomact)$ according to the assignment $\anab$
rather than $\anab(\last(\afincmpath),\anomact)$ to avoid ambiguity).
Similarly, for finite paths $\afinpath, \afinpath' \in \mfinpaths{\asched}(\astate)$,
if $\afinpath' = \afinpath (\anomact,\ana) \astate'$,
then $\probfin{\asched}{\astate}(\afincmpath') =
\probfin{\asched}{\astate}(\afincmpath) \cdot \asched(\afincmpath)(\anomact,\ana) \cdot \ana(\astate')$
(where $\probfin{\asched}{\astate}(\astate) = 1$).
We show that $\probfin{\ansched}{\astate}(\afincmpath) = \probfin{\asched}{\astate}(\cmTOm(\afincmpath))$,
for $\afincmpath \in \mfinpathsendsin{\ansched}{\omstates}(\astate)$,
by induction on the length of paths,
where the length of a finite path of $\mfinpaths{\asched}(\astate)$ or
$\mfinpathsendsin{\ansched}{\omstates}(\astate)$
is the number of states from $\omstates$ along the path
(although paths $\mfinpathsendsin{\ansched}{\omstates}(\astate)$
visit states from $\omstates$ \emph{and} from $\omstacts$,
when considering the length of such a path, we only consider states from $\omstates$).
For the base case, i.e., for the path of length $1$, which comprises state $\astate$ only,
we note that $\cmTOm(\astate)=\astate$,
and that $\probfin{\ansched}{\astate}(\astate) = \probfin{\asched}{\astate}(\cmTOm(\astate)) = 1$.
Now consider $\afincmpath, \afincmpath' \in \mfinpathsendsin{\ansched}{\omstates}(\astate)$
such that $\afincmpath' = \afincmpath \anab (\last(\afincmpath),\anomact) \anac \astate'$,
and assume that we have already established that
$\probfin{\ansched}{\astate}(\afincmpath) = \probfin{\asched}{\astate}(\cmTOm(\afincmpath))$.
Furthermore (given that $\afincmpath' \in \mfinpathsendsin{\ansched}{\omstates}(\astate)$),
we have $\anab[(\last(\afincmpath),\anomact)] = 1$ and $\anac = \pol{\afincmpath}{\anomact}$.
Recalling that
$\cmTOm(\afincmpath') = \cmTOm(\afincmpath) (\anomact,\pol{\afincmpath}{\anomact}) \astate'$
and, by construction, $\ansched(\afincmpath)(\anab)=\asched(\cmTOm(\afincmpath))(\anomact,\pol{\afincmpath}{\anomact})$
and $\ansched(\afincmpath \anab (\last(\afincmpath),\anomact))(\pol{\afincmpath}{\anomact})=1$,
we can establish the following:
\begin{eqnarray*}
\probfin{\ansched}{\astate}(\afincmpath')
& = &
\probfin{\ansched}{\astate}(\afincmpath) \cdot \ansched(\afincmpath)(\anab) \cdot \anab[(\last(\afincmpath),\anomact)]
\cdot \ansched(\afincmpath \anab (\last(\afincmpath),\anomact))(\anac) \cdot \anac(\astate')
\\
& = &
\probfin{\asched}{\astate}(\cmTOm(\afincmpath))
\cdot \asched(\cmTOm(\afincmpath))(\anomact,\pol{\afincmpath}{\anomact})
\cdot \pol{\afincmpath}{\anomact}(\astate')
\\
& = &
\probfin{\asched}{\astate}(\cmTOm(\afincmpath')) \; .
\end{eqnarray*}
The fact that $\probfin{\ansched}{\astate}$ and $\probfin{\asched}{\astate}$
assign the same probability to finite paths related by $\cmTOm$,
together with the fact that finite paths related by $\cmTOm$ visit the same states from $\omstates$,
means that
$\probfin{\ansched}{\astate}$ and $\probfin{\asched}{\astate}$
assign the same probability to sets of finite paths that visit $\targetset$.
By standard reasoning, this means that
$\probl{\asched}{\astate}(\Diamond \targetset) = \probl{\ansched}{\astate}(\Diamond \targetset)$.
\end{proof}

\begin{lem}\label{lem:imc2imdp}
Let $\anoimdp = (\omstates,\omactions,\omtrans)$ be an IMDP,
$\astate \in \omstates$ and $\targetset \subseteq \omstates$.
Then, for each $\ansched \in \strats{\sem{\oimcfrom{\anoimdp}}}$,
there exists $\asched \in \strats{\sem{\anoimdp}}$
such that
$\probl{\asched}{\astate}(\Diamond \targetset) = \probl{\ansched}{\astate}(\Diamond \targetset)$.
\end{lem}
\begin{proof}
We proceed by constructing $\asched$ from $\ansched$
by using the ``reverse'' of the construction used for \lemref{lem:imdp2imc}:
each sequence of two choices of $\ansched$,
the first determining a transition from a state in $\omstates$ to a state in $\omstacts$,
the second determining a transition from a state in $\omstacts$ to a state in $\omstates$,
is mimicked by a single choice of $\asched$.

In the following,
we assume w.l.o.g.\ that, for each finite path $\afinpath \in \mfinpaths{\sem{\oimcfrom{\anoimdp}}}$,
there exists some $\anab \in \ass{\cftrans(\last(\afinpath),\cdot)}$
such that $\ansched(\afinpath)(\anab) = 1$.
The assumption is w.l.o.g.\ because
the set $\ass{\cftrans(\last(\afinpath),\cdot)}$ is closed under convex combinations.

As in the proof of \lemref{lem:imdp2imc},
we introduce a function from the set of finite paths of $\mfinpaths{\sem{\anoimdp}}(\astate)$
to the set of finite paths of $\mfinpaths{\sem{\oimcfrom{\anoimdp}}}$;
however, in contrast to the proof of \lemref{lem:imdp2imc},
the range of the function includes only finite paths that are generated by $\ansched$,
with $\bot$ used to represent cases in which a finite path of $\mfinpaths{\sem{\anoimdp}}(\astate)$
has no corresponding finite path in $\mfinpaths{\ansched}(\astate)$.
Let $\mTOcm: \mfinpaths{\sem{\anoimdp}}(\astate) \ra \mfinpaths{\ansched}(\astate) \cup \set{\bot}$
be the function defined as follows:
for $\afinpath \in \mfinpaths{\sem{\anoimdp}}(\astate)$,
where  $\afinpath = \astate_0 (\anomact_0,\ana_0) \astate_1 (\anomact_1,\ana_1) \cdots (\anomact_{n-1},\ana_{n-1}) \astate_n$,
if there exists sequence $\anab_0 \anab_1 \cdots \anab_{n-1}$
for which the finite path
\[
\afincmpath' = \astate_0 \anab_0 (\astate_0,\anomact_0) \ana_0
\astate_1 \anab_1 (\astate_1,\anomact_1) \ana_1  \cdots
\astate_{n-1} \anab_{n-1} (\astate_{n-1},\anomact_{n-1}) \ana_{n-1}
\astate_n
\]
is such that $\afincmpath' \in \mfinpaths{\ansched}(\astate)$,
then $\mTOcm(\afinpath) = \afincmpath'$,
otherwise $\mTOcm(\afinpath) = \bot$.
Note that, from the assumption on $\ansched$ made in the previous paragraph,
if such a sequence $\anab_0 \anab_1 \cdots \anab_{n-1}$ exists,
it is unique, and hence the function $\mTOcm$ is well-defined.

Let $\afinpath \in \mfinpaths{\sem{\anoimdp}}(\astate)$.
If $\mTOcm(\afinpath) = \bot$,
then we define $\asched(\afinpath)$ to be an arbitrary distribution.
Otherwise we derive $\asched(\afinpath)$ from $\ansched(\mTOcm(\afinpath))$.
First, recall that $\ansched(\mTOcm(\afinpath))(\anab) = 1$
for some $\anab \in \ass{\cftrans(\last(\afinpath),\cdot)}$.
Furthermore, recall that $\anab \in \dist(\omstacts)$ and note that, by construction,
for all $(\astate',\anomact) \in \support(\anab)$
we have $\astate' = \last(\afinpath)$.
Consider some $(\last(\afinpath),\anomact) \in \support(\anab)$.
For the resulting finite path $\mTOcm(\afinpath) \anab (\last(\afinpath),\anomact)$,
we now consider $\ansched(\mTOcm(\afinpath) \anab (\last(\afinpath),\anomact))$.
As above, 
we can assume w.l.o.g.\ that $\ansched(\mTOcm(\afinpath) \anab (\last(\afinpath),\anomact))(\anac) = 1$
for some $\anac \in \ass{\cftrans((\last(\afinpath),\anomact),\cdot)}$.
We can now define $\asched(\afinpath)(\anomact',\ana)$ for each $(\anomact',\ana) \in \mactionsof{\astate}$:
if $(\anomact',\ana) = (\anomact,\anac)$,
then $\asched(\afinpath)(\anomact',\ana) = \anab[(\last(\afinpath),\anomact')]$
and $\asched(\afinpath)(\anomact',\ana) = 0$ otherwise.

It remains to establish that
$\probl{\asched}{\astate}(\Diamond \targetset) = \probl{\ansched}{\astate}(\Diamond \targetset)$.
We proceed by showing that
$\probfin{\asched}{\astate}(\afinpath) = \probfin{\ansched}{\astate}(\mTOcm(\afinpath))$
for $\afinpath \in \mfinpaths{\asched}(\astate)$
by induction on the length of paths.
For the base case,
given that $\mTOcm(\astate)=\astate$,
we have $\probfin{\asched}{\astate}(\astate) = \probfin{\ansched}{\astate}(\mTOcm(\astate)) = 1$.
Now consider $\afinpath, \afinpath' \in \mfinpaths{\asched}(\astate)$
such that $\afinpath' = \afinpath (\anomact,\ana) \astate'$,
and assume that we have already established that
$\probfin{\asched}{\astate}(\afinpath) = \probfin{\ansched}{\astate}(\mTOcm(\afinpath))$.
From the definition of $\mTOcm$, 
we have that $\mTOcm(\afinpath') = \mTOcm(\afinpath) \anab (\last(\afinpath),\anomact) \ana \astate'$,
where $\anab \in \ass{\cftrans(\last(\afincmpath),\cdot)}$ such that
$\ansched(\mTOcm(\afinpath))(\anab) = 1$
(and hence $\mTOcm(\afinpath') \in \mfinpaths{\ansched}(\astate)$).
Recall that we have constructed $\asched$ such that
$\asched(\afinpath)(\anomact,\ana) = \anab[(\last(\afinpath),\anomact)]$;
furthermore, we have assumed w.l.o.g.\ $\ansched(\mTOcm(\afinpath))(\anab) = 1$
and $\ansched(\mTOcm(\afinpath) \anab (\last(\afinpath),\anomact))(\ana) = 1$.
Then we have:
\begin{eqnarray*}
\probfin{\asched}{\astate}(\afinpath')
& = & \probfin{\asched}{\astate}(\afinpath)
\cdot \asched(\afinpath)(\anomact,\ana)
\cdot \ana(\astate')
\\
& = & \probfin{\ansched}{\astate}(\mTOcm(\afinpath)) \cdot \ansched(\mTOcm(\afinpath))(\anab)
\cdot \anab[(\last(\afinpath),\anomact)]
\cdot \ansched(\mTOcm(\afinpath) \anab (\last(\afincmpath),\anomact))(\ana) \cdot \ana(\astate')
\\
& = & \probfin{\ansched}{\astate}(\mTOcm(\afinpath')) \; .
\end{eqnarray*}
As in the proof of \lemref{lem:imdp2imc},
this fact suffices to establish that
$\probl{\asched}{\astate}(\Diamond \targetset) = \probl{\ansched}{\astate}(\Diamond \targetset)$.
\end{proof}

\section{Clock-Dependent Probabilistic Timed Automata with One Clock}\label{sec:onec-cdpta}
%

In this section, we recall the formalism of clock-dependent probabilistic timed automata.
The definition of clock-dependent probabilistic timed automata of~\cite{Spr21}
features an arbitrary number of clock variables.
In contrast,
we consider models with only \emph{one} clock variable.
This clock variable will be denoted $\aclock$ for the remainder of the paper.

A \emph{clock valuation} is a value $\aval \in \nnr$,
interpreted as the current value of clock $\aclock$.
Following the usual notational conventions for modelling formalisms based on timed automata,
we use the powerset notation $2^{\set{\aclock}}$
to refer to the set $\set{\set{\aclock},\emptyset}$,
which we will use in the sequel to indicate whether the clock is reset to $0$
(denoted by $\set{\aclock}$)
or retains its current value (denoted by $\emptyset$).

%
The set $\ccons$ of \emph{clock constraints} over $\aclock$ is defined
as the set of conjunctions over atomic formulae of the
form $\aclock \sim c$,
where $\sim \in \{ <,\leq,\geq,> \}$ and $c \in \Nset$.
A clock valuation $\aval$ satisfies a clock constraint $\ccl$,
denoted by $\aval \vinz \ccl$,
if $\ccl$ resolves to $\true$ when substituting each occurrence of clock $\aclock$ with $\aval$.

For a set $Q$,
a \emph{distribution template} $\adisttemp : \nnr \ra \dist(Q)$
gives a distribution over $Q$ for each clock valuation.
In the following, we use notation $\adtpara{\adisttemp}{\aval}$, rather than $\adisttemp(\aval)$,
to denote the distribution corresponding to distribution template $\adisttemp$ and clock valuation $\aval$.
Let $\disttemps(Q)$ be the set of distribution templates over $Q$.

A \emph{one-clock clock-dependent probabilistic timed automaton} (1c-cdPTA)
$\acdpta = (\locs, \inv, \pedges)$
comprises the following components:
\begin{itemize}
\item
a finite set
$\locs$ of \emph{locations}; 
\item
a function $\inv: \locs \ra \ccons$ associating an \emph{invariant condition} with each location;
\item
a set $\pedges \subseteq \locs \times \ccons \times \disttemps(2^{\set{\aclock}} \times \locs)$
of \emph{probabilistic edges}.
\end{itemize}
A probabilistic edge $(\aloc,\g,\pd) \in \pedges$ comprises:
(1)~a source location $\aloc$;
(2)~a clock constraint $\g$, called a \emph{guard}; and
(3)~a distribution template $\pd$
with respect to pairs of the form $(\aclockset,\aloc') \in 2^{\set{\aclock}} \times \locs$
(i.e., pairs consisting of a first element indicating whether $\aclock$ should be reset to $0$ or not,
and a second element corresponding to a target location $\aloc'$).
We refer to pairs $(\aclockset,\aloc') \in 2^{\set{\aclock}} \times \locs$ as \emph{outcomes}.
The behaviour of a 1c-cdPTA takes a similar form to
that of a standard (one-clock) probabilistic timed automaton~\cite{GJ95,KNSS02,JLS08}.
%
A \emph{state} of a 1c-cdPTA is a pair comprising a location and a clock valuation
satisfying the location's invariant condition,
i.e., $(\aloc,\aval) \in \locs \times \nnr$ such that $\aval \vinz \inv(\aloc)$.
In any state $(\aloc,\aval)$, a certain amount of time $\adelay \in \nnr$ elapses,
then a probabilistic edge is traversed.
The choice of $\adelay$ is nondeterministic.
It is required that the invariant $\inv(\aloc)$ remains satisfied continuously while time passes.
The resulting state after the elapse of time is $(\aloc,\aval{+}\adelay)$.
A probabilistic edge $(\aloc',\g,\pd) \in \pedges$ can then be chosen
from state $(\aloc,\aval{+}\adelay)$
if $\aloc = \aloc'$ and it is \emph{enabled},
i.e., the clock constraint $\g$ is satisfied by $\aval{+}\adelay$.
The choice of which enabled probabilistic edge to take is nondeterministic.
Once a probabilistic edge $(\aloc',\g,\pd)$ is chosen,
a successor location, and whether to reset the clock to $0$,
is chosen at random,
according to the distribution $\adtpara{\pd}{\aval{+}\adelay}$.
For example, in the case of the 1c-cdPTA of \figref{fig:one_clock},
from state $(\mathrm{W},0)$
(i.e., the location is $\mathrm{W}$ and the value of clock $\aclock$ is equal to $0$),
a time delay $\adelay \in (1,3)$ elapses,
increasing the value of $\aclock$ to $\adelay$,
before the probabilistic edge leaving $\mathrm{W}$ is traversed.
Then the resulting state will be $(\mathrm{S},\adelay)$ with probability $\frac{3 \adelay - 3}{8}$,
$(\mathrm{T},\adelay)$ with probability $\frac{11 \adelay - 3}{16}$,
and $(\mathrm{F},\adelay)$ with probability $\frac{11 \adelay - 3}{16}$.

We make the following assumptions on 1c-cdPTAs,
in order to simplify the definition of their semantics.
Firstly, we consider 1c-cdPTAs featuring invariant conditions
that prevent the clock from exceeding some upper bound
and impose no lower bound:
formally, for each location $\aloc \in \locs$,
we have that $\inv(\aloc)$
is a constraint $\aclock \leq c$ for some $c \in \Nset$,
or a constraint $\aclock < c$ for some $c \in \Nset \setminus \set{0}$.
Secondly, we restrict our attention to 1c-cdPTAs for which it is always possible to take a probabilistic edge,
either immediately or after letting time elapse.
Formally, for each location $\aloc \in \locs$,
if $\inv(\aloc) = (\aclock \leq c)$ then (viewing $c$ as a clock valuation) $c \vinz \g$
for some $(\aloc,\g,\pd) \in \pedges$;
instead, if $\inv(\aloc) = (\aclock < c)$ then $c-\varepsilon \vinz \g$
for all $\varepsilon \in (0,1)$ and $(\aloc,\g,\pd) \in \pedges$.
Thirdly, we assume that all possible target states of probabilistic edges satisfy their invariants.
Observe that,
given the first assumption, this may not be the case only when the clock is \emph{not} reset.
Formally, for all probabilistic edges $(\aloc,\g,\pd) \in \pedges$,
for all clock valuations $\aval \in \nnr$ such that $\aval \vinz \g$,
and for all $\aloc' \in \locs$,
we have that $\adtpara{\pd}{\aval}(\emptyset,\aloc')>0$ implies
$\aval \vinz \inv(\aloc')$
(recall that the clock is not reset in the case of outcomes for which $\emptyset$ is the first component).
Note that we relax some of these assumptions when depicting 1c-cdPTAs graphically
(for example, the 1c-cdPTA  of \figref{fig:one_clock}
can be modified so that it satisfies these assumptions by adding invariant conditions
and self-looping probabilistic edges to locations $\mathrm{S}$ and $\mathrm{T}$).

The semantics of the 1c-cdPTA
$\acdpta = (\locs, \inv, \pedges)$
is the MDP
$\semmdp{\acdpta} = ( \mstates, \mactions, \mtrans )$
where:
\begin{itemize}
\item
$\mstates =
\{ (\aloc,\aval) \in \locs \times \nnr \setsep \aval \vinz \inv(\aloc) \}$;
\item
$\mactions = \nnr \times \pedges$;
\item
for $(\aloc,\aval) \in \mstates$,
$\timeelapse{\aval} \in \nnr$ and $(\aloc,\g,\pd) \in \pedges$ such that
(1)~$\timeelapse{\aval} \geq \aval$, (2)~$\timeelapse{\aval} \vinz \g$
and (3)~$\anval \vinz \inv(\aloc)$ for all $\aval \leq \anval \leq \timeelapse{\aval}$,
then we let $\mtrans((\aloc,\aval),(\timeelapse{\aval},(\aloc,\g,\pd)))$ be the distribution
such that, for $(\aloc',\aval') \in \mstates$:
\[
\mtrans((\aloc,\aval),(\timeelapse{\aval},(\aloc,\g,\pd)))(\aloc',\aval') =
\left\{
\begin{array}{ll}
\adtpara{\pd}{\timeelapse{\aval}}(\set{\aclock},\aloc')
+
\adtpara{\pd}{\timeelapse{\aval}}(\emptyset,\aloc')
&
\mbox{if } \aval' = \timeelapse{\aval} = 0
\\
\adtpara{\pd}{\timeelapse{\aval}}(\emptyset,\aloc')
&
\mbox{if } \aval' = \timeelapse{\aval} > 0
\\
\adtpara{\pd}{\timeelapse{\aval}}(\set{\aclock},\aloc')
&
\mbox{if } \aval' = 0 \mbox{ and } \timeelapse{\aval} > 0
\\
0
&
\mbox{otherwise.}
\end{array}
\right.
\]
\end{itemize}
Note that the summation in the first case of the definition of $\mtrans$
is performed because the valuation $\aval' = 0$ may be obtained from $\timeelapse{\aval} = 0$
either by resetting the clock to $0$ (summand $\adtpara{\pd}{\timeelapse{\aval}}(\set{\aclock},\aloc')$)
or by not resetting the clock to $0$ (summand $\adtpara{\pd}{\timeelapse{\aval}}(\emptyset,\aloc')$).

Let $\targetlocs \subseteq \locs$ be a set of locations,
and let $\targetsetof{\targetlocs} = \set{ (\aloc,\aval) \in \mstates \setsep \aloc \in \targetlocs}$
be the set of states of $\semmdp{\acdpta}$
that have their location component in $\targetlocs$.
Then the maximum value of reaching $\targetlocs$ from state $(\aloc,\aval) \in \mstates$
corresponds to $\maxval{\semmdp{\acdpta}}{\targetsetof{\targetlocs}}{(\aloc,\aval)}$.
Similarly, the minimum value of reaching $\targetlocs$ from state $(\aloc,\aval) \in \mstates$
corresponds to $\minval{\semmdp{\acdpta}}{\targetsetof{\targetlocs}}{(\aloc,\aval)}$.
As in \sectref{sec:oimdps}, we can define a number of quantitative and qualitative reachability problems on 1c-cdPTA,
where the initial state is $(\aloc,0)$ for a particular $\aloc \in \locs$.
The maximal reachability problem for $\acdpta$, $\targetlocs \subseteq \locs$,
$\aloc \in \locs$, $\unrhd \in \{ \geq, > \}$ and $\pthresh \in [0,1]$
is to decide whether $\maxval{\semmdp{\acdpta}}{\targetsetof{\targetlocs}}{(\aloc,0)} \unrhd \pthresh$;
similarly, the minimal reachability problem for $\acdpta$, $\targetlocs \subseteq \locs$,
$\aloc \in \locs$, $\unlhd \in \{ \leq, < \}$ and $\pthresh \in [0,1]$
is to decide whether $\minval{\semmdp{\acdpta}}{\targetsetof{\targetlocs}}{(\aloc,0)} \unlhd \pthresh$.
Furthermore, we can define analogues of the qualitative problems
featured in \sectref{sec:oimdps}:
($\forall 0$)
decide whether
$\probl{\astrat}{(\aloc,0)}(\Diamond \targetsetof{\targetlocs})=0$ for all $\astrat \in \strats{\semmdp{\acdpta}}$;
($\exists 0$)
decide whether there exists $\astrat \in \strats{\semmdp{\acdpta}}$
such that $\probl{\astrat}{(\aloc,0)}(\Diamond \targetsetof{\targetlocs})=0$;
($\exists 1$)
decide whether there exists $\astrat \in \strats{\semmdp{\acdpta}}$
such that $\probl{\astrat}{(\aloc,0)}(\Diamond \targetsetof{\targetlocs})=1$;
($\forall 1$)
decide whether $\probl{\astrat}{(\aloc,0)}(\Diamond \targetsetof{\targetlocs})=1$ for all $\astrat \in \strats{\semmdp{\acdpta}}$.


\subsection{Affine Clock Dependencies.}
In this paper, we consider distribution templates
that are defined in terms
of sets of affine functions in the following way.
Given probabilistic edge $\apedge = (\aloc,\g,\pd) \in \pedges$,
let $\interval{\apedge}$ be the set of clock valuations in which $\apedge$ is enabled,
i.e., $\interval{\apedge} = \set{\aval \in \nnr \setsep \aval \sat \g \wedge \inv(\aloc)}$.
Note that $\interval{\apedge} \subseteq \nnr$ corresponds to an interval
with natural-numbered endpoints.
Let $\closure{\interval{\apedge}}$ be the closure of $\interval{\apedge}$.
We say that $\apedge$ is \emph{affine} if,
for each $\outcome \in 2^{\set{\aclock}} \times \locs$,
there exists a pair $(\fc{\apedge}{\outcome}, \fd{\apedge}{\outcome}) \in \Qset^2$ of rational constants,
such that
$\adtpara{\pd}{\aval}(\outcome) = \fc{\apedge}{\outcome} + \fd{\apedge}{\outcome} \cdot \aval$
for all $\aval \in \closure{\interval{\apedge}}$.
Note that, by the definition of distribution templates,
for all $\aval \in \closure{\interval{\apedge}}$,
we have $\fc{\apedge}{\outcome} + \fd{\apedge}{\outcome} \cdot \aval \geq 0$
for each $\outcome \in 2^{\set{\aclock}} \times \locs$, and
$\sum_{\outcome  \in 2^{\set{\aclock}} \times \locs} (\fc{\apedge}{\outcome} + \fd{\apedge}{\outcome} \cdot \aval) = 1$.
A 1c-cdPTA is affine if all of its probabilistic edges are affine.
Henceforth we assume that the 1c-cdPTAs we consider are affine.
An affine probabilistic edge $\apedge$ is \emph{constant} if, for each $\outcome \in 2^{\set{\aclock}} \times \locs$,
we have $\fd{\apedge}{\outcome} = 0$,
i.e., $\adtpara{\pd}{\aval}(\outcome) = \fc{\apedge}{\outcome}$  for some $\fc{\apedge}{\outcome} \in \Qset$,
for all $\aval \in \closure{\interval{\apedge}}$.
The following technical fact will be useful in subsequent sections:
for a probabilistic edge $\apedge \in \pedges$,
outcome $\outcome \in 2^{\set{\aclock}} \times \locs$
and open interval $\anint \subseteq \interval{\apedge}$,
if $\fd{\apedge}{\outcome} \neq 0$,
then $\adtpara{\pd}{\aval}(\outcome) > 0$ for all $\aval \in \anint$
(because the existence of $\aval_{=0} \in \anint$ such that
$\adtpara{\pd}{\aval_{=0}}(\outcome) = 0$,
together with $\fd{\apedge}{\outcome} \neq 0$ and the fact that $\anint$ is open,
would mean that there exists $\aval' \in \anint$ such that
$\adtpara{\pd}{\aval'}(\outcome) < 0$,
which contradicts the definition of distribution templates).

\subsection{Initialisation.}
In this paper, we also introduce a specific requirement for 1c-cdPTAs
that allows us to analyse faithfully 1c-cdPTA using IMDPs
in Section~\ref{sec:general}.
A \emph{symbolic path fragment} is a sequence
\[
(\aloc_0,\g_0,\pd_0) (\aclockset_0,\aloc_1) (\aloc_1,\g_1,\pd_1) (\aclockset_1,\aloc_2)
\cdots
(\aloc_n,\g_n,\pd_n)
\in {(\pedges \times (2^{\set{\aclock}} \times \locs))}^+ \times \pedges
\]
of probabilistic edges and outcomes
such that
$\adtpara{\pd_i}{\aval}(\aclockset_i,\aloc_{i+1}) > 0$ for all $\aval \in \interval{(\aloc_i,\g_i,\pd_i)}{}$
for all $i<n$.
In this paper, we consider 1c-cdPTAs for which each symbolic path fragment
that begins and ends with non-constant probabilistic edges
requires that the clock takes a natural numbered value
at some point along the path fragment,
either from being reset or from passing through guards that have at most one (natural numbered) value in common.
Formally, a 1c-cdPTA is \emph{initialised} if,
for any symbolic path fragment
$(\aloc_0,\g_0,\pd_0) (\aclockset_0,\aloc_1) (\aloc_1,\g_1,\pd_1) (\aclockset_1,\aloc_2)
\linebreak[0] \cdots \linebreak[0]
(\aloc_n,\g_n,\pd_n)$
such that $(\aloc_0,\g_0,\pd_0)$ and $(\aloc_n,\g_n,\pd_n)$ are non-constant,
either (1)~$\aclockset_i = \set{\aclock}$ for some $0 \leq i < n$
or (2)~$\interval{(\aloc_i,\g_i,\pd_i)}{} \cap \interval{(\aloc_{i+1},\g_{i+1},\pd_{i+1})}{}$
is empty or contains a single valuation,
for some $0< i < n$.
We henceforth assume that all 1c-cdPTAs considered in this paper are initialised.

\section{Translation from 1c-cdPTAs to IMDPs}\label{sec:general}
In this section,
we show that we can solve quantitative and qualitative problems of (affine and initialised) 1c-cdPTAs.
In contrast to the approach for quantitative problems of multiple-clock cdPTAs presented in~\cite{Spr21},
which involves the construction of an \emph{approximate} MDP,
we represent the 1c-cdPTA \emph{precisely} using an IMDP,
by adapting the standard region-graph construction for
one-clock (probabilistic) timed automata of~\cite{LMS04,JLS08}.


\begin{figure}[t]
\centering

\tiny

\begin{tikzpicture}[shorten >=1pt,auto, thin] 

\matrix (m) [matrix of math nodes,row sep=3em,column sep=6em,minimum width=2em]
  {
     \mbox{1c-cdPTA }\acdpta  &  \mbox{IMDP }\oimdpfrom{\acdpta} & \mbox{IMC }\oimcfrom{\oimdpfrom{\acdpta}}\\
     \mbox{MDP }\sem{\acdpta}  &  \mbox{MDP }\sem{\oimdpfrom{\acdpta}} & \mbox{MDP }\sem{\oimcfrom{\oimdpfrom{\acdpta}}} \\};
  \path[->]
    (m-1-1) edge (m-2-1)
            edge [double] (m-1-2)
    (m-1-2) edge (m-2-2)
            edge [double] (m-1-3)
    (m-1-3) edge (m-2-3);
  \path[dashed]
    (m-2-1) edge node [below,<->] {\tiny \begin{tabular}{c}{\lemref{lem:sched_cdPTA2IMDP}}\\ {\lemref{lem:sched_IMDP2cdPTA}}\end{tabular}} (m-2-2)
    (m-2-2) edge node [below,<->] {\tiny \begin{tabular}{c}{\lemref{lem:imdp2imc}}\\ {\lemref{lem:imc2imdp}}\end{tabular}} (m-2-3);

\end{tikzpicture}

\caption{Overview of the use of the IMDP construction in the overall solution process.}%
\label{fig:summary}

\end{figure}


We summarise our overall approach in \figref{fig:summary}.
From the 1c-cdPTA $\acdpta$,
we construct an IMDP $\oimdpfrom{\acdpta}$,
from which we can obtain in turn the IMC $\oimcfrom{\oimdpfrom{\acdpta}}$
according to the construction of Section~\ref{sec:oimdps}
(top line, from left to right).
Furthermore, by \lemref{lem:imdp2imc} and \lemref{lem:imc2imdp},
for any scheduler of $\sem{\oimdpfrom{\acdpta}}$,
we can find an equivalent (in terms of assigning the same probabilities for reaching a set of locations)
scheduler of $\sem{\oimcfrom{\oimdpfrom{\acdpta}}}$,
and vice versa
(as indicated by the lower right dashed line).
In this section, in \lemref{lem:sched_cdPTA2IMDP} and \lemref{lem:sched_IMDP2cdPTA},
we will also show that analogous results allow us to find,
for any scheduler of $\sem{\acdpta}$, an equivalent scheduler of $\sem{\oimdpfrom{\acdpta}}$,
and vice versa
(lower left dashed line).
Overall, this allows us to relate the quantitative and qualitative problems
defined at the level of 1c-cdPTA to analogous problems defined at the level of IMCs,
for which there exist efficient solution algorithms~\cite{PLSS13,CHK13,CK15,Spr18}.

Let $\acdpta = (\locs, \inv, \pedges)$ be a 1c-cdPTA\@.
Let $\constants{\acdpta}$ be the set of constants that are used in the guards of probabilistic edges
and invariants of $\acdpta$,
and let $\borders = \constants{\acdpta} \cup \set{0}$.
We write $\borders = \set{\ab_0, \ab_1, \ldots , \ab_k}$,
where $0 = \ab_0 < \ab_1 < \ldots < \ab_k$.
The set $\borders$ defines the set $\intervalsof{\borders} =
\set{[\ab_0,\ab_0], (\ab_0,\ab_1), [\ab_1,\ab_1], \cdots , [\ab_k,\ab_k]}$,
i.e., $\intervalsof{\borders}$ is a partition of the interval $[0,\ab_k]$
into subintervals with endpoints in $\borders$,
and where each element of $\borders$ has a corresponding closed interval in $\intervalsof{\borders}$
comprising only that element.
We define a total order on $\intervalsof{\borders}$ in the following way:
$[\ab_0,\ab_0] < (\ab_0,\ab_1) < [\ab_1,\ab_1] < \cdots < [\ab_k,\ab_k]$.
Given an open interval $\abint=(\ab,\ab') \in \intervalsof{\borders}$,
its closure is written as $\closure{\abint}$, i.e., $\closure{\abint}=[\ab,\ab']$.
Furthermore, let $\lb(\abint)=\ab$ and $\rb(\abint)=\ab'$ refer to the left- and right-endpoints of $\abint$.
For a closed interval $[\ab,\ab] \in  \intervalsof{\borders}$,
we let $\lb(\abint)=\rb(\abint)=\ab$.

Let $\ccl$ be a guard of a probabilistic edge or an invariant of $\acdpta$.
By definition, we have that, for each $\abint \in \intervalsof{\borders}$,
either $\abint \subseteq \set{\aval \in \nnr \setsep \aval \vinz \ccl}$ or
$\abint \cap \set{\aval \in \nnr \setsep \aval \vinz \ccl} = \emptyset$.
We write $\abint \vinz \ccl$ in the case of $\abint \subseteq \set{\aval \in \nnr \setsep \aval \vinz \ccl}$
(representing the fact that all valuations of $\abint$ satisfy $\ccl$).

\begin{exa}\label{ex:constants}
Consider the 1c-cdPTA of \figref{fig:one_clock}.
We have $\borders = \set{0,1,3,4,5}$
and $\intervalsof{\borders} =
\set{[0,0], (0,1), [1,1], (1,3), [3,3], (3,4), [4,4], (4,5), [5,5]}$.
Consider the clock constraint $\aclock < 3$:
we have $\abint \vinz (\aclock < 3)$ for all $\abint \in \set{[0,0], (0,1), [1,1], (1,3)}$.
Similarly, for the clock constraint $4 < \aclock < 5$,
we have $(4,5) \vinz (4 < \aclock < 5)$.
\end{exa}

\subsection{\texorpdfstring{$\borders$}--minimal Schedulers.}
The following technical lemma specifies that any scheduler of the 1c-cdPTA
can be made ``more deterministic'' in the following way:
for each interval $\timeelapse{\abint} \in \intervalsof{\borders}$ and probabilistic edge $\apedge \in \pedges$,
if, after executing a certain finite path,
a scheduler chooses (assigns positive probability to) multiple actions
$(\timeelapse{\aval}_1,\apedge), \cdots, (\timeelapse{\aval}_n,\apedge)$
that share the same probabilistic edge $\apedge$
and for which $\timeelapse{\aval}_i \in \timeelapse{\abint}$
for all $1 \leq i \leq n$,
then we can construct another scheduler for which the aforementioned actions are replaced by
a \emph{single} action $(\timeelapse{\aval},\apedge)$ such that $\timeelapse{\aval} \in \timeelapse{\abint}$.
Formally, we say that a scheduler $\asched \in \strats{\semmdp{\acdpta}}$ of $\semmdp{\acdpta}$
is \emph{$\borders$-minimal} if,
for all finite paths $\afinpath \in \mfinpaths{\semmdp{\acdpta}}$,
for all probabilistic edges $\apedge \in \pedges$,
and for all pairs of actions
$(\timeelapse{\aval}_1,\apedge_1), (\timeelapse{\aval}_2,\apedge_2) \in \support(\asched(\afinpath))$,
either $\apedge_1 \neq \apedge_2$ or $\timeelapse{\aval}_1$ and $\timeelapse{\aval}_2$
belong to distinct intervals in $\intervalsof{\borders}$,
i.e., the intervals $\timeelapse{\abint}_1, \timeelapse{\abint}_2 \in \intervalsof{\borders}$
for which $\timeelapse{\aval}_1 \in \timeelapse{\abint}_1$
and $\timeelapse{\aval}_2 \in \timeelapse{\abint}_2$
are such that $\timeelapse{\abint}_1 \neq \timeelapse{\abint}_2$.
Let $\bstrats{\semmdp{\acdpta}}$ be
the set of schedulers of $\strats{\semmdp{\acdpta}}$ that are $\borders$-minimal.
The lemma allows us to consider only $\borders$-minimal schedulers in the sequel,
permitting us to obtain a close correspondence between the schedulers of $\semmdp{\acdpta}$
and the schedulers of the IMDP that we describe how to construct in the next subsection.
\begin{lem}\label{lem:bminimal}
Let $(\aloc,\aval) \in \mfromstates{\acdpta}$ and $\targetlocs \subseteq \locs$.
Then, for each $\asched \in \strats{\semmdp{\acdpta}}$,
there exists $\ansched \in \bstrats{\semmdp{\acdpta}}$
such that
$\probl{\asched}{(\aloc,\aval)}(\Diamond \targetsetof{\targetlocs})
= \probl{\ansched}{(\aloc,\aval)}(\Diamond \targetsetof{\targetlocs}) $.
\end{lem}
%

Before presenting the proof of \lemref{lem:bminimal},
we introduce the general approach for the construction
of the scheduler $\ansched \in \bstrats{\semmdp{\acdpta}}$
from $\asched \in \strats{\semmdp{\acdpta}}$.
Each finite path of $\ansched$ will be associated with
a number of finite paths of $\asched$;
then a choice of $\ansched$ after a finite path $\afinpath$
is based on a weighted average of the corresponding choices
made by $\asched$ after the finite paths associated with $\afinpath$.
The association between finite paths of $\asched$ and finite paths of $\ansched$ is defined
on the basis of the locations that are visited,
the intervals from $\intervalsof{\borders}$ that the value of the clock passes through,
and the probabilistic edges that are taken
along those finite paths.
To define the choice of transition of $\ansched$ after a finite path $\afinpath$,
we consider each pair $(\timeelapse{\abint},\apedge) \in \intervalsof{\borders} \times \pedges$,
and, if $\asched$ assigns positive probability to at least one transition $(\timeelapse{\aval},\apedge)$
such that $\timeelapse{\aval} \in \timeelapse{\abint}$,
we define a unique transition $(\timeelapse{\aval}',\apedge)$
that is assigned positive probability by $\ansched$
and for which $\timeelapse{\aval}' \in \timeelapse{\abint}$.
There are two principal cases for the transitions assigned positive probability by $\ansched$:
if $\timeelapse{\abint}$ is a closed interval,
then there is only one possible choice for $\timeelapse{\aval}'$
(i.e., if $\timeelapse{\abint} = [\ab,\ab]$ then $\timeelapse{\aval}' = \ab$);
instead, if $\timeelapse{\abint}$ is an open interval,
the choice of clock valuation $\timeelapse{\aval}' \in \timeelapse{\abint}$
used in the definition of the $\ansched$ is substantially more complicated.
%
While the scheduler $\ansched$ chooses transitions
for which the value of the clock remains within $\timeelapse{\abint}$,
the scheduler proceeds between three phases,
where the number of transitions in the first and third phases is arbitrary,
and where the second phase consists of at most one transition.
The first phase consists of transitions derived from constant probabilistic edges,
where the value of the clock is kept ``low''
(equal to the minimum of the value of the clock at the end of the corresponding finite paths of $\asched$).
The second phase consists of at most one transition derived from a non-constant probabilistic edge,
where the value of the clock is chosen
as a weighted average of the corresponding choices
made by $\asched$,
so as to replicate exactly the probability of taking the same
probabilistic edge by $\asched$.
The third phase consists of transitions derived from constant probabilistic edges,
and the value of the clock can be chosen in an arbitrary manner.
Note that the assumption of initialisation guarantees that we cannot have
more than one non-constant probabilistic edge between points at which the clock is equal to a value from $\borders$,
and hence at most one non-constant probabilistic edge can be taken while the value of the clock
remains continuously in the open interval $\timeelapse{\abint}$.

\begin{exa}\label{ex:bminimal}
Consider the 1c-cdPTA of \figref{fig:one_clock}.
Note that the 1c-cdPTA contains only non-constant probabilistic edges,
and hence our focus will be on the choice of clock valuations obtained as weighted averages
(the second phase described above).
In the following, we denote the outgoing probabilistic edges from $\mathrm{W}$ and $\mathrm{F}$
as $\apedge_\mathrm{W}$ and $\apedge_\mathrm{F}$, respectively.
Consider a scheduler $\asched \in \strats{\semmdp{\acdpta}}$,
where $\asched(\mathrm{W},0)$
(i.e., the choice of $\asched$ after the finite path comprising the single state $(\mathrm{W},0)$)
assigns probability $\frac{1}{2}$ to the action $(\frac{5}{4},\apedge_\mathrm{W})$
and probability $\frac{1}{2}$ to the action $(\frac{7}{4},\apedge_\mathrm{W})$
(where the two actions refer to either $\frac{5}{4}$ or $\frac{7}{4}$ time units elapsing,
after which the probabilistic edge $\apedge_\mathrm{W}$ is taken).
Then, following the approach that we will describe in the proof of \lemref{lem:bminimal},
we construct a $\borders$-minimal scheduler $\ansched \in \bstrats{\semmdp{\acdpta}}$
such that $\ansched(\mathrm{W},0)$ assigns probability $1$
to the action $(\frac{3}{2},\apedge_\mathrm{W})$
(i.e., where $\frac{3}{2} = \frac{1}{2} \cdot \frac{5}{4} + \frac{1}{2} \cdot \frac{7}{4}$).
Now consider finite paths
$\afinpath = (\mathrm{W},0) (\frac{5}{4},\apedge_\mathrm{W}) (\mathrm{F},\frac{5}{4})$
and
$\afinpath' = (\mathrm{W},0) (\frac{7}{4},\apedge_\mathrm{W}) (\mathrm{F},\frac{7}{4})$,
which are generated by $\asched$.
Note that $\probfin{\asched}{(\mathrm{W},0)}(\afinpath)
= \frac{1}{2} \cdot \frac{11 - \frac{3 \cdot 5}{4}}{16}$
and $\probfin{\asched}{(\mathrm{W},0)}(\afinpath')
= \frac{1}{2} \cdot \frac{11 - \frac{3 \cdot 7}{4}}{16}$.
Say that $\asched(\afinpath)$ assigns probability $1$ to $(\frac{17}{4},\apedge_\mathrm{F})$
and $\asched(\afinpath')$ assigns probability $1$ to $(\frac{19}{4},\apedge_\mathrm{F})$.
Then we continue the construction of $\ansched$
by letting $\ansched((\mathrm{W},0)(\frac{3}{2},\apedge_\mathrm{W})(\mathrm{F},\frac{3}{2}))$
assign probability $1$
to action $(\timeelapse{\aval},\apedge_\mathrm{F})$,
where $\timeelapse{\aval} =
\probfin{\asched}{(\mathrm{W},0)}(\afinpath) \cdot 1 \cdot \frac{17}{4}
+
\probfin{\asched}{(\mathrm{W},0)}(\afinpath') \cdot 1 \cdot \frac{19}{4}$,
i.e., a weighted sum of the time delays chosen by $\asched$ after $\afinpath$ and $\afinpath'$,
where the weights correspond to the probabilities of $\afinpath$ and $\afinpath'$
under $\asched$.
As we will show in the proof of \lemref{lem:bminimal},
repeating this reasoning for all finite paths will yield a $\borders$-minimal scheduler $\ansched$ such that
the probability of reaching a set of target states from $(\mathrm{W},0)$
is the same for both $\asched$ and $\ansched$.
\end{exa}

\begin{proof}[Proof of \lemref{lem:bminimal}.]
%
Let $(\aloc,\aval) \in \mfromstates{\acdpta}$, $\targetlocs \subseteq \locs$
and $\asched \in \strats{\semmdp{\acdpta}}$.
We proceed by describing formally the construction of $\ansched  \in \bstrats{\semmdp{\acdpta}}$,
and then show that
$\probl{\asched}{(\aloc,\aval)}(\Diamond \targetsetof{\targetlocs})
= \probl{\ansched}{(\aloc,\aval)}(\Diamond \targetsetof{\targetlocs})$.
We first introduce the following notation.
Given a probabilistic edge $(\aloc,\g,\pd) \in \pedges$,
we let $\sourceof{\aloc,\g,\pd} = \aloc$ and $\guardof{\aloc,\g,\pd} = \g$.
Furthermore, given interval $\timeelapse{\abint} \in \intervalsof{\borders}$,
we say that $(\aloc',\abint') \in \locs \times \intervalsof{\borders}$
is a \emph{successor of $(\timeelapse{\abint},(\aloc,\g,\pd))$}
if either $\abint'=[0,0]$ and
$\adtpara{\pd}{\aval}(\set{\aclock},\aloc') > 0$ for some $\aval \in \timeelapse{\abint}$,
or $\abint'=\timeelapse{\abint}$ and
$\adtpara{\pd}{\aval}(\emptyset,\aloc') > 0$ for some $\aval \in \timeelapse{\abint}$.

A (finite) \emph{$\borders$-path} is a sequence
$(\aloc_0,\abint_0)(\timeelapse{\abint}_0,\apedge_0)
(\aloc_1,\abint_1)(\timeelapse{\abint}_1,\apedge_1) \linebreak[0]
\cdots \linebreak[0] (\timeelapse{\abint}_{n-1},\apedge_{n-1}) \linebreak[0] (\aloc_n,\abint_n)$
where
(1)~$\aloc_i \in \locs$, $\abint_i \in \intervalsof{\borders}$ and $\abint_i \vinz \inv(\aloc_i)$
for all $0 \leq i \leq n$,
and (2)~$\sourceof{\apedge_i} = \aloc_i$,
$\timeelapse{\abint}_i \in \intervalsof{\borders}$,
$\timeelapse{\abint}_i \geq \abint_i$,
$\timeelapse{\abint}_i \vinz \guardof{\apedge_i} \wedge \inv(\aloc_i)$
and $(\aloc_{i+1},\abint_{i+1})$ is a successor of $(\timeelapse{\abint}_i,\apedge_i)$,
for all $0 \leq i < n$.

Consider a finite path $\afinpath =
(\aloc_0,\aval_0) (\timeelapse{\aval}_0,\apedge_0)
(\aloc_1,\aval_1) (\timeelapse{\aval}_1,\apedge_1)
\cdots
(\timeelapse{\aval}_{n-1},\apedge_{n-1}) (\aloc_n,\aval_n)$
of $\semmdp{\acdpta}$
(i.e., $\afinpath \in \mfinpaths{\semmdp{\acdpta}}$),
and a $\borders$-path
$\abfinpath =
(\aloc_0',\abint_0)(\timeelapse{\abint}_0,\apedge_0')
(\aloc_1',\abint_1)(\timeelapse{\abint}_1,\apedge_1')
\cdots (\timeelapse{\abint}_{m-1},\apedge_{m-1}')(\aloc_m',\abint_m)$.
Then $\afinpath$ \emph{corresponds to} $\abfinpath$ if
(1)~$n=m$,
(2)~$\aloc_i = \aloc_i'$ and $\aval_i \in \abint_i$ for all $0 \leq i \leq n$,
and (3)~$\timeelapse{\aval}_i \in \timeelapse{\abint}_i$ and $\apedge_i = \apedge_i'$ for all $0 \leq i < n$.
Given a scheduler $\asched \in \strats{\semmdp{\acdpta}}$ 
and a $\borders$-path $\abfinpath$,
we let $\corrtwo{\asched}{\abfinpath} \subseteq \mfinpaths{\asched}$
be the set of finite paths of $\asched$ that correspond to $\abfinpath$.

Note that, for a $\borders$-minimal scheduler $\ansched \in \bstrats{\semmdp{\acdpta}}$
and a $\borders$-path
\[
\abfinpath =
(\aloc_0,\abint_0)(\timeelapse{\abint}_0,\apedge_0)
(\aloc_1,\abint_1)(\timeelapse{\abint}_1,\apedge_1)
\cdots (\timeelapse{\abint}_{n-1},\apedge_{n-1})(\aloc_n,\abint_n) \; ,
\]
there exists \emph{at most one} finite path $\afinpath \in \corrtwo{\ansched}{\abfinpath}$
(i.e., $\corrtwo{\ansched}{\abfinpath}$ is either empty or a singleton).
The reasoning underlying this fact is as follows:
for each $i < n$,
letting $\abfinpathprefix{i} =
(\aloc_0,\abint_0)(\timeelapse{\abint}_0,\apedge_0)
(\aloc_1,\abint_1)(\timeelapse{\abint}_1,\apedge_1)
\cdots (\timeelapse{\abint}_{i-1},\apedge_{i-1})(\aloc_i,\abint_i)$,
and assuming that $\corrtwo{\ansched}{\abfinpathprefix{i}}$ contains a unique path denoted by $\afinpath$,
then there exists at most one valuation $\timeelapse{\aval} \in \timeelapse{\abint}_i$
such that $\ansched(\afinpath)(\timeelapse{\aval},\apedge_i) > 0$.
Furthermore,
from the definition of $\mfromtrans{\acdpta}$,
there exists at most one valuation $\aval' \in \abint_{i+1}$
such that $\mfromtrans{\acdpta}(\last(\afinpath),(\timeelapse{\aval},\apedge_i))(\aloc_{i+1},\aval') > 0$.
Conversely, by using similar reasoning,
for any finite path $\afinpath \in \mfinpaths{\ansched}$,
there exists exactly one $\borders$-path $\abfinpath$ such that
$\afinpath \in \corrtwo{\ansched}{\abfinpath}$.

We now construct $\ansched \in \bstrats{\semmdp{\acdpta}}$
such that $\probl{\asched}{(\aloc,\aval)}(\Diamond \targetsetof{\targetlocs})
= \probl{\ansched}{(\aloc,\aval)}(\Diamond \targetsetof{\targetlocs})$,
proceeding by induction on the length of paths.
Given a finite path $\afinpath =
(\aloc_0,\aval_0) (\timeelapse{\aval}_0,\apedge_0)
(\aloc_1,\aval_1) (\timeelapse{\aval}_1,\apedge_1)
\cdots \linebreak[0]
(\timeelapse{\aval}_{n-1},\apedge_{n-1}) (\aloc_n,\aval_n)$,
we consider the problem of defining $\ansched(\afinpath)$
(assuming that $\ansched$ has been defined for all prefixes of $\afinpath$).
Let $\abfinpath =
(\aloc_0,\abint_0)(\timeelapse{\abint}_0,\apedge_0)
(\aloc_1,\abint_1)(\timeelapse{\abint}_1,\apedge_1)
\cdots (\timeelapse{\abint}_{n-1},\apedge_{n-1})(\aloc_n,\abint_n)$
be the unique $\borders$-path such that $\afinpath$ corresponds to $\abfinpath$.
Now consider the extension of $\abfinpath$
with $(\timeelapse{\abint},\apedge) \in \intervalsof{\borders} \times \pedges$
such that $\sourceof{\apedge} = \aloc_n$,
$\timeelapse{\abint} \in \intervalsof{\borders}$,
$\timeelapse{\abint} \geq \abint_n$, and
$\timeelapse{\abint} \vinz \guardof{\apedge} \wedge \inv(\aloc_n)$.
Let $\valnset{\abfinpath}{\timeelapse{\abint}}{\apedge}$
be the set of clock valuations in $\timeelapse{\abint}$ that are featured with $\apedge$
in actions that are assigned positive probability by $\asched$ after paths corresponding to $\abfinpath$;
formally:
\[
\valnset{\abfinpath}{\timeelapse{\abint}}{\apedge} =
\set{ \timeelapse{\aval} \in \timeelapse{\abint} \setsep (\timeelapse{\aval},\apedge) \in \support(\asched(\afinpath')) \mbox{ and } \afinpath' \in \corrtwo{\asched}{\abfinpath} } \; .
\]
In order to define $\ansched(\afinpath)$,
we consider a number of cases that depend on
$\afinpath$ and the pair $(\timeelapse{\abint},\apedge)$ used to extend $\abfinpath$.
In all of the cases, we identify a clock valuation $\timeelapsespecial{\abfinpath}{\timeelapse{\abint}}{\apedge} \in \timeelapse{\abint}$
that depends on the case, 
and define:
\begin{eqnarray*}
\ansched(\afinpath)(\timeelapsespecial{\abfinpath}{\timeelapse{\abint}}{\apedge},\apedge) & = &
\frac{\sum_{\afinpath' \in \corrtwo{\asched}{\abfinpath}} \probfin{\asched}{(\aloc,\aval)}(\afinpath')
\cdot \sum_{\timeelapse{\aval} \in \valnset{\abfinpath}{\timeelapse{\abint}}{\apedge}} \asched(\afinpath')(\timeelapse{\aval},\apedge) }
{\probfin{\asched}{(\aloc,\aval)}(\corrtwo{\asched}{\abfinpath})} \; ,
\end{eqnarray*}
where $\probfin{\asched}{(\aloc,\aval)}(\corrtwo{\asched}{\abfinpath})
= \sum_{\afinpath' \in \corrtwo{\asched}{\abfinpath}} \probfin{\asched}{(\aloc,\aval)}(\afinpath')$.

First we consider the case in which $\timeelapse{\abint}$ is a closed interval,
i.e., $\timeelapse{\abint} = [\ab,\ab]$ for $\ab \in \borders$.
In this case, we simply let $\timeelapsespecial{\abfinpath}{\timeelapse{\abint}}{\apedge} = \ab$.
Note that $\valnset{\abfinpath}{\timeelapse{\abint}}{\apedge}$ is a singleton,
i.e., $\valnset{\abfinpath}{\timeelapse{\abint}}{\apedge}
= \set{\timeelapsespecial{\abfinpath}{\timeelapse{\abint}}{\apedge}}$.

Next, we consider the case in which $\timeelapse{\abint}$ is an open interval.
Recall the description of the three phases given above.
Let $k \leq n$ be the maximum index for which $\abint_k \neq \timeelapse{\abint}$
(and let $k = 0$ if no such index exists).
Note that $\timeelapse{\abint}_k = \abint_{k+1} = \timeelapse{\abint}_{k+1} = \cdots
= \timeelapse{\abint}_{n-1} = \abint_n = \timeelapse{\abint}$.
First we consider the subcase in which all probabilistic edges in the sequence $\apedge_k \cdots \apedge_n$
are constant, and $\apedge$ is also constant:
this corresponds to the first phase described above.
Note that
there may be multiple choices made by $\asched$ that correspond to $(\timeelapse{\abint},\apedge)$
(i.e., from each finite path in $\corrtwo{\asched}{\abfinpath}$,
the scheduler $\asched$ may assign positive probability to \emph{multiple} time delays,
each of which corresponds to a clock valuation in $\timeelapse{\abint}$).
In order to replicate these choices in $\ansched$ we consider a time delay that results in a clock valuation
that is equal to the \emph{minimum} clock valuation assigned positive probability
after any finite path in $\corrtwo{\asched}{\abfinpath}$,
with the motivation that taking such a minimum gives $\ansched$ sufficient freedom in the second phase.
Hence we let $\timeelapsespecial{\abfinpath}{\timeelapse{\abint}}{\apedge} =
\min \valnset{\abfinpath}{\timeelapse{\abint}}{\apedge}$.

Now consider the subcase in which $\timeelapse{\abint}$ is an open interval,
and all probabilistic edges in the sequence $\apedge_k \cdots \apedge_n$
are constant, but $\apedge$ is non-constant,
i.e., corresponding to the second phase described above.
The definition of $\ansched$ is similar to that of the first phase,
although the choice of clock valuation by $\ansched$ corresponds to the \emph{weighted average}
of the choice of clock valuations made by $\asched$,
where the weights refer to the probabilities of the finite paths of $\asched$
multiplied by the probability assigned by $\asched$ to the particular clock valuation.
Formally, let:
\begin{eqnarray*}
\timeelapsespecial{\abfinpath}{\timeelapse{\abint}}{\apedge} & = &
\frac{\sum_{\afinpath' \in \corrtwo{\asched}{\abfinpath}} \probfin{\asched}{(\aloc,\aval)}(\afinpath')
\cdot \sum_{\timeelapse{\aval} \in \valnset{\abfinpath}{\timeelapse{\abint}}{\apedge}} \asched(\afinpath')(\timeelapse{\aval},\apedge) \cdot  \timeelapse{\aval}}
{\sum_{\afinpath' \in \corrtwo{\asched}{\abfinpath}}
\probfin{\asched}{(\aloc,\aval)}(\afinpath')
\cdot
\sum_{\timeelapse{\aval} \in \valnset{\abfinpath}{\timeelapse{\abint}}{\apedge}}
\asched(\afinpath')(\timeelapse{\aval},\apedge)} \; .
\end{eqnarray*}
We will see later in this proof that the use of the weighted average
for the choice of clock valuation by $\ansched$
is appropriate due to the fact that the clock dependencies featured in this paper are affine.

Finally, we consider the subcase in which $\timeelapse{\abint}$ is an open interval,
there exists $i \in \set{k,\ldots,n}$ such that $\apedge_i$ is non-constant,
and for all other $j \in \set{k,\ldots,i{-}1,i{+}1,\ldots,n}$ we have that $\apedge_i$ is constant,
and also $\apedge$ is constant.
This corresponds to the third phase described above.
For this subcase, the clock valuation can be arbitrary:
for simplicity we retain the same clock valuation as was used in the final state of $\afinpath$.
Formally, if $\last(\afinpath)$ is equal to $(\aloc',\aval')$,
we let $\timeelapsespecial{\abfinpath}{\timeelapse{\abint}}{\apedge} = \aval'$.

We repeat this process for all $(\timeelapse{\abint},\apedge) \in \intervalsof{\borders} \times \pedges$
such that $\sourceof{\apedge} = \aloc_n$,
$\timeelapse{\abint} \in \intervalsof{\borders}$,
$\timeelapse{\abint} \geq \abint_i$,
and
$\timeelapse{\abint} \vinz \guardof{\apedge} \wedge \inv(\aloc_n)$.
This suffices to define comprehensively the distribution $\ansched(\afinpath)$.
A formal justification for this fact now follows.
Let $\bactsset$ be the set of pairs $(\timeelapse{\abint},\apedge)$
such that $\sourceof{\apedge} = \aloc_n$,
$\timeelapse{\abint} \in \intervalsof{\borders}$,
$\timeelapse{\abint} \geq \abint_i$,
and
$\timeelapse{\abint} \vinz \guardof{\apedge} \wedge \inv(\aloc_n)$.
Now observe that,
for any $\afinpath' \in  \corrtwo{\asched}{\abfinpath}$,
we have:
\[
\sum_{(\timeelapse{\abint},\apedge) \in \bactsset} \;\;
\sum_{\timeelapse{\aval} \in \valnset{\abfinpath}{\timeelapse{\abint}}{\apedge}} \asched(\afinpath')(\timeelapse{\aval},\apedge)  = 1
\]
(because, for each $(\timeelapse{\aval},\apedge) \in \support(\asched(\afinpath'))$
there exists $(\timeelapse{\abint},\apedge) \in \bactsset$
such that $\timeelapse{\aval} \in \valnset{\abfinpath}{\timeelapse{\abint}}{\apedge}$).
We now show that
$\sum_{(\timeelapse{\abint},\apedge) \in \bactsset} \ansched(\afinpath)(\timeelapsespecial{\abfinpath}{\timeelapse{\abint}}{\apedge},\apedge) = 1$,
i.e., $\ansched(\afinpath)$ is a distribution:
\begin{eqnarray*}
\sum_{(\timeelapse{\abint},\apedge) \in \bactsset} \ansched(\afinpath)(\timeelapsespecial{\abfinpath}{\timeelapse{\abint}}{\apedge},\apedge)
& = &
\sum_{(\timeelapse{\abint},\apedge) \in \bactsset} \left(
\frac{\sum_{\afinpath' \in \corrtwo{\asched}{\abfinpath}} \probfin{\asched}{(\aloc,\aval)}(\afinpath')
\cdot \sum_{\timeelapse{\aval} \in \valnset{\abfinpath}{\timeelapse{\abint}}{\apedge}} \asched(\afinpath')(\timeelapse{\aval},\apedge) }
{\probfin{\asched}{(\aloc,\aval)}(\corrtwo{\asched}{\abfinpath})} \right)
\\
& = &
\frac{\sum_{\afinpath' \in \corrtwo{\asched}{\abfinpath}} \probfin{\asched}{(\aloc,\aval)}(\afinpath')
\cdot
\sum_{(\timeelapse{\abint},\apedge) \in \bactsset}
\sum_{\timeelapse{\aval} \in \valnset{\abfinpath}{\timeelapse{\abint}}{\apedge}} \asched(\afinpath')(\timeelapse{\aval},\apedge) }
{\probfin{\asched}{(\aloc,\aval)}(\corrtwo{\asched}{\abfinpath})}
\\
& = &
\frac{\sum_{\afinpath' \in \corrtwo{\asched}{\abfinpath}} \probfin{\asched}{(\aloc,\aval)}(\afinpath') }
{\probfin{\asched}{(\aloc,\aval)}(\corrtwo{\asched}{\abfinpath})}
\\
& = &
1 \; .
\end{eqnarray*}

We now proceed to show that
$\probl{\asched}{(\aloc,\aval)}(\Diamond \targetsetof{\targetlocs})
= \probl{\ansched}{(\aloc,\aval)}(\Diamond \targetsetof{\targetlocs})$.
It suffices to show that
$\probfin{\asched}{(\aloc,\aval)}(\corrtwo{\asched}{\abfinpath})
= \probfin{\ansched}{(\aloc,\aval)}(\corrtwo{\ansched}{\abfinpath})$
for all $\borders$-paths $\abfinpath$
(because all paths in $\corrtwo{\asched}{\abfinpath}$ reach a location in $\targetlocs$
if and only if the unique path in $\corrtwo{\ansched}{\abfinpath}$ reaches a location in $\targetlocs$).
We proceed by induction on the length of $\borders$-paths.
Let $\abfinpath$ be a $\borders$-path,
and assume that we have established that $\probfin{\asched}{(\aloc,\aval)}(\corrtwo{\asched}{\abfinpath})
= \probfin{\ansched}{(\aloc,\aval)}(\corrtwo{\ansched}{\abfinpath})$.
Now consider the $\borders$-path $\abfinpath (\timeelapse{\abint},\apedge) (\aloc',\abint')$
that extends $\abfinpath$ with one transition.
Our aim is to show that
$\probfin{\asched}{(\aloc,\aval)}(\corrtwo{\asched}{\abfinpath (\timeelapse{\abint},\apedge) (\aloc',\abint')})
= \probfin{\ansched}{(\aloc,\aval)}(\corrtwo{\ansched}{\abfinpath (\timeelapse{\abint},\apedge) (\aloc',\abint')})$.
In the sequel,
we overload the notation $\mfromtrans{\acdpta}$ in the following way:
given $\timeelapse{\aval} \in \valnset{\abfinpath}{\timeelapse{\abint}}{\apedge}$, we let
$\mfromtrans{\acdpta}(\last(\afinpath'),(\timeelapse{\aval},\apedge))(\aloc',\abint') = 0$
if $\mfromtrans{\acdpta}(\last(\afinpath'),(\timeelapse{\aval},\apedge))(\aloc',\aval') = 0$
for all $\aval' \in \abint'$,
and $\mfromtrans{\acdpta}(\last(\afinpath'),(\timeelapse{\aval},\apedge))(\aloc',\abint') =
\mfromtrans{\acdpta}(\last(\afinpath'),(\timeelapse{\aval},\apedge))(\aloc',\aval')$
for the unique clock valuation $\aval' \in \abint'$
such that $\mfromtrans{\acdpta}(\last(\afinpath'),(\timeelapse{\aval},\apedge))(\aloc',\aval') > 0$
(note that uniqueness follows from the fact that either $\aval' = 0$ or $\aval' = \timeelapse{\aval}$).
In the following we consider the most involved case,
namely that concerning $\timeelapse{\abint}$ being open and $\apedge$ being non-constant,
i.e., the second phase of the case of open $\timeelapse{\abint}$ described above.
Assume that $\mfromtrans{\acdpta}(\last(\afinpath'),(\timeelapse{\aval},\apedge))(\aloc',\abint') > 0$,
i.e., there is a clock valuation $\aval' \in \abint'$ such that
$\mfromtrans{\acdpta}(\last(\afinpath'),(\timeelapse{\aval},\apedge))(\aloc',\aval') > 0$.
Furthermore, assume that $\aval' = \timeelapse{\aval}$,
which means that $\mfromtrans{\acdpta}(\last(\afinpath'),(\timeelapse{\aval},\apedge))(\aloc',\abint')
= \adtpara{\pd}{\timeelapse{\aval}}(\emptyset,\aloc')$
(the case in which $\abint'=[0,0]$,
and hence $\aval' = 0$ and $\mfromtrans{\acdpta}(\last(\afinpath'),(\timeelapse{\aval},\apedge))(\aloc',\abint')
= \adtpara{\pd}{\timeelapse{\aval}}(\set{\aclock},\aloc')$,
is similar),
where we use $\adisttemp$ to denote the distribution template of probabilistic edge $\apedge$.
Let $(\fc{\apedge}{(\emptyset,\aloc')}, \fd{\apedge}{(\emptyset,\aloc')}) \in \Qset^2$ be the pair associated
with $\apedge$ and $(\emptyset,\aloc')$ defining the clock dependency;
in order to simplify notation in the sequel,
we let $\fcno = \fc{\apedge}{(\emptyset,\aloc')}$ and $\fdno = \fd{\apedge}{(\emptyset,\aloc')}$.
Furthermore, we let $\corrtwo{\ansched}{\abfinpath} = \afinpath$,
and write $\probfinone{\asched}$ and $\probfinone{\ansched}$
rather than $\probfin{\asched}{(\aloc,\aval)}$ and $\probfin{\ansched}{(\aloc,\aval)}$, respectively.
Then we have:
\begin{longtable}{RCL}
& &
\multicolumn{1}{p{0.9\textwidth}}{
$\probfinone{\asched}(\corrtwo{\asched}{\abfinpath (\timeelapse{\abint},\apedge) (\aloc',\abint')})$
}
\\
& = &
\sum_{\afinpath' \in \corrtwo{\asched}{\abfinpath (\timeelapse{\abint},\apedge) (\aloc',\abint')}}
\probfinone{\asched}(\afinpath')
\\
\multicolumn{3}{r}{
\mbox{(by definition of $\corrtwo{\asched}{\abfinpath (\timeelapse{\abint},\apedge) (\aloc',\abint')}$)}
}
\\
& = &
\sum_{\afinpath' \in \corrtwo{\asched}{\abfinpath}}
\sum_{\timeelapse{\aval} \in \valnset{\abfinpath}{\timeelapse{\abint}}{\apedge}}
\probfinone{\asched}(\afinpath')
\cdot
\asched(\afinpath')(\timeelapse{\aval},\apedge)
\cdot
\mfromtrans{\acdpta}(\last(\afinpath'),(\timeelapse{\aval},\apedge))(\aloc',\abint')
\\
\multicolumn{3}{r}{
\mbox{(by definition of $\probfinone{\asched}$)}
}
\\
& = &
\sum_{\afinpath' \in \corrtwo{\asched}{\abfinpath}}
\probfinone{\asched}(\afinpath')
\cdot
\sum_{\timeelapse{\aval} \in \valnset{\abfinpath}{\timeelapse{\abint}}{\apedge}}
\asched(\afinpath')(\timeelapse{\aval},\apedge)
\cdot
\adtpara{\pd}{\timeelapse{\aval}}(\emptyset,\aloc')
\hspace{4cm}
\\
\multicolumn{3}{r}{
\mbox{(rearranging and by definition of $\mfromtrans{\acdpta}$)}
}
\\
& = &
\sum_{\afinpath' \in \corrtwo{\asched}{\abfinpath}}
\probfinone{\asched}(\afinpath')
\cdot
\sum_{\timeelapse{\aval} \in \valnset{\abfinpath}{\timeelapse{\abint}}{\apedge}}
\asched(\afinpath')(\timeelapse{\aval},\apedge)
\cdot
(\fcno + \fdno \cdot \timeelapse{\aval})
\\
\multicolumn{3}{r}{
\mbox{(by definition of non-constant probabilistic edges)}
}
\\
& = &
\left(\sum_{\afinpath' \in \corrtwo{\asched}{\abfinpath}}
\probfinone{\asched}(\afinpath')
\cdot
\sum_{\timeelapse{\aval} \in \valnset{\abfinpath}{\timeelapse{\abint}}{\apedge}}
\asched(\afinpath')(\timeelapse{\aval},\apedge)
\right)
\\
& &
\times
\left(\fcno
+
\fdno \frac{
\sum_{\afinpath' \in \corrtwo{\asched}{\abfinpath}}
\probfinone{\asched}(\afinpath')
\cdot
\sum_{\timeelapse{\aval} \in \valnset{\abfinpath}{\timeelapse{\abint}}{\apedge}}
\asched(\afinpath')(\timeelapse{\aval},\apedge)
\cdot
\timeelapse{\aval}
}
{\sum_{\afinpath' \in \corrtwo{\asched}{\abfinpath}}
\probfinone{\asched}(\afinpath')
\cdot
\sum_{\timeelapse{\aval} \in \valnset{\abfinpath}{\timeelapse{\abint}}{\apedge}}
\asched(\afinpath')(\timeelapse{\aval},\apedge)}
\right)
\hspace{4cm}
\\
\multicolumn{3}{r}{
\mbox{(rearranging)}
}
\\
& = &
\left(\sum_{\afinpath' \in \corrtwo{\asched}{\abfinpath}}
\probfinone{\asched}(\afinpath')
\cdot
\sum_{\timeelapse{\aval} \in \valnset{\abfinpath}{\timeelapse{\abint}}{\apedge}}
\asched(\afinpath')(\timeelapse{\aval},\apedge)
\right)
(\fcno + \fdno \cdot \timeelapsespecial{\abfinpath}{\timeelapse{\abint}}{\apedge})
\\
\multicolumn{3}{r}{
\mbox{(by definition of $\timeelapsespecial{\abfinpath}{\timeelapse{\abint}}{\apedge}$)}
}
\\
& = &
\left(\sum_{\afinpath' \in \corrtwo{\asched}{\abfinpath}}
\probfinone{\asched}(\afinpath')
\cdot
\sum_{\timeelapse{\aval} \in \valnset{\abfinpath}{\timeelapse{\abint}}{\apedge}}
\asched(\afinpath')(\timeelapse{\aval},\apedge)
\right)
\adtpara{\pd}{\timeelapsespecial{\abfinpath}{\timeelapse{\abint}}{\apedge}}(\emptyset,\aloc')
\\
\multicolumn{3}{r}{
\mbox{(by definition of non-constant probabilistic edges)}
}
\\
& = &
\probfinone{\asched}(\corrtwo{\asched}{\abfinpath})
\cdot
\ansched(\afinpath)(\timeelapsespecial{\abfinpath}{\timeelapse{\abint}}{\apedge},\apedge)
\cdot
\mfromtrans{\acdpta}(\last(\afinpath),(\timeelapsespecial{\abfinpath}{\timeelapse{\abint}}{\apedge},\apedge))(\aloc',\abint')
\hspace{2cm}
\\
\multicolumn{3}{r}{
\mbox{(from the construction of $\ansched$ and by the overloaded definition of $\mfromtrans{\acdpta}$)}
}
\\
& = &
\probfinone{\ansched}(\afinpath)
\cdot
\ansched(\afinpath)(\timeelapsespecial{\abfinpath}{\timeelapse{\abint}}{\apedge},\apedge)
\cdot
\mfromtrans{\acdpta}(\last(\afinpath),(\timeelapsespecial{\abfinpath}{\timeelapse{\abint}}{\apedge},\apedge))(\aloc',\abint')
\\
\multicolumn{3}{r}{
\mbox{(by induction)}
}
\\
& = &
\probfinone{\ansched}(\corrtwo{\ansched}{\abfinpath}  (\timeelapse{\abint},\apedge) (\aloc',\abint'))
\\
\multicolumn{3}{r}{
\mbox{(by definition of $\probfinone{\ansched}$).}
}
\end{longtable}
%
Observe in particular the fact that we use the property that the clock dependencies are affine,
together with the fact that the clock valuation
$\timeelapsespecial{\abfinpath}{\timeelapse{\abint}}{\apedge}$
is a weighted average of the choices of clock valuations made by $\asched$,
to obtain the fourth to seventh steps.

We remark briefly that the cases for the first and third phases are simpler,
because the probabilities $\adtpara{\pd}{\cdot}(\emptyset,\aloc')$ and $\adtpara{\pd}{\cdot}(\set{\aclock},\aloc')$
are constant
(this means that the choice of clock valuation made by $\ansched$,
as long as it is in the open region denoted by $\timeelapse{\abint}$ above,
is arbitrary).
The case in which $\timeelapse{\abint}$ is closed is even more straightforward,
because there is only \emph{one} choice of clock valuation
when considering choices of $\asched$ and $\ansched$ that correspond to $\abfinpath$
extended with $(\timeelapse{\abint},\apedge)$.
\end{proof}

\subsection{IMDP Construction.}
We now present the idea of the IMDP construction.
The states of the IMDP fall into two categories:
(1)~pairs comprising a location and an interval from $\intervalsof{\borders}$,
with the intuition that the state $(\aloc,\abint) \in \locs \times \intervalsof{\borders}$ of the IMDP
represents all states $(\aloc,\aval)$ of $\semmdp{\acdpta}$ such that $\aval \in \abint$;
(2)~triples comprising an interval from $\intervalsof{\borders}$,
a probabilistic edge and a bit that specifies whether the state refers to the left- or right-endpoint of the interval.
A single transition of the semantics of the 1c-cdPTA,
which we recall represents the elapse of time (therefore increasing the value of the clock)
followed by the traversal of a probabilistic edge,
is represented by a sequence of \emph{two} transitions in the IMDP\@.
The first IMDP transition in the sequence represents the choice of
(i)~the probabilistic edge,
(ii)~the interval in $\intervalsof{\borders}$ which contains the valuation of the clock
after letting time elapse and
immediately before the probabilistic edge is traversed,
and (iii)~in the case in which the aforementioned interval is open,
an endpoint of the interval chosen in (ii).
The second IMDP transition in the sequence represents the probabilistic choice made
according to the chosen probabilistic edge, interval and endpoint chosen in the first transition of the sequence.
%


\begin{figure}[t]
\centering

\tiny

\begin{tikzpicture}[->,>=stealth',shorten >=1pt,auto, thin] 

  \tikzstyle{state}=[draw=black, text=black, shape=rectangle, inner sep=3pt, outer sep=0pt, rectangle, rounded corners] 
  \tikzstyle{delay}=[draw=black, text=black, shape=rectangle, inner sep=3pt, outer sep=0pt, rectangle, rounded corners,fill=gray!20] 
  \tikzstyle{final_state}=[draw=black, text=black, shape=rectangle, inner sep=3pt, outer sep=0pt, rectangle, rounded corners, fill= lg] 


	\node at (5,5)[state](W){$(\mathrm{W},[0,0])$};

	\node at (5,4)[fill=black] (nail_W) {};

	\node at (3.5,3)[delay] (left_W) {$((1,3),\apedge_\mathrm{W},\lb)$};

	\node at (6.5,3)[delay] (right_W) {$((1,3),\apedge_\mathrm{W},\rb)$};

	\node at (2,2)[state](S) {$(\mathrm{S},(1,3))$};

	\node at (5,0)[state](Tleft) {$(\mathrm{T},(1,3))$};

	\node at (8,2)[state](F) {$(\mathrm{F},(1,3))$};

	\node at (10,2)[fill=black] (nail_F) {};

	\node at (12,3)[delay] (left_F) {$((4,5),\apedge_\mathrm{W},\lb)$};

	\node at (12,1)[delay] (right_F) {$((4,5),\apedge_\mathrm{W},\rb)$};

	\node at (12,0)[state](Tright) {$(\mathrm{T},(4,5))$};

	\path[rounded corners,dashed]
	(left_W)
		edge [left] node {$[0,0]$} (S);

	\path[rounded corners]

	(W)
		edge [left] node[yshift=0mm] {$((1,3),\apedge_\mathrm{W})$} (nail_W)

	(nail_W)
		edge [left] node {$(0,1)$} (left_W)
		edge [right] node {$(0,1)$} (right_W)

	(left_W)
		edge [left] node {$[\frac{1}{2},\frac{1}{2}]$} (Tleft)
		edge [below,pos=0.8] node {$[\frac{1}{2},\frac{1}{2}]$} (F)

	(right_W)
		edge [below,pos=0.8] node {$[\frac{3}{4},\frac{3}{4}]$} (S)
		edge [right] node {$[\frac{1}{8},\frac{1}{8}]$} (Tleft)
		edge [right] node {$[\frac{1}{8},\frac{1}{8}]$} (F)

	(F)
		edge [above] node[yshift=0mm] {$((4,5),\apedge_\mathrm{F})$} (nail_F)

	(nail_F)
		edge [right] node {$(0,1)$} (left_F)
		edge [right] node {$(0,1)$} (right_F)

	(right_F)
		edge [left] node {$[\frac{1}{2},\frac{1}{2}]$} (Tright);

\draw[->,rounded corners,dashed] (left_F.north) -- (12,5)node[left,text=black,pos=.7]{$[0,0]$} -- (W.east);
\draw[->,rounded corners] (left_F.east) -- (13.2,3) -- (13.2,0) node[left,text=black,pos=.35]{$[1,1]$} -- (Tright.east);
\draw[->,rounded corners] (right_F.west) -- (10,1) -- (10,-0.5) -- (13.7,-0.5) -- node[right,text=black,pos=.7]{$[\frac{1}{2},\frac{1}{2}]$}(13.7,5.5) -- (5,5.5) -- (W.north);

\end{tikzpicture}

\caption{Interval Markov decision process $\oimdpfrom{\acdpta}$ obtained from $\acdpta$ of \figref{fig:one_clock}.}%
\label{fig:imdp}

\end{figure}


\begin{exa}\label{ex:imdp}
The IMDP construction, applied to the example of \figref{fig:one_clock},
is shown in \figref{fig:imdp}
(note that transitions corresponding to probability $0$ are shown with a dashed line).
The location $\mathrm{W}$, and the value of the clock being $0$,
is represented by the state $(\mathrm{W},[0,0])$.
Recall that the outgoing probabilistic edge from $\mathrm{W}$
is enabled when the clock is between $1$ and $3$:
hence the single action $((1,3),\apedge_\mathrm{W})$ is available from $(\mathrm{W},[0,0])$
(representing the set of actions $(\timeelapse{\aval},\apedge_\mathrm{W})$ of $\semmdp{\acdpta}$
with $\timeelapse{\aval} \in (1,3)$).
The action $((1,3),\apedge_\mathrm{W})$ is associated with two target states,
$((1,3),\apedge_\mathrm{W},\lb)$ and $((1,3),\apedge_\mathrm{W},\rb)$,
where the third component of the state ($\lb$ or $\rb$)
denotes whether the state refers to the left- or right-endpoint of $(1,3)$.
Each of the states $((1,3),\apedge_\mathrm{W},\lb)$ and $((1,3),\apedge_\mathrm{W},\rb)$
is associated with the probability interval $(0,1)$,
referring to the probability of making a transition to those states from $(\mathrm{W},[0,0])$
with action $((1,3),\apedge_\mathrm{W})$.
The choice of probability within the interval is done in the IMDP
to represent a choice of clock valuation in $(1,3)$:
for example, the clock valuation $\frac{3}{2}$ is represented by the assignment
that associates probability $\frac{3}{4}$ with $((1,3),\apedge_\mathrm{W},\lb)$
and $\frac{1}{4}$ with $((1,3),\apedge_\mathrm{W},\rb)$
(i.e., assigns a weight of $\frac{3}{4}$ to the lower bound of $(1,3)$,
and a weight of $\frac{1}{4}$ to the upper bound of $(1,3)$,
obtaining the weighted combination $\frac{3}{4} \cdot 1 + \frac{1}{4} \cdot 3 = \frac{3}{2}$).
Then, from both $((1,3),\apedge_\mathrm{W},\lb)$ and $((1,3),\apedge_\mathrm{W},\rb)$,
there is a probabilistic choice regarding the target IMDP state to make the subsequent transition to,
i.e., the transitions from $((1,3),\apedge_\mathrm{W},\lb)$ and $((1,3),\apedge_\mathrm{W},\rb)$
do not involve nondeterminism,
because there is only one action available, and because the resulting interval distribution assigns singleton intervals
to all possible target states.\footnote{Given that there is only one action available from states such as
$((1,3),\apedge_\mathrm{W},\lb)$ and $((1,3),\apedge_\mathrm{W},\rb)$,
we omit both the action and the usual black box from the figure.}
The probabilities of the transitions from $((1,3),\apedge_\mathrm{W},\lb)$ and $((1,3),\apedge_\mathrm{W},\rb)$
are derived from the clock dependencies associated with $1$ (i.e., the left endpoint of $(1,3)$)
and $3$ (i.e., the right endpoint of $(1,3)$), respectively.
Hence the multiplication of the probabilities of the two aforementioned transitions
(from $(\mathrm{W},[0,0])$ to either $((1,3),\apedge_\mathrm{W},\lb)$ or $((1,3),\apedge_\mathrm{W},\rb)$,
and then to $(\mathrm{S},(1,3))$, $(\mathrm{T},(1,3))$ or $(\mathrm{F},(1,3))$)
represents exactly the probability of a single transition in the 1c-cdPTA\@.
For example,
in the 1c-cdPTA, considering again the example of the clock valuation associating $\frac{3}{2}$ with $\aclock$,
the probability of making a transition to location $\mathrm{S}$ is $\frac{3\aclock - 3}{8} = \frac{3}{16}$;
in the IMDP, assigning $\frac{3}{4}$ to the transition to $((1,3),\apedge_\mathrm{W},\lb)$
and $\frac{1}{4}$ to the transition to $((1,3),\apedge_\mathrm{W},\rb)$,
we then obtain that the probability of making a transition to $(\mathrm{S},(1,3))$ from $(\mathrm{W},[0,0])$
is $\frac{3}{4} \cdot 0 + \frac{1}{4} \cdot \frac{3}{4} = \frac{3}{16}$.
Similar reasoning applies to the transitions available from $(\mathrm{F},(1,3))$.
\end{exa}

We now describe formally the construction of the IMDP
$\oimdpfrom{\acdpta} =
(\omfromstates{\oimdpfrom{\acdpta}},\omfromactions{\oimdpfrom{\acdpta}},\omfromtrans{\oimdpfrom{\acdpta}})$.
The set of states of $\oimdpfrom{\acdpta}$ is defined as
$\omfromstates{\oimdpfrom{\acdpta}} = \regstates \cup \eis$,
where:
\begin{eqnarray*}
\regstates & = & \set{(\aloc,\abint) \in \locs \times \intervalsof{\borders} \setsep \abint \vinz \inv(\aloc)}
\\
\eis & = &
\set{(\timeelapse{\abint},(\aloc,\g,\pd),\dir) \in \intervalsof{\borders} \times \pedges \times \set{\lb,\rb}
\setsep
\timeelapse{\abint} \vinz \g \wedge \inv(\aloc)} \; .
\end{eqnarray*}
In order to distinguish states of $\semmdp{\acdpta}$ and states of $\oimdpfrom{\acdpta}$,
we refer to elements of
$\regstates$ as \emph{regions},
and elements of
$\eis$
as \emph{endpoint indicators}.
The set of actions of $\oimdpfrom{\acdpta}$ is defined as
\[
\omfromactions{\oimdpfrom{\acdpta}} = \set{(\timeelapse{\abint},(\aloc,\g,\pd)) \in \intervalsof{\borders} \times \pedges
\setsep
\timeelapse{\abint} \vinz \g \wedge \inv(\aloc)}
\cup \set{ \dummyact }
\]
(i.e., there is an action for each combination of interval from $\intervalsof{\borders}$
and probabilistic edge such that all valuations from the interval satisfy
both the guard of the probabilistic edge
and the invariant condition of its source location).
For each region $(\aloc,\abint) \in \regstates$,
let $\omfromactionsof{\oimdpfrom{\acdpta}}{\aloc,\abint} =
\set{(\timeelapse{\abint},(\aloc',\g,\pd)) \in \omfromactions{\oimdpfrom{\acdpta}} \setsep \aloc = \aloc' \mbox{ and } \timeelapse{\abint} \geq \abint}$.\footnote{Note that
$\omfromactionsof{\oimdpfrom{\acdpta}}{\aloc,\abint} \neq \emptyset$
for each $(\aloc,\abint) \in \regstates$,
by the assumptions that we made on 1c-cdPTA in \sectref{sec:onec-cdpta}
(namely, that it is always possible to take a probabilistic edge,
either immediately or after letting time elapse).}
For each $(\timeelapse{\abint},\apedge,\dir) \in \eis$,
let $\omfromactionsof{\oimdpfrom{\acdpta}}{\timeelapse{\abint},\apedge,\dir} = \set{ \dummyact }$.
The transition function $\omfromtrans{\oimdpfrom{\acdpta}}: \omfromstates{\oimdpfrom{\acdpta}} \times \omfromactions{\oimdpfrom{\acdpta}} \ra \intdist(\omfromstates{\oimdpfrom{\acdpta}}) \cup \set{ \noaction }$
is defined as follows:
\begin{itemize}
%
\item
For each $(\aloc,\abint) \in \regstates$ and $(\timeelapse{\abint},\apedge) \in \omfromactionsof{\oimdpfrom{\acdpta}}{\aloc,\abint}$,
let
$\omfromtrans{\oimdpfrom{\acdpta}}((\aloc,\abint),(\timeelapse{\abint},\apedge))$
be the interval distribution such that
(1)~$\omfromtrans{\oimdpfrom{\acdpta}}((\aloc,\abint),(\timeelapse{\abint},\apedge))
(\timeelapse{\abint},\apedge,\dir)
=(0,1)$
for $\dir \in \set{\lb, \rb}$,
and (2)~$\omfromtrans{\oimdpfrom{\acdpta}}((\aloc,\abint),(\timeelapse{\abint},\apedge))(\astate)
= [0,0]$ for all
$\astate \in \omfromstates{\oimdpfrom{\acdpta}} \setminus \set{(\timeelapse{\abint},\apedge,\lb),(\timeelapse{\abint},\apedge,\rb)}$.

\item
For each $(\timeelapse{\abint},(\aloc,\g,\pd),\dir) \in \eis$
and $(\aloc',\abint') \in \regstates$,
let:
\[
\probof{(\timeelapse{\abint},(\aloc,\g,\pd),\dir)}{(\aloc',\abint')} =
\left\{
\begin{array}{ll}
\adtpara{\pd}{\dir(\timeelapse{\abint})}(\set{\aclock},\aloc')
+
\adtpara{\pd}{\dir(\timeelapse{\abint})}(\emptyset,\aloc')
&
\mbox{if } \abint' = \timeelapse{\abint} = [0,0]
\\
\adtpara{\pd}{\dir(\timeelapse{\abint})}(\emptyset,\aloc')
&
\mbox{if } \abint' = \timeelapse{\abint} > [0,0]
\\
\adtpara{\pd}{\dir(\timeelapse{\abint})}(\set{\aclock},\aloc')
&
\mbox{if } \abint' = [0,0] \mbox{ and } \timeelapse{\abint} > [0,0]
\\
0
&
\mbox{otherwise.}
\end{array}
\right.
\]
Then $\omfromtrans{\oimdpfrom{\acdpta}}((\timeelapse{\abint},(\aloc,\g,\pd),\dir),\dummyact)$ is the interval distribution
such that, for all $\astate \in \omfromstates{\oimdpfrom{\acdpta}}$:
\[
\omfromtrans{\oimdpfrom{\acdpta}}((\timeelapse{\abint},(\aloc,\g,\pd),\dir),\dummyact)(\astate) =
\left\{
\begin{array}{ll}
[\probof{(\timeelapse{\abint},(\aloc,\g,\pd),\dir)}{\astate},\probof{(\timeelapse{\abint},(\aloc,\g,\pd),\dir)}{\astate}]
&
\mbox{if } \astate \in \regstates
\\
\, \! [0,0]
&
\mbox{otherwise.}
\end{array}
\right.
\]
\end{itemize}
We recall that $\omfromtrans{\oimdpfrom{\acdpta}}(\astate,\anomact) = \noaction$
for $\astate \in \omfromstates{\oimdpfrom{\acdpta}}$ and $\anomact \in \omfromactions{\oimdpfrom{\acdpta}} \setminus \omfromactionsof{\oimdpfrom{\acdpta}}{\astate}$.

\subsection{Correctness of the IMDP construction.}
Next, we establish the correctness of the construction of $\oimdpfrom{\acdpta}$,
i.e., that $\oimdpfrom{\acdpta}$ can be used for solving quantitative and qualitative problems of the 1c-cdPTA $\acdpta$.
The proof relies on showing that a transition of the semantic MDP $\semmdp{\acdpta}$ of $\acdpta$
can be mimicked by a sequence of \emph{two} transitions
of the semantic MDP $\sem{\oimdpfrom{\acdpta}}$ of $\oimdpfrom{\acdpta}$,
and vice versa.
Let $\semmdp{\acdpta} = (\mfromstates{\acdpta}, \mfromactions{\acdpta}, \mfromtrans{\acdpta})$
be the semantic MDP of $\acdpta$.
Given state $(\aloc,\aval) \in \mfromstates{\acdpta}$,
we let $\regof{\aloc,\aval} = (\aloc,\abint) \in \regstates$
be the unique region such that $\aval \in \abint$.
In the following, we let
$\sem{\oimdpfrom{\acdpta}} =
(\omfromstates{\oimdpfrom{\acdpta}},\mfromactions{\oimdpfrom{\acdpta}},\mfromtrans{\oimdpfrom{\acdpta}})$
be the semantic MDP of $\oimdpfrom{\acdpta}$.

We now show that,
for any scheduler of (the semantics of) the 1c-cdPTA $\acdpta$,
there exists a scheduler of (the semantics of) the IMDP $\oimdpfrom{\acdpta}$
such that the schedulers assign the same probability
to reaching a certain set of locations
from a given location with the value of the clock equal to $0$.
Let $\targetsetimdpof{\targetlocs} = \set{ (\aloc,\abint) \in \regstates \setsep \aloc \in \targetlocs }$
be the set of regions with location component in $\targetlocs$.

\begin{lem}\label{lem:sched_cdPTA2IMDP}
Let $\aloc \in \locs$ be a location
and let $\targetlocs \subseteq \locs$ be a set of locations.
Given a $\borders$-minimal scheduler $\ansched \in \bstrats{\semmdp{\acdpta}}$,
there exists a scheduler $\anosched \in \strats{\sem{\oimdpfrom{\acdpta}}}$
such that $\probl{\ansched}{(\aloc,0)}(\Diamond \targetsetof{\targetlocs}) =
\probl{\anosched}{\regof{\aloc,0}}(\Diamond \targetsetimdpof{\targetlocs})$.
\end{lem}
%
%
Before presenting the formal proof of \lemref{lem:sched_cdPTA2IMDP},
we sketch the overall approach for the construction of a scheduler $\anosched$ of $\sem{\oimdpfrom{\acdpta}}$
from a $\borders$-minimal scheduler $\ansched$ of $\semmdp{\acdpta}$,
and give an example.
For each finite path $\afinpath$ of $\ansched$, we can identify a set of finite paths of $\anosched$
of length twice that of $\afinpath$, which visit the same locations in order,
choose the same probabilistic edges in order,
and visit the same intervals in order, both regarding the clock valuations/intervals in states and in actions;
in fact, finite paths of $\anosched$ that are
associated with $\afinpath$ differ \emph{only} in terms of the $\lb$ and $\rb$ components used in endpoint indicators.
Furthermore, $\anosched$ replicates exactly the choice of $\ansched$ made after $\afinpath$
in terms of interval of $\intervalsof{\borders}$ and probabilistic edge chosen
in all of its finite paths associated with $\afinpath$.
Finally, $\anosched$ chooses assignments (over edges labelled with $(0,1)$)
in order to represent exactly the choices of clock valuations made by $\ansched$,
in the manner described in Example~\ref{ex:imdp} above:
more precisely, the choice of action $(\timeelapse{\aval},\apedge)$ by $\ansched$,
where $\timeelapse{\abint}$ is the unique open interval such that $\timeelapse{\aval} \in \timeelapse{\abint}$,
is mimicked by $\anosched$ choosing the action $((\timeelapse{\abint},\apedge),\ana)$
for which $\ana(\timeelapse{\abint},\apedge,\lb) =
\frac{\rb(\timeelapse{\abint}) - \timeelapse{\aval}}{\rb(\timeelapse{\abint}) - \lb(\timeelapse{\abint})}$,
and $\ana(\timeelapse{\abint},\apedge,\rb) = 1 - \ana(\timeelapse{\abint},\apedge,\lb)
= \frac{\timeelapse{\aval} - \lb(\timeelapse{\abint})}{\rb(\timeelapse{\abint}) - \lb(\timeelapse{\abint})}$.

\begin{exa}\label{ex:sched_cdPTA2IMDP}
Consider the 1c-cdPTA of \figref{fig:one_clock}.
Let $\ansched \in \bstrats{\semmdp{\acdpta}}$ be a $\borders$-minimal scheduler such that $\ansched(\mathrm{W},0)$
assigns probability $1$ to the action $(\frac{3}{2},\apedge_\mathrm{W})$.
Then $\anosched \in \strats{\sem{\oimdpfrom{\acdpta}}}$ is constructed such that
$\anosched(\mathrm{W},[0,0])$ assigns probability $1$ to
$(((1,3),\apedge_\mathrm{W}),\ana)$,
where $\ana((1,3),\apedge_\mathrm{W},\lb) = \frac{3}{4}$
and $\ana((1,3),\apedge_\mathrm{W},\rb) = \frac{1}{4}$
(observe that $\ana((1,3),\apedge_\mathrm{W},\lb) = \frac{3-\frac{3}{2}}{2}$
and that $\ana((1,3),\apedge_\mathrm{W},\rb) = 1 - \ana((1,3),\apedge_\mathrm{W},\lb) = \frac{\frac{3}{2}-1}{2}$).
Furthermore, $\anosched((\mathrm{W},[0,0])(((1,3),\apedge_\mathrm{W}),\ana))$
assigns probability $1$ to $\rsuccact$.
Now consider the finite path $\afinpath =
(\mathrm{W},0)(\frac{3}{2},\apedge_\mathrm{W})(\mathrm{F},\frac{3}{2})$ of $\ansched$:
then the corresponding set of finite paths of $\anosched$ comprises
$\afinpath' =
(\mathrm{W},[0,0]) (((1,3),\apedge_\mathrm{W}),\ana)
((1,3),\apedge_\mathrm{W},\lb) (\mathrm{F},(1,3))$
and
$\afinpath'' =
(\mathrm{W},[0,0]) (((1,3),\apedge_\mathrm{W}),\ana)
((1,3),\apedge_\mathrm{W},\rb) (\mathrm{F},(1,3))$.
Now say that $\ansched(\afinpath)$ assigns probability $1$
to the action $(\frac{9}{2},\apedge_\mathrm{F})$:
then both $\anosched(\afinpath')$ \emph{and} $\anosched(\afinpath'')$
assign probability $1$ to the action $(((4,5),\apedge_\mathrm{F}),\ana')$,
where $\ana'((4,5),\apedge_\mathrm{F},\lb) = \frac{1}{2}$
and $\ana'((4,5),\apedge_\mathrm{F},\rb) = \frac{1}{2}$
(note that $\ana'((4,5),\apedge_\mathrm{F},\lb) = 5-\frac{9}{2}$
and $\ana'((4,5),\apedge_\mathrm{F},\rb) = 1 - \ana'((4,5),\apedge_\mathrm{F},\lb) = \frac{9}{2}-4$).
Hence, regardless of whether $((1,3),\apedge_\mathrm{W},\lb)$
or $((1,3),\apedge_\mathrm{W},\rb)$ was visited,
scheduler $\anosched$ makes the \emph{same} choice to mimic $\ansched(\afinpath)$.
\end{exa}

\begin{proof}[Proof of \lemref{lem:sched_cdPTA2IMDP}.]
Let $\aloc \in \locs$,
$\targetlocs \subseteq \locs$
and $\ansched \in \bstrats{\semmdp{\acdpta}}$.
We proceed by describing the construction of scheduler $\anosched \in \strats{\sem{\oimdpfrom{\acdpta}}}$,
then show that $\probl{\ansched}{(\aloc,0)}(\Diamond \targetsetof{\targetlocs}) =
\probl{\anosched}{\regof{\aloc,0}}(\Diamond \targetsetimdpof{\targetlocs})$.

Before describing the construction of $\anosched \in \strats{\sem{\oimdpfrom{\acdpta}}}$,
we first simplify the notation for transitions along paths of $\sem{\oimdpfrom{\acdpta}}$
to remove redundant elements,
in the following way:
a sequence
$(\aloc,\abint) ((\timeelapse{\abint},\apedge),\ana) (\timeelapse{\abint},\apedge,\dir)
(\dummyact,\anab) (\aloc',\abint')$
of two transitions between regions
will be simplified to
$(\aloc,\abint) ((\timeelapse{\abint},\apedge),\ana) (\dir) (\aloc',\abint')$
(note that, for the endpoint indicator $(\timeelapse{\abint},\apedge,\dir)$,
the interval $\timeelapse{\abint}$ and the probabilistic edge $\apedge$
have been featured in the previous position in the sequence,
and that, from $(\timeelapse{\abint},\apedge,\dir)$,
only one action is available, namely $(\dummyact,\anab)$,
and hence we can omit it from the sequence).

Recall that, from the proof of \lemref{lem:bminimal},
each finite path $\afinpath \in \mfinpaths{\ansched}$ of $\ansched$ corresponds to a $\borders$-path.
Furthermore, for a given $\borders$-path, we can identify a set of finite paths of $\sem{\oimdpfrom{\acdpta}}$
that agree with the $\borders$-path in terms of locations, intervals, and probabilistic edges.
Formally, for a $\borders$-path $\abfinpath$ which equals
\[
(\aloc_0,\abint_0)(\timeelapse{\abint}_0,\apedge_0)
(\aloc_1,\abint_1)(\timeelapse{\abint}_1,\apedge_1)
\cdots (\timeelapse{\abint}_{n-1},\apedge_{n-1})(\aloc_n,\abint_n)
\]
and a finite path $\afinpath$ of $\oimdpfrom{\acdpta}$, where $\afinpath$ equals
\[
(\aloc_0',\abint_0') ((\timeelapse{\abint}_0',\apedge_0'),\ana_0) (\dir_0)
(\aloc_1',\abint_1') ((\timeelapse{\abint}_1',\apedge_1'),\ana_1) (\dir_1) \linebreak[0]
\cdots \linebreak[0]
((\timeelapse{\abint}_{m-1}',\apedge_{m-1}'),\ana_{m-1}) \linebreak[0] (\dir_{m-1}) \linebreak[0] (\aloc_m',\abint_m') ,
\]
we say that $\afinpath$ \emph{corresponds to} $\abfinpath$
if (1)~$n=m$,
(2)~$\aloc_i = \aloc_i'$, $\abint_i = \abint_i'$ for all $0 \leq i \leq n$,
and (3)~$\timeelapse{\abint}_i = \timeelapse{\abint}_i'$ and $\apedge_i = \apedge_i'$ for all $0 \leq i < n$.
Given a scheduler $\anosched \in \strats{\sem{\oimdpfrom{\acdpta}}}$
and a $\borders$-path $\abfinpath$,
we let $\corrtwo{\anosched}{\abfinpath} \subseteq \mfinpaths{\anosched}$
be the set of finite paths of $\anosched$ that correspond to $\abfinpath$.

Recall that $\oimdpfrom{\acdpta}$ alternates between regions (elements of $\regstates$)
and endpoint indicators (elements of $\eis$).
We partition $\mfinpaths{\sem{\oimdpfrom{\acdpta}}}(\regof{\aloc,0})$
into the set $\mfinpathsendsin{\sem{\oimdpfrom{\acdpta}}}{\regstates}(\regof{\aloc,0})$
of finite paths of the form
\[
(\aloc_0,\abint_0) ((\timeelapse{\abint}_0,\apedge_0),\ana_0) (\dir_0)
(\aloc_1,\abint_1) ((\timeelapse{\abint}_1,\apedge_1),\ana_1) (\dir_1)
\cdots \linebreak[0]
((\timeelapse{\abint}_{n-1},\apedge_{n-1}),\ana_{n-1}) \linebreak[0] (\dir_{n-1}) \linebreak[0] (\aloc_n,\abint_n)
\]
that end with a region,
and the set $\mfinpathsendsin{\sem{\oimdpfrom{\acdpta}}}{\eis}(\regof{\aloc,0})$
of finite paths of the form
\[
(\aloc_0,\abint_0) ((\timeelapse{\abint}_0,\apedge_0),\ana_0) (\dir_0)
(\aloc_1,\abint_1) ((\timeelapse{\abint}_1,\apedge_1),\ana_1) (\dir_1)
\cdots
((\timeelapse{\abint}_m,\apedge_m),\ana_m) (\dir_m)
\]
that end with an endpoint indicator
(note that $\dir_m$, in the context of the last transition of the finite path,
denotes the endpoint indicator $(\timeelapse{\abint}_m,\apedge_m,\dir_m)$).

The construction of $\anosched$ is based on the following principles:
for any two finite paths that correspond to the same $\borders$-path,
the choice made by $\anosched$ after those finite paths will be the same;
furthermore, for any action of $\oimdpfrom{\acdpta}$ of the form $(\timeelapse{\abint},\apedge)$,
the scheduler $\anosched$ assigns positive probability to at most one action of $\sem{\oimdpfrom{\acdpta}}$
of the form $((\timeelapse{\abint},\apedge),\ana)$,
where $\ana$ is used to represent the actual clock valuation chosen by $\ansched$
within the interval $\timeelapse{\abint}$.

We now describe the formal construction of $\anosched$.
As elsewhere in this paper, we proceed by induction on the length of paths.
Let $\afinpath \in \mfinpathsendsin{\sem{\oimdpfrom{\acdpta}}}{\regstates}(\regof{\aloc,0})$
be a finite path of $\oimdpfrom{\acdpta}$ that ends in a region,
and assume that we have already defined the choices of $\anosched$ along all of the prefixes of $\afinpath$.
Let $\abfinpath$ be the unique $\borders$-path to which $\afinpath$ corresponds
(it is obtained by simply removing the components denoted by $\ana$ and $\dir)$.
If $\corrtwo{\ansched}{\abfinpath} = \emptyset$,
i.e., there is no path of $\ansched$ that corresponds to $\abfinpath$,
then the choice of $\anosched$ after $\afinpath$ can be made in an arbitrary manner.
Otherwise,
letting $\afinpath' \in \mfinpaths{\ansched}$ be the unique path of $\ansched$ that corresponds to $\abfinpath$,
we aim to mimic each choice of $(\timeelapse{\aval},\apedge) \in \support(\ansched(\afinpath'))$
in the construction of $\anosched(\afinpath)$.
Consider $(\timeelapse{\aval},\apedge) \in \support(\ansched(\afinpath'))$.
%
%
Let $\timeelapse{\abint} \in \intervalsof{\borders}$ be the unique interval
for which $\timeelapse{\aval} \in \timeelapse{\abint}$.
By definition, $(\timeelapse{\abint},\apedge) \in \omfromactionsof{\oimdpfrom{\acdpta}}{\last(\afinpath)}$.
Then we let $\anosched(\afinpath)((\timeelapse{\abint},\apedge),\ana) = \ansched(\afinpath')(\timeelapse{\aval},\apedge)$,
where $\ana$ is defined as follows:
let $\ana(\timeelapse{\abint},\apedge,\lb) =
\frac{\rb(\timeelapse{\abint}) - \timeelapse{\aval}}{\rb(\timeelapse{\abint}) - \lb(\timeelapse{\abint})}$,
and let $\ana(\timeelapse{\abint},\apedge,\rb) = 1 - \ana(\timeelapse{\abint},\apedge,\lb)
= \frac{\timeelapse{\aval} - \lb(\timeelapse{\abint})}{\rb(\timeelapse{\abint}) - \lb(\timeelapse{\abint})}$.
That is, $\ana$ represents the position of the clock valuation $\timeelapse{\aval}$
within the interval $\timeelapse{\abint}$.
The fact that this definition of $\ana$ is adequate for our purposes is due to the fact that
clock dependencies are affine,
as we will see later in the proof.
This completes the construction of $\anosched$,
because only the action $\dummyact$ is available in the final states of finite paths in
$\mfinpathsendsin{\sem{\oimdpfrom{\acdpta}}}{\eis}(\regof{\aloc,0})$,
and because $\anosched$ can be defined in an arbitrary manner for paths not starting in $\regof{\aloc,0}$.

Next we show that
$\probl{\ansched}{(\aloc,0)}(\Diamond \targetsetof{\targetlocs}) =
\probl{\anosched}{\regof{\aloc,0}}(\Diamond \targetsetimdpof{\targetlocs})$.
Our approach is to show that
$\probl{\ansched}{(\aloc,0)}(\corrtwo{\ansched}{\abfinpath}) =
\probl{\anosched}{\regof{\aloc,0}}(\corrtwo{\anosched}{\abfinpath})$
for all $\borders$-paths $\abfinpath$
(this is sufficient to show
$\probl{\ansched}{(\aloc,0)}(\Diamond \targetsetof{\targetlocs}) =
\probl{\anosched}{\regof{\aloc,0}}(\Diamond \targetsetimdpof{\targetlocs})$
because the path in $\corrtwo{\ansched}{\abfinpath}$,
if it exists,
reaches a location in $\targetlocs$
if and only if all paths in $\corrtwo{\anosched}{\abfinpath}$ reach a location in $\targetlocs$).
We proceed by induction on the length of $\borders$-paths.
Consider the $\borders$-path $\abfinpath (\timeelapse{\abint},\apedge) (\aloc',\abint')$,
and assume that it has already been established that
$\probl{\ansched}{(\aloc,0)}(\corrtwo{\ansched}{\abfinpath}) =
\probl{\anosched}{\regof{\aloc,0}}(\corrtwo{\anosched}{\abfinpath})$.
Our aim is now to show that
$\probl{\ansched}{(\aloc,0)}(\corrtwo{\ansched}{\abfinpath (\timeelapse{\abint},\apedge) (\aloc',\abint')}) =
\probl{\anosched}{\regof{\aloc,0}}(\corrtwo{\anosched}{\abfinpath (\timeelapse{\abint},\apedge) (\aloc',\abint')})$.
As in the construction of $\anosched$,
also in the following we denote $\corrtwo{\ansched}{\abfinpath}$ by $\afinpath'$.
Given that $\ansched$ is a $\borders$-minimal scheduler,
there is at most
one clock valuation $\timeelapse{\aval} \in \timeelapse{\abint}$ such that
$(\timeelapse{\aval},\apedge) \in \support(\asched(\afinpath'))$.
For the case in which
no such $\timeelapse{\aval} \in \timeelapse{\abint}$ exists,
i.e., $\set{ \timeelapse{\aval} \setsep (\timeelapse{\aval},\apedge) \in \support(\asched(\afinpath'))}
\cap \timeelapse{\abint} = \emptyset$,
then
$\probl{\ansched}{(\aloc,0)}(\corrtwo{\ansched}{\abfinpath (\timeelapse{\abint},\apedge) (\aloc',\abint')}) = 0$
and
$\probl{\anosched}{\regof{\aloc,0}}(\corrtwo{\anosched}{\abfinpath (\timeelapse{\abint},\apedge) (\aloc',\abint')}) = 0$,
and we are done.
In the remainder of this proof,
we consider the case in which there exists $\timeelapse{\aval} \in \timeelapse{\abint}$ such that
$(\timeelapse{\aval},\apedge) \in \support(\ansched(\afinpath'))$.
As in the proof of \lemref{lem:bminimal},
we let $\mfromtrans{\acdpta}(\last(\afinpath'),(\timeelapse{\aval},\apedge))(\aloc',\abint')$
be the probability of making a transition to a state in $(\aloc',\abint')$
from state $\last(\afinpath')$ with action $(\timeelapse{\aval},\apedge)$.
We also concentrate on the case in which
$\timeelapse{\abint}$ is open, $\apedge$ is non-constant,
and $\abint' = \timeelapse{\abint}$,
which corresponds to the outcome $(\emptyset,\aloc')$
(i.e., the clock is not reset),
with the other cases being similar.
We write $\fcno$ and $\fdno$ rather than
$\fc{\apedge}{(\emptyset,\aloc')}$ and $\fdno = \fd{\apedge}{(\emptyset,\aloc')}$, respectively.
Let $\ana \in \dist(\eis)$ be the unique assignment such that
$((\timeelapse{\abint},\apedge),\ana) \in \support(\anosched(\afinpath))$
for all $\afinpath \in \corrtwo{\anosched}{\abfinpath}$
(such an assignment will be unique
from the construction of $\anosched$).
%
Then we write $\ana(\dir)$ rather than $\ana(\timeelapse{\abint},\apedge,\dir)$,
for $\dir \in \set{\lb,\rb}$.

First we observe the following (where the second step is obtained from
$\timeelapse{\aval} - \lb(\timeelapse{\abint}) = \ana(\rb) ({\rb(\timeelapse{\abint}) - \lb(\timeelapse{\abint})})$,
and the fourth and fifth steps also use the fact that $\ana(\lb)+\ana(\rb) = 1$):
\begin{eqnarray*}
\fcno + \fdno \cdot \timeelapse{\aval}
& = &
\fcno + \fdno \cdot (\lb(\timeelapse{\abint}) + \timeelapse{\aval} - \lb(\timeelapse{\abint}))
\\
& = &
\fcno + \fdno \cdot (\lb(\timeelapse{\abint}) + \ana(\rb)(\rb(\timeelapse{\abint}) - \lb(\timeelapse{\abint})))
\\
& = &
\fcno + \fdno \cdot (\lb(\timeelapse{\abint}) +
\ana(\rb) \cdot \rb(\timeelapse{\abint}) - \ana(\rb) \cdot \lb(\timeelapse{\abint}))
\\
& = &
\fcno + \fdno \cdot (\lb(\timeelapse{\abint}) +
\ana(\rb) \cdot \rb(\timeelapse{\abint}) - \lb(\timeelapse{\abint}) + \ana(\lb) \cdot \lb(\timeelapse{\abint}))
\\
& = &
(\ana(\lb)+\ana(\rb)) \cdot \fcno
+ \fdno \cdot \ana(\rb) \cdot \rb(\timeelapse{\abint}) + \fdno \cdot \ana(\lb) \cdot \lb(\timeelapse{\abint})
\\
& = &
\ana(\lb) (\fcno + \fdno \cdot \lb(\timeelapse{\abint}))
+ \ana(\rb) (\fcno + \fdno \cdot \rb(\timeelapse{\abint})) \; .
\end{eqnarray*}
The fact that $\fcno + \fdno \cdot \timeelapse{\aval}
= \ana(\lb) (\fcno + \fdno \cdot \lb(\timeelapse{\abint}))
+ \ana(\rb) (\fcno + \fdno \cdot \rb(\timeelapse{\abint}))$
is then used in the following derivation
(where $\adisttemp$ denotes the distribution template of probabilistic edge $\apedge$):
\[
\begin{array}{rcl}
& &
\probl{\ansched}{(\aloc,0)}(\corrtwo{\ansched}{\abfinpath (\timeelapse{\abint},\apedge) (\aloc',\abint')})
\\
& = &
\probl{\ansched}{(\aloc,0)}(\corrtwo{\ansched}{\abfinpath})
\cdot
\ansched(\afinpath')(\timeelapse{\aval},\apedge)
\cdot
\mfromtrans{\acdpta}(\last(\afinpath'),(\timeelapse{\aval},\apedge))(\aloc',\abint')
\hspace{3cm}
\\
& = &
\probl{\ansched}{(\aloc,0)}(\corrtwo{\ansched}{\abfinpath})
\cdot
\ansched(\afinpath')(\timeelapse{\aval},\apedge)
\cdot
\adtpara{\pd}{\timeelapse{\aval}}(\emptyset,\aloc')
\\
& = &
\probl{\ansched}{(\aloc,0)}(\corrtwo{\ansched}{\abfinpath})
\cdot
\ansched(\afinpath')(\timeelapse{\aval},\apedge)
\cdot
( \fcno + \fdno \cdot \timeelapse{\aval} )
\\
& = &
\probl{\ansched}{(\aloc,0)}(\corrtwo{\ansched}{\abfinpath})
\cdot
\ansched(\afinpath')(\timeelapse{\aval},\apedge)
\cdot
\left(
\ana(\lb) (\fcno + \fdno \cdot \lb(\timeelapse{\abint}))
+\ana(\rb) (\fcno + \fdno \cdot \rb(\timeelapse{\abint}))
\right)
\\
& = &
\probl{\ansched}{(\aloc,0)}(\corrtwo{\ansched}{\abfinpath})
\cdot
\ansched(\afinpath')(\timeelapse{\aval},\apedge)
\cdot
\left(
\ana(\lb) \cdot \adtpara{\pd}{\lb(\timeelapse{\abint})}(\emptyset,\aloc')
+ \ana(\rb) \cdot \adtpara{\pd}{\rb(\timeelapse{\abint})}(\emptyset,\aloc')
\right) \; .
\end{array}
\]
Note that the combination of the definition of $\ana$
and the fact that clock dependencies are affine allows us to obtain the last three steps above.
For the subsequent steps, we require the following notation.
Recall that, for any given $(\timeelapse{\abint},\apedge,\dir) \in \eis$,
the interval distribution $\omfromtrans{\oimdpfrom{\acdpta}}((\timeelapse{\abint},\apedge,\dir),\dummyact)$
assigns \emph{singleton} intervals to each
$\astate \in \omfromstates{\oimdpfrom{\acdpta}}$.
Hence, there is only one action associated with $(\timeelapse{\abint},\apedge,\dir)$ in $\sem{\oimdpfrom{\acdpta}}$,
which we denote as $(\dummyact,\dummyassof{\timeelapse{\abint},\apedge,\dir})$,
where $\dummyassof{\timeelapse{\abint},\apedge,\dir}$ is the assignment such that
$\dummyassof{\timeelapse{\abint},\apedge,\dir}(\astate) =  \probof{(\timeelapse{\abint},\apedge,\dir)}{\astate}$
for each $\astate \in \regstates$.
Recall that we have assumed
$\probl{\ansched}{(\aloc,0)}(\corrtwo{\ansched}{\abfinpath})
= \probl{\anosched}{\regof{\aloc,0}}(\corrtwo{\anosched}{\abfinpath})$,
and, from the construction of $\anosched$,
have $\anosched(\afinpath)((\timeelapse{\abint},\apedge),\ana) = \ansched(\afinpath')(\timeelapse{\aval},\apedge)$
for all finite paths $\afinpath \in \corrtwo{\anosched}{\abfinpath}$.
In the following, we let $\afinpathrepr{\anosched}{\abfinpath} \in \corrtwo{\anosched}{\abfinpath}$
be an arbitrary finite path from $\corrtwo{\anosched}{\abfinpath}$.
Then by construction:
\begin{eqnarray*}
& &
\probl{\ansched}{(\aloc,0)}(\corrtwo{\ansched}{\abfinpath})
\cdot
\ansched(\afinpath')(\timeelapse{\aval},\apedge)
\cdot
\left(
\ana(\lb) \cdot \adtpara{\pd}{\lb(\timeelapse{\abint})}(\emptyset,\aloc')
+ \ana(\rb) \cdot \adtpara{\pd}{\rb(\timeelapse{\abint})}(\emptyset,\aloc')
\right)
\\
& = &
\probl{\anosched}{\regof{\aloc,0}}(\corrtwo{\anosched}{\abfinpath})
\cdot
\anosched(\afinpathrepr{\anosched}{\abfinpath})((\timeelapse{\abint},\apedge),\ana)
\\
& &
\times
\left(
\ana(\lb) \cdot \mfromtrans{\oimdpfrom{\acdpta}}((\timeelapse{\abint},\apedge,\lb),(\dummyact,\dummyassof{\timeelapse{\abint},\apedge,\lb}))(\aloc',\abint')
\right.
\\
& &
\left.
~~~+ \ana(\rb) \cdot \mfromtrans{\oimdpfrom{\acdpta}}((\timeelapse{\abint},\apedge,\rb),(\dummyact,\dummyassof{\timeelapse{\abint},\apedge,\rb}))(\aloc',\abint')
\right)
\\
& = &
\probl{\anosched}{\regof{\aloc,0}}(\corrtwo{\anosched}{\abfinpath})
\cdot
\anosched(\afinpathrepr{\anosched}{\abfinpath})((\timeelapse{\abint},\apedge),\ana)
\cdot
\ana(\lb) \cdot \mfromtrans{\oimdpfrom{\acdpta}}((\timeelapse{\abint},\apedge,\lb),(\dummyact,\dummyassof{\timeelapse{\abint},\apedge,\lb}))(\aloc',\abint')
\\
& &
+
\probl{\anosched}{\regof{\aloc,0}}(\corrtwo{\anosched}{\abfinpath})
\cdot
\anosched(\afinpathrepr{\anosched}{\abfinpath})((\timeelapse{\abint},\apedge),\ana)
\cdot
\ana(\rb) \cdot \mfromtrans{\oimdpfrom{\acdpta}}((\timeelapse{\abint},\apedge,\rb),(\dummyact,\dummyassof{\timeelapse{\abint},\apedge,\rb}))(\aloc',\abint')
\\
& = &
\sum_{\afinpath \in \corrtwo{\anosched}{\abfinpath}}
\sum_{\dir \in \set{\lb,\rb}}
\probl{\anosched}{\regof{\aloc,0}}(\afinpath)
\cdot
\anosched(\afinpath)((\timeelapse{\abint},\apedge),\ana) \cdot \ana(\dir)
\hspace{4cm}
\\
& &
~~~~~~~~~~~~~~~~~~~~~\times
\mfromtrans{\oimdpfrom{\acdpta}}((\timeelapse{\abint},\apedge,\dir),(\dummyact,\dummyassof{\timeelapse{\abint},\apedge,\dir}))(\aloc',\abint')
\\
& = &
\probl{\anosched}{\regof{\aloc,0}}(\corrtwo{\anosched}{\abfinpath (\timeelapse{\abint},\apedge) (\aloc',\abint')}) \; .
\end{eqnarray*}
This completes showing that
$\probl{\ansched}{(\aloc,0)}(\corrtwo{\ansched}{\abfinpath (\timeelapse{\abint},\apedge) (\aloc',\abint')}) =
\probl{\anosched}{\regof{\aloc,0}}(\corrtwo{\anosched}{\abfinpath (\timeelapse{\abint},\apedge) (\aloc',\abint')})$
when $\abint'=\timeelapse{\abint}$.
The case for $\abint'=[0,0]$ follows similarly.
\end{proof}

The following lemma considers the converse direction,
namely that (starting from a given location with the clock equal to $0$)
any scheduler of $\sem{\oimdpfrom{\acdpta}}$
can be mimicked by a $\borders$-minimal scheduler of $\semmdp{\acdpta}$
such that the schedulers assign the same probability of reaching a certain set of locations.
%
%
\begin{lem}\label{lem:sched_IMDP2cdPTA}
Let $\aloc \in \locs$ be a location
and let $\targetlocs \subseteq \locs$ be a set of locations.
Given a scheduler $\anosched \in \strats{\sem{\oimdpfrom{\acdpta}}}$,
there exists a $\borders$-minimal scheduler $\ansched \in \bstrats{\semmdp{\acdpta}}$,
such that $\probl{\ansched}{(\aloc,0)}(\Diamond \targetsetof{\targetlocs}) =
\probl{\anosched}{\regof{\aloc,0}}(\Diamond \targetsetimdpof{\targetlocs})$.
\end{lem}
\begin{proof}
Let $\aloc \in \locs$,
$\targetlocs \subseteq \locs$
and $\anosched \in \strats{\sem{\oimdpfrom{\acdpta}}}$.
We describe the construction of scheduler $\ansched \in \bstrats{\semmdp{\acdpta}}$,
and then show that
$\probl{\ansched}{(\aloc,0)}(\Diamond \targetsetof{\targetlocs}) =
\probl{\anosched}{\regof{\aloc,0}}(\Diamond \targetsetimdpof{\targetlocs})$.

Before presenting the details, we sketch the overall approach of the proof.
The proof proceeds in a similar manner to that of \lemref{lem:bminimal}:
here, as for the proof of \lemref{lem:bminimal},
our aim is to obtain a ($\borders$-minimal) scheduler $\ansched$ of $\semmdp{\acdpta}$,
i.e., a scheduler for which there is a one-to-one relationship between its finite paths and
(a subset of) $\borders$-paths.
The principles underlying the construction of $\ansched$ are the same as those underlying
the analogous construction of \lemref{lem:bminimal}:
for a finite path $\afinpath$ of $\semmdp{\acdpta}$,
in order to define $\ansched(\afinpath)$,
we consider extensions of the unique $\borders$-path $\abfinpath$ that corresponds to $\afinpath$;
for each of those extensions that correspond to a set of finite paths of $\anosched$,
we define a choice made by the distribution $\ansched(\afinpath)$
that mimics the final transition of the aforementioned set of finite paths of $\anosched$.
The probabilities of the distribution $\ansched(\afinpath)$
are obtained as weighted averages of the choices of $\anosched$.
Showing that
$\probl{\ansched}{(\aloc,0)}(\Diamond \targetsetof{\targetlocs}) =
\probl{\anosched}{\regof{\aloc,0}}(\Diamond \targetsetimdpof{\targetlocs})$
will be done in a similar manner to the analogous part of the proof of \lemref{lem:sched_cdPTA2IMDP}.

We first describe the construction of $\ansched \in \bstrats{\semmdp{\acdpta}}$.
Consider the finite path
\[
\afinpath =
(\aloc_0,\aval_0) (\timeelapse{\aval}_0,\apedge_0)
(\aloc_1,\aval_1) (\timeelapse{\aval}_1,\apedge_1)
\cdots
(\timeelapse{\aval}_{n-1},\apedge_{n-1}) (\aloc_n,\aval_n)
\]
of $\semmdp{\acdpta}$.
Let
\[
\abfinpath =
(\aloc_0,\abint_0)(\timeelapse{\abint}_0,\apedge_0)
(\aloc_1,\abint_1)(\timeelapse{\abint}_1,\apedge_1)
\cdots (\timeelapse{\abint}_{n-1},\apedge_{n-1})(\aloc_n,\abint_n)
\]
be the unique $\borders$-path such that $\afinpath$ corresponds to $\abfinpath$.
Similarly, recall that $\corrtwo{\anosched}{\abfinpath}$ is the set of finite paths of $\anosched$
that correspond to $\abfinpath$
(where the notion of correspondence of finite paths of $\anosched$ to a $\borders$-path
is given in the proof of \lemref{lem:sched_cdPTA2IMDP}).
Consider the extension of $\abfinpath$ with $(\timeelapse{\abint},\apedge)$
such that $\sourceof{\apedge} = \aloc_n$,
$\timeelapse{\abint} \in \intervalsof{\borders}$,
$\timeelapse{\abint} \geq \abint_n$, and
$\timeelapse{\abint} \vinz \guardof{\apedge} \wedge \inv(\aloc_n)$.

We recall an assumption that we can make without loss of generality,
as already described in the proof of \lemref{lem:imdp2imc}:
given $\afinpath' \in \corrtwo{\anosched}{\abfinpath}$,
we assume that there exists at most one
$\ana \in
\ass{\omfromtrans{\oimdpfrom{\acdpta}}(\last(\afinpath'),(\timeelapse{\abint},\apedge))}$
such that
$((\timeelapse{\abint},\apedge),\ana) \in \support(\anosched(\afinpath'))$
(this follows from the fact that the set of assignments
$\ass{\omfromtrans{\oimdpfrom{\acdpta}}(\last(\afinpath'),(\timeelapse{\abint},\apedge))}$
is closed under convex combinations).
If such $\ana \in
\ass{\omfromtrans{\oimdpfrom{\acdpta}}(\last(\afinpath'),(\timeelapse{\abint},\apedge))}$
exists,
we denote it by $\anaof{\afinpath'}{\timeelapse{\abint}}{\apedge}$.

As in the proof of \lemref{lem:bminimal},
to define $\ansched(\afinpath)$,
for each pair $(\timeelapse{\abint},\apedge)$ used to extend $\abfinpath$,
we identify a clock valuation $\timeelapsespecial{\abfinpath}{\timeelapse{\abint}}{\apedge}$,
and then let:
\begin{eqnarray*}
\ansched(\afinpath)(\timeelapsespecial{\abfinpath}{\timeelapse{\abint}}{\apedge},\apedge) & = &
\frac{\sum_{\afinpath' \in \corrtwo{\anosched}{\abfinpath}} \probfin{\anosched}{\regof{\aloc,0}}(\afinpath')
\cdot
\anosched(\afinpath')((\timeelapse{\abint},\apedge),\anaof{\afinpath'}{\timeelapse{\abint}}{\apedge}) }
{\probfin{\anosched}{\regof{\aloc,0}}(\corrtwo{\anosched}{\abfinpath})} \; ,
\end{eqnarray*}
where $\probfin{\anosched}{\regof{\aloc,0}}(\corrtwo{\anosched}{\abfinpath})
= \sum_{\afinpath' \in \corrtwo{\anosched}{\abfinpath}} \probfin{\anosched}{\regof{\aloc,0}}(\afinpath')$,
and where $\anosched(\afinpath')((\timeelapse{\abint},\apedge),\anaof{\afinpath'}{\timeelapse{\abint}}{\apedge}) = 0$
if there does not exist any $\ana \in
\ass{\omfromtrans{\oimdpfrom{\acdpta}}(\last(\afinpath'),(\timeelapse{\abint},\apedge))}$
such that
$((\timeelapse{\abint},\apedge),\ana) \in \support(\anosched(\afinpath'))$.

For the case in which $\timeelapse{\abint}$ is a closed interval,
i.e., $\timeelapse{\abint}=[\ab,\ab]$ for $\ab \in \borders$,
we let $\timeelapsespecial{\abfinpath}{\timeelapse{\abint}}{\apedge} = \ab$.

We now consider the case in which $\timeelapse{\abint}$ is an open interval.
The approach we take is analogous to that taken in the proof of \lemref{lem:bminimal}:
intuitively, if $\apedge$ is a constant probabilistic edge,
and no non-constant probabilistic edge has been taken while the clock value
remains in $\timeelapse{\abint}$,
then the scheduler $\ansched$ maintains the clock as having a value that is
a (small) amount greater than the left endpoint of $\timeelapse{\abint}$;
if $\apedge$ is instead a non-constant probabilistic edge,
then the value of the clock is advanced to a particular value in the interval $\timeelapse{\abint}$,
in order to mimic the choice of particular assignments in $\anosched$;
finally, if $\apedge$ is a constant probabilistic edge,
and a non-constant probabilistic edge \emph{has} been taken previously while the clock value
remains in $\timeelapse{\abint}$,
then the scheduler $\ansched$ can choose the value of the clock from $\timeelapse{\abint}$
in an arbitrary manner.
As in the proof of \lemref{lem:bminimal},
the fact that we need to consider only these three cases relies on the assumption of initialisation.
The construction is more complicated than that of the proof of \lemref{lem:bminimal},
because the clock valuations must be chosen by $\ansched$
to reflect \emph{assignments} chosen by $\anosched$
(instead, for \lemref{lem:bminimal}, the clock valuations were chosen
in the construction of $\ansched$ to reflect \emph{clock valuations} chosen by $\asched$,
which was simpler conceptually).
As in the proof of \lemref{lem:bminimal},
let $k \leq n$ be the maximum index for which $\abint_k \neq \timeelapse{\abint}$
(with $k = 0$ if no such index exists),
noting that $\timeelapse{\abint}_k = \abint_{k+1} = \timeelapse{\abint}_{k+1} = \cdots
= \timeelapse{\abint}_{n-1} = \abint_n = \timeelapse{\abint}$.
We consider the following function
$\assTOvaln{\timeelapse{\abint}}{\apedge}:
\bigcup_{\apedge' \in \pedges}
\dist(\set{(\timeelapse{\abint},\apedge',\lb), (\timeelapse{\abint},\apedge',\rb)})
\ra \timeelapse{\abint}$
that maps an assignment (over the target endpoint indicators
$(\timeelapse{\abint},\apedge',\lb), (\timeelapse{\abint},\apedge',\rb)$
for some $\apedge' \in \pedges$)
to a clock valuation in $\timeelapse{\abint}$.
The function $\assTOvaln{\timeelapse{\abint}}{\apedge}$ is defined as follows:
given $\ana \in \dist(\set{(\timeelapse{\abint},\apedge',\lb), (\timeelapse{\abint},\apedge',\rb)})$,
we let $\assTOvaln{\timeelapse{\abint}}{\apedge}(\ana)
= \rb(\timeelapse{\abint}) -
\ana(\lb) \cdot (\rb(\timeelapse{\abint}) - \lb(\timeelapse{\abint}))$
(note that, as in the proof of \lemref{lem:sched_cdPTA2IMDP},
we write $\ana(\lb)$ rather than $\ana(\timeelapse{\abint},\apedge',\lb)$ for simplicity;
in the sequel, we will also usually write $\ana(\rb)$ rather than $\ana(\timeelapse{\abint},\apedge',\rb)$).
The intuition underlying the definition of $\assTOvaln{\timeelapse{\abint}}{\apedge}$
is that $\assTOvaln{\timeelapse{\abint}}{\apedge}(\ana)$
is a clock valuation that represents faithfully the assignment $\ana$,
in the sense that the greater the probability of $\ana(\lb)$,
the closer the clock valuation is to the left endpoint of $\timeelapse{\abint}$.

First we consider the subcase in which all probabilistic edges in $\apedge_k \cdots \apedge_n$
are constant, and $\apedge$ is also constant.
While the clock must be set to some value in $\timeelapse{\abint}$,
it is important that the value chosen is not greater than the value needed in the future to replicate
the probabilities corresponding to a non-constant probabilistic edge.
We now identify such values.
We first define the set $\futfpaths{\abfinpath}{\timeelapse{\abint}}{\apedge}$,
which contains finite suffixes of paths in $\corrtwo{\anosched}{\abfinpath}$,
where the suffixes
(1)~start with a transition derived from $(\timeelapse{\abint},\apedge)$,
(2)~feature only $\timeelapse{\abint}$ (apart from the suffixes' first state),
and (3)~features only constant probabilistic edges until the final transition,
which features a non-constant probabilistic edge.
We now define formally $\futfpaths{\abfinpath}{\timeelapse{\abint}}{\apedge}$.
For a finite path $\afinpath' \in \mfinpathsendsin{\anosched}{\regstates}(\regof{\aloc,0})$,
let $\msuffpathfrom{\anosched}{\afinpath'}$ be the set of (finite) suffixes of $\afinpath'$
generated by $\anosched$
that terminate in a state in $\eis$;
formally, $\msuffpathfrom{\anosched}{\afinpath'}$ is the smallest set such that,
if $\afinpath'' \in \mfinpathsendsin{\anosched}{\eis}(\regof{\aloc,0})$,
where $\afinpath'' = \afinpath' \afinpath'''$,
then $\afinpath''' \in \msuffpathfrom{\anosched}{\afinpath'}$.
We let $\msuffpath{\anosched}(\regof{\aloc,0})$ be the set of all suffixes of finite paths of $\anosched$,
starting in a state of $\regstates$ and terminating in a state in $\eis$, that is:
\[
\msuffpath{\anosched}(\regof{\aloc,0})
= \bigcup_{\afinpath' \in \mfinpathsendsin{\anosched}{\regstates}(\regof{\aloc,0})}
\msuffpathfrom{\anosched}{\afinpath'} \; .
\]
Consider $\afinpath' \in \msuffpath{\anosched}(\regof{\aloc,0})$,
where \[\afinpath' = (\aloc_0,\abint_0) ((\timeelapse{\abint}_0,\apedge_0),\ana_0) (\dir_0)
\cdots \linebreak[0]
((\timeelapse{\abint}_m,\apedge_m),\ana_m) (\dir_m).\]
We let $\finalass{\afinpath'} = \ana_m$
be the assignment featured in the final transition of $\afinpath'$.
Furthermore, we say that:
\begin{itemize}
\item
$\afinpath'$ \emph{starts with $(\timeelapse{\abint},\apedge)$}
if $\timeelapse{\abint}_0 = \timeelapse{\abint}$
and $\apedge_0 = \apedge$;
\item
$\afinpath'$ \emph{remains in $\timeelapse{\abint}$}
if $\timeelapse{\abint}_i = \timeelapse{\abint}$ for all $0 \leq i \leq m$
and $\abint_i = \timeelapse{\abint}$ for all $0 < i \leq m$;
\item
$\afinpath'$ \emph{terminates with a non-constant probabilistic edge}
if $\apedge_m$ is a non-constant probabilistic edge
and $\apedge_i$ is a constant probabilistic edge for all $0 \leq i < m$.
\end{itemize}
We say that $\afinpath' \in \msuffpath{\anosched}(\regof{\aloc,0})$
is a \emph{$(\timeelapse{\abint},\apedge)$-critical path}
if $\afinpath'$ starts with $(\timeelapse{\abint},\apedge)$,
remains in $\timeelapse{\abint}$,
and terminates with a non-constant probabilistic edge.
Let $\critical{(\timeelapse{\abint},\apedge)}{\anosched}{(\regof{\aloc,0})}$
be the set of $(\timeelapse{\abint},\apedge)$-critical paths.
We can now let:
\[
\futfpaths{\abfinpath}{\timeelapse{\abint}}{\apedge}
=
\critical{(\timeelapse{\abint},\apedge)}{\anosched}{(\regof{\aloc,0})}
\cap
\bigcup_{\afinpath' \in \corrtwo{\anosched}{\abfinpath}} \msuffpathfrom{\anosched}{\afinpath'} \; .
\]
Then we let $\futass{\abfinpath}{\timeelapse{\abint}}{\apedge}$
be the set of assignments featured in the final transition
of finite paths in $\futfpaths{\abfinpath}{\timeelapse{\abint}}{\apedge}$:
formally $\futass{\abfinpath}{\timeelapse{\abint}}{\apedge}
= \set{\finalass{\afinpath'} \setsep \afinpath' \in \futfpaths{\abfinpath}{\timeelapse{\abint}}{\apedge}}$.
Finally we let $\timeelapsespecial{\abfinpath}{\timeelapse{\abint}}{\apedge}
= \min \set{ \assTOvaln{\timeelapse{\abint}}{\apedge}(\ana)
\setsep \ana \in \futass{\abfinpath}{\timeelapse{\abint}}{\apedge} }$.

Now consider the subcase in which $\timeelapse{\abint}$ is an open interval,
and all probabilistic edges in the sequence $\apedge_k \cdots \apedge_n$
are constant, but $\apedge$ is non-constant.
As in the analogous case of the proof of \lemref{lem:bminimal},
the choice of clock valuation by $\ansched$ corresponds to a weighted average
from a set of clock valuations.
In the case of this proof,
the set of clock valuations used is that obtained from assignments used
by $\anosched$ after finite paths corresponding to $\abfinpath$ extended
with $(\timeelapse{\abint},\apedge)$,
using a similar ``transformation from assignment to clock valuation''
approach as used in the previous paragraph.
We now give the formal details.
Recall that we have assumed w.l.o.g.\ that,
given $\afinpath' \in \corrtwo{\anosched}{\abfinpath}$,
there exists at most one
$\ana \in
\ass{\omfromtrans{\oimdpfrom{\acdpta}}(\last(\afinpath'),(\timeelapse{\abint},\apedge))}$
such that
$((\timeelapse{\abint},\apedge),\ana) \in \support(\anosched(\afinpath'))$,
and that we denote such an assignment $\ana$ by $\anaof{\afinpath'}{\timeelapse{\abint}}{\apedge}$.
Then we let:
\begin{eqnarray*}
\timeelapsespecial{\abfinpath}{\timeelapse{\abint}}{\apedge} & = &
\frac{
\sum_{\afinpath' \in \corrtwo{\anosched}{\abfinpath}} \probfin{\anosched}{(\regof{\aloc,0})}(\afinpath')
\cdot
\anosched(\afinpath')((\timeelapse{\abint},\apedge),\anaof{\afinpath'}{\timeelapse{\abint}}{\apedge})
\cdot
\assTOvaln{\timeelapse{\abint}}{\apedge}(\anaof{\afinpath'}{\timeelapse{\abint}}{\apedge})
}
{
\sum_{\afinpath' \in \corrtwo{\anosched}{\abfinpath}}
\probfin{\anosched}{(\regof{\aloc,0})}(\afinpath')
\cdot
\anosched(\afinpath')((\timeelapse{\abint},\apedge),\anaof{\afinpath'}{\timeelapse{\abint}}{\apedge})
} \; .
\end{eqnarray*}

Finally, we consider the subcase in which $\timeelapse{\abint}$ is an open interval,
there exists $i \in \set{k,\ldots,n}$ such that $\apedge_i$ is non-constant,
and for all other $j \in \set{k,\ldots,n}$ we have that $\apedge_i$ is constant,
and also $\apedge$ is constant.
As in the proof of \lemref{lem:bminimal},
for this subcase, the clock valuation can be arbitrary;
for simplicity we retain the same clock valuation as was used in the final state of $\afinpath$,
i.e., if $\last(\afinpath)$ is equal to $(\aloc',\aval')$,
we let $\timeelapsespecial{\abfinpath}{\timeelapse{\abint}}{\apedge} = \aval'$.

Repeating this overall process for all
$(\timeelapse{\abint},\apedge) \in \omfromactionsof{\oimdpfrom{\acdpta}}{\aloc,\abint}$
suffices to define the distribution $\ansched(\afinpath)$.
Following similar reasoning to the analogous part of \lemref{lem:bminimal},
we can show that $\ansched(\afinpath)$ is indeed a distribution,
i.e.,
$\sum_{(\timeelapse{\abint},\apedge) \in \omfromactionsof{\oimdpfrom{\acdpta}}{\regof{\last(\afinpath)}}}
\ansched(\afinpath)(\timeelapsespecial{\abfinpath}{\timeelapse{\abint}}{\apedge},\apedge)
= 1$
(recall that $\regof{\last(\afinpath)}=(\aloc_n,\abint_n)$).
From the construction of $\ansched$,
and from the fact that $\sum_{(\timeelapse{\abint},\apedge) \in \omfromactionsof{\oimdpfrom{\acdpta}}{\aloc_n,\abint_n}}
\anosched(\afinpath')((\timeelapse{\abint},\apedge),\anaof{\afinpath'}{\timeelapse{\abint}}{\apedge})
= 1$,
we have the following:
\begin{eqnarray*}
& & \sum_{(\timeelapse{\abint},\apedge) \in \omfromactionsof{\oimdpfrom{\acdpta}}{\aloc_n,\abint_n}} \ansched(\afinpath)(\timeelapsespecial{\abfinpath}{\timeelapse{\abint}}{\apedge},\apedge)
\\
& = &
\sum_{(\timeelapse{\abint},\apedge) \in \omfromactionsof{\oimdpfrom{\acdpta}}{\aloc_n,\abint_n}} \left(
\frac{\sum_{\afinpath' \in \corrtwo{\anosched}{\abfinpath}} \probfin{\anosched}{\regof{\aloc,0}}(\afinpath')
\cdot
\anosched(\afinpath')((\timeelapse{\abint},\apedge),\anaof{\afinpath'}{\timeelapse{\abint}}{\apedge}) }
{\probfin{\anosched}{\regof{\aloc,0}}(\corrtwo{\anosched}{\abfinpath})}
\right)
\\
& = &
\frac{
\sum_{\afinpath' \in \corrtwo{\anosched}{\abfinpath}} \probfin{\anosched}{\regof{\aloc,0}}(\afinpath')
\cdot
\sum_{(\timeelapse{\abint},\apedge) \in \omfromactionsof{\oimdpfrom{\acdpta}}{\aloc_n,\abint_n}}
\anosched(\afinpath')((\timeelapse{\abint},\apedge),\anaof{\afinpath'}{\timeelapse{\abint}}{\apedge}) }
{\probfin{\anosched}{\regof{\aloc,0}}(\corrtwo{\anosched}{\abfinpath})}
\\
& = &
\frac{\sum_{\afinpath' \in \corrtwo{\anosched}{\abfinpath}} \probfin{\anosched}{\regof{\aloc,0}}(\afinpath') }
{\probfin{\anosched}{\regof{\aloc,0}}(\corrtwo{\anosched}{\abfinpath})}
\\
& = &
1 \; .
\end{eqnarray*}

The next step of the proof is to show that
$\probl{\ansched}{(\aloc,0)}(\Diamond \targetsetof{\targetlocs}) =
\probl{\anosched}{\regof{\aloc,0}}(\Diamond \targetsetimdpof{\targetlocs})$.
As in previous proofs,
it is sufficient to show that
$\probfin{\ansched}{(\aloc,0)}(\corrtwo{\ansched}{\abfinpath}) =
\probfin{\anosched}{\regof{\aloc,0}}(\corrtwo{\anosched}{\abfinpath})$
for all $\borders$-paths $\abfinpath$.
Proceeding by induction on the length of $\borders$-paths,
consider the $\borders$-path $\abfinpath (\timeelapse{\abint},\apedge) (\aloc',\abint')$,
and assume that
$\probfin{\ansched}{(\aloc,0)}(\corrtwo{\ansched}{\abfinpath}) =
\probfin{\anosched}{\regof{\aloc,0}}(\corrtwo{\anosched}{\abfinpath})$.
Our aim is to show that
\[
\probfin{\ansched}{(\aloc,0)}(\corrtwo{\ansched}{\abfinpath (\timeelapse{\abint},\apedge) (\aloc',\abint')}) =
\probfin{\anosched}{\regof{\aloc,0}}(\corrtwo{\anosched}{\abfinpath (\timeelapse{\abint},\apedge) (\aloc',\abint')}) \; .
\]
We consider the case in which $\timeelapse{\abint}$ is open,
$\apedge$ is non-constant
and $\abint' = \timeelapse{\abint}$,
i.e., the target region $(\aloc',\abint')$ is obtained using the outcome $(\emptyset,\aloc')$
(which does not reset the clock, and hence the interval $\abint'$ of the target region
is equal to the interval $\timeelapse{\abint}$ obtained just before the probabilistic edge $\apedge$
is taken).
Other cases are dealt with similarly;
in particular, the cases for constant probabilistic edges are more straightforward.
Our approach is slightly different from that used in the proof of \lemref{lem:sched_cdPTA2IMDP}
because, in the scheduler of $\sem{\oimdpfrom{\acdpta}}$
constructed in the proof of \lemref{lem:sched_cdPTA2IMDP},
the choice made after all finite paths corresponding to a particular $\borders$-path was the same;
instead, this property does not necessarily hold for $\anosched$.
Note that in this final part of the proof,
the fact that clock dependencies are affine will be used at multiple points.
Let $\afinpath \in \corrtwo{\anosched}{\abfinpath}$.
Now consider the finite paths in
$\corrtwo{\anosched}{\abfinpath (\timeelapse{\abint},\apedge) (\aloc',\abint')}$
that have $\afinpath$ as a prefix:
these finite paths are
$\afinpath
((\timeelapse{\abint},\apedge),\anaof{\afinpath}{\timeelapse{\abint}}{\apedge})(\lb)(\aloc',\abint')$
and
$\afinpath
((\timeelapse{\abint},\apedge),\anaof{\afinpath}{\timeelapse{\abint}}{\apedge})(\rb)(\aloc',\abint')$.
From the construction of $\anosched$ and by definition
(where $\adisttemp$ denotes the distribution template of probabilistic edge $\apedge$,
and $\fcno$ and $\fdno$ denote
$\fc{\apedge}{(\emptyset,\aloc')}$ and $\fd{\apedge}{(\emptyset,\aloc')}$, respectively),
we have the following:
\begin{eqnarray*}
& &
\sum_{\dir \in \set{\lb,\rb}}
\probfin{\anosched}{\regof{\aloc,0}}(\afinpath
((\timeelapse{\abint},\apedge),\anaof{\afinpath}{\timeelapse{\abint}}{\apedge})(\dir)(\aloc',\abint'))
\\
& = &
\probfin{\anosched}{\regof{\aloc,0}}(\afinpath)
\cdot
\anosched(\afinpath)((\timeelapse{\abint},\apedge),\anaof{\afinpath}{\timeelapse{\abint}}{\apedge})
\\
& &
\times
\sum_{\dir \in \set{\lb,\rb}}
\anaof{\afinpath}{\timeelapse{\abint}}{\apedge}(\timeelapse{\abint},\apedge,\dir)
\cdot
\mfromtrans{\oimdpfrom{\acdpta}}((\timeelapse{\abint},\apedge,\dir),(\dummyact,\dummyassof{\timeelapse{\abint},\apedge,\dir}))(\aloc',\abint')
\\
& = &
\probfin{\anosched}{\regof{\aloc,0}}(\afinpath)
\cdot
\anosched(\afinpath)((\timeelapse{\abint},\apedge),\anaof{\afinpath}{\timeelapse{\abint}}{\apedge})
\cdot
\sum_{\dir \in \set{\lb,\rb}}
\anaof{\afinpath}{\timeelapse{\abint}}{\apedge}(\timeelapse{\abint},\apedge,\dir)
\cdot
\adtpara{\pd}{\dir(\timeelapse{\abint})}(\emptyset,\aloc')
\\
& = &
\probfin{\anosched}{\regof{\aloc,0}}(\afinpath)
\cdot
\anosched(\afinpath)((\timeelapse{\abint},\apedge),\anaof{\afinpath}{\timeelapse{\abint}}{\apedge})
\cdot
\sum_{\dir \in \set{\lb,\rb}}
\anaof{\afinpath}{\timeelapse{\abint}}{\apedge}(\timeelapse{\abint},\apedge,\dir)
\cdot
( \fcno + \fdno \cdot \dir(\timeelapse{\abint}) )
\\
& = &
\fcno
\cdot
\probfin{\anosched}{\regof{\aloc,0}}(\afinpath)
\cdot
\anosched(\afinpath)((\timeelapse{\abint},\apedge),\anaof{\afinpath}{\timeelapse{\abint}}{\apedge})
\cdot
\sum_{\dir \in \set{\lb,\rb}}
\anaof{\afinpath}{\timeelapse{\abint}}{\apedge}(\timeelapse{\abint},\apedge,\dir)
\\
& &
+
\,
\fdno
\cdot
\probfin{\anosched}{\regof{\aloc,0}}(\afinpath)
\cdot
\anosched(\afinpath)((\timeelapse{\abint},\apedge),\anaof{\afinpath}{\timeelapse{\abint}}{\apedge})
\cdot
\sum_{\dir \in \set{\lb,\rb}}
\anaof{\afinpath}{\timeelapse{\abint}}{\apedge}(\timeelapse{\abint},\apedge,\dir)
\cdot \dir(\timeelapse{\abint}) \; .
\end{eqnarray*}
Note that $\sum_{\dir \in \set{\lb,\rb}}
\anaof{\afinpath}{\timeelapse{\abint}}{\apedge}(\timeelapse{\abint},\apedge,\dir)
= 1$
(because $\anaof{\afinpath}{\timeelapse{\abint}}{\apedge}$
is a distribution over the set $\set{(\timeelapse{\abint},\apedge,\lb),(\timeelapse{\abint},\apedge,\rb)}$).
Furthermore,
recalling the definition of
$\assTOvaln{\timeelapse{\abint}}{\apedge}$,
we observe that:
\begin{eqnarray*}
\anaof{\afinpath}{\timeelapse{\abint}}{\apedge}(\timeelapse{\abint},\apedge,\lb)
& = &
\frac{\rb(\timeelapse{\abint}) -
\assTOvaln{\timeelapse{\abint}}{\apedge}(\anaof{\afinpath}{\timeelapse{\abint}}{\apedge})}
{\rb(\timeelapse{\abint}) - \lb(\timeelapse{\abint})}
\\
\anaof{\afinpath}{\timeelapse{\abint}}{\apedge}(\timeelapse{\abint},\apedge,\rb)
& = &
\frac{
\assTOvaln{\timeelapse{\abint}}{\apedge}(\anaof{\afinpath}{\timeelapse{\abint}}{\apedge})
- \lb(\timeelapse{\abint})}
{\rb(\timeelapse{\abint}) - \lb(\timeelapse{\abint})} \; .
\end{eqnarray*}
Hence we have:
\begin{eqnarray*}
& &
\sum_{\dir \in \set{\lb,\rb}}
\anaof{\afinpath}{\timeelapse{\abint}}{\apedge}(\timeelapse{\abint},\apedge,\dir)
\cdot \dir(\timeelapse{\abint})
\\
& = &
\frac{\rb(\timeelapse{\abint}) -
\assTOvaln{\timeelapse{\abint}}{\apedge}(\anaof{\afinpath}{\timeelapse{\abint}}{\apedge})}
{\rb(\timeelapse{\abint}) - \lb(\timeelapse{\abint})}
\cdot
\lb(\timeelapse{\abint})
+
\frac{
\assTOvaln{\timeelapse{\abint}}{\apedge}(\anaof{\afinpath}{\timeelapse{\abint}}{\apedge})
- \lb(\timeelapse{\abint})}
{\rb(\timeelapse{\abint}) - \lb(\timeelapse{\abint})}
\cdot
\rb(\timeelapse{\abint})
\\
& = &
\frac{
\rb(\timeelapse{\abint}) \cdot \lb(\timeelapse{\abint})
-
\assTOvaln{\timeelapse{\abint}}{\apedge}(\anaof{\afinpath}{\timeelapse{\abint}}{\apedge})
\cdot \lb(\timeelapse{\abint})
+
\assTOvaln{\timeelapse{\abint}}{\apedge}(\anaof{\afinpath}{\timeelapse{\abint}}{\apedge})
\cdot \rb(\timeelapse{\abint})
-
\lb(\timeelapse{\abint})
\cdot \rb(\timeelapse{\abint})
}
{\rb(\timeelapse{\abint}) - \lb(\timeelapse{\abint})}
\\
& = &
\frac{
\assTOvaln{\timeelapse{\abint}}{\apedge}(\anaof{\afinpath}{\timeelapse{\abint}}{\apedge})
\cdot \rb(\timeelapse{\abint})
-
\assTOvaln{\timeelapse{\abint}}{\apedge}(\anaof{\afinpath}{\timeelapse{\abint}}{\apedge})
\cdot \lb(\timeelapse{\abint})
}
{\rb(\timeelapse{\abint}) - \lb(\timeelapse{\abint})}
\\
& = &
\assTOvaln{\timeelapse{\abint}}{\apedge}(\anaof{\afinpath}{\timeelapse{\abint}}{\apedge}) \; .
\end{eqnarray*}
Therefore we have:
\begin{eqnarray*}
& &
\fcno
\cdot
\probfin{\anosched}{\regof{\aloc,0}}(\afinpath)
\cdot
\anosched(\afinpath)((\timeelapse{\abint},\apedge),\anaof{\afinpath}{\timeelapse{\abint}}{\apedge})
\cdot
\sum_{\dir \in \set{\lb,\rb}}
\anaof{\afinpath}{\timeelapse{\abint}}{\apedge}(\timeelapse{\abint},\apedge,\dir)
\\
& &
+
\,
\fdno
\cdot
\probfin{\anosched}{\regof{\aloc,0}}(\afinpath)
\cdot
\anosched(\afinpath)((\timeelapse{\abint},\apedge),\anaof{\afinpath}{\timeelapse{\abint}}{\apedge})
\cdot
\sum_{\dir \in \set{\lb,\rb}}
\anaof{\afinpath}{\timeelapse{\abint}}{\apedge}(\timeelapse{\abint},\apedge,\dir)
\cdot \dir(\timeelapse{\abint})
\\
& = &
\fcno
\cdot
\probfin{\anosched}{\regof{\aloc,0}}(\afinpath)
\cdot
\anosched(\afinpath)((\timeelapse{\abint},\apedge),\anaof{\afinpath}{\timeelapse{\abint}}{\apedge})
\\
& &
+
\,
\fdno
\cdot
\probfin{\anosched}{\regof{\aloc,0}}(\afinpath)
\cdot
\anosched(\afinpath)((\timeelapse{\abint},\apedge),\anaof{\afinpath}{\timeelapse{\abint}}{\apedge})
\cdot
\assTOvaln{\timeelapse{\abint}}{\apedge}(\anaof{\afinpath}{\timeelapse{\abint}}{\apedge}) \; .
\\
\end{eqnarray*}
We now use this equality to establish that
$\probfin{\ansched}{(\aloc,0)}(\corrtwo{\ansched}{\abfinpath}) =
\probfin{\anosched}{\regof{\aloc,0}}(\corrtwo{\anosched}{\abfinpath})$.
First observe the following:
\begin{eqnarray*}
& &
\probfin{\anosched}{(\aloc,0)}(\corrtwo{\ansched}{\abfinpath (\timeelapse{\abint},\apedge)} (\aloc',\abint'))
\\
& = &
\sum_{\afinpath \in \corrtwo{\anosched}{\abfinpath (\timeelapse{\abint},\apedge)} (\aloc',\abint')}
\probfin{\anosched}{\regof{\aloc,0}}(\afinpath)
\\
& = &
\sum_{\afinpath \in \corrtwo{\anosched}{\abfinpath}}
\sum_{\dir \in \set{\lb,\rb}}
\probfin{\anosched}{\regof{\aloc,0}}(\afinpath
((\timeelapse{\abint},\apedge),\anaof{\afinpath}{\timeelapse{\abint}}{\apedge})(\dir)(\aloc',\abint'))
\\
& = &
\sum_{\afinpath \in \corrtwo{\anosched}{\abfinpath}}
\fcno
\cdot
\probfin{\anosched}{\regof{\aloc,0}}(\afinpath)
\cdot
\anosched(\afinpath)((\timeelapse{\abint},\apedge),\anaof{\afinpath}{\timeelapse{\abint}}{\apedge})
\\
& &
+
\,
\sum_{\afinpath \in \corrtwo{\anosched}{\abfinpath}}
\fdno
\cdot
\probfin{\anosched}{\regof{\aloc,0}}(\afinpath)
\cdot
\anosched(\afinpath)((\timeelapse{\abint},\apedge),\anaof{\afinpath}{\timeelapse{\abint}}{\apedge})
\cdot
\assTOvaln{\timeelapse{\abint}}{\apedge}(\anaof{\afinpath}{\timeelapse{\abint}}{\apedge})
\\
\end{eqnarray*}
Recalling the definition of $\timeelapsespecial{\abfinpath}{\timeelapse{\abint}}{\apedge}$
(for the case of open $\timeelapse{\abint}$ and non-constant $\apedge$),
we have that
$\sum_{\afinpath \in \corrtwo{\anosched}{\abfinpath}} \probfin{\anosched}{(\regof{\aloc,0})}(\afinpath)
\cdot
\anosched(\afinpath)((\timeelapse{\abint},\apedge),\anaof{\afinpath}{\timeelapse{\abint}}{\apedge})
\cdot
\assTOvaln{\timeelapse{\abint}}{\apedge}(\anaof{\afinpath}{\timeelapse{\abint}}{\apedge})
=
\timeelapsespecial{\abfinpath}{\timeelapse{\abint}}{\apedge}
\cdot
\sum_{\afinpath \in \corrtwo{\anosched}{\abfinpath}}
\probfin{\anosched}{(\regof{\aloc,0})}(\afinpath)
\cdot
\anosched(\afinpath)((\timeelapse{\abint},\apedge),\anaof{\afinpath}{\timeelapse{\abint}}{\apedge})
$.
Hence we then obtain:
\begin{eqnarray*}
& &
\sum_{\afinpath \in \corrtwo{\anosched}{\abfinpath}}
\fcno
\cdot
\probfin{\anosched}{\regof{\aloc,0}}(\afinpath)
\cdot
\anosched(\afinpath)((\timeelapse{\abint},\apedge),\anaof{\afinpath}{\timeelapse{\abint}}{\apedge})
\\
& &
+
\,
\sum_{\afinpath \in \corrtwo{\anosched}{\abfinpath}}
\fdno
\cdot
\probfin{\anosched}{\regof{\aloc,0}}(\afinpath)
\cdot
\anosched(\afinpath)((\timeelapse{\abint},\apedge),\anaof{\afinpath}{\timeelapse{\abint}}{\apedge})
\cdot
\assTOvaln{\timeelapse{\abint}}{\apedge}(\anaof{\afinpath}{\timeelapse{\abint}}{\apedge})
\\
& = &
\fcno
\cdot
\sum_{\afinpath \in \corrtwo{\anosched}{\abfinpath}}
\probfin{\anosched}{\regof{\aloc,0}}(\afinpath)
\cdot
\anosched(\afinpath)((\timeelapse{\abint},\apedge),\anaof{\afinpath}{\timeelapse{\abint}}{\apedge})
\\
& &
+
\,
\fdno
\cdot
\timeelapsespecial{\abfinpath}{\timeelapse{\abint}}{\apedge}
\cdot
\sum_{\afinpath \in \corrtwo{\anosched}{\abfinpath}}
\probfin{\anosched}{(\regof{\aloc,0})}(\afinpath)
\cdot
\anosched(\afinpath)((\timeelapse{\abint},\apedge),\anaof{\afinpath}{\timeelapse{\abint}}{\apedge}) \; .
\\
\end{eqnarray*}
Let $\afinpathspecial$
denote the unique path in $\corrtwo{\ansched}{\abfinpath}$.
By the definition of $\ansched(\afinpathspecial)$,
and from the fact that $\probfin{\ansched}{(\aloc,0)}(\corrtwo{\ansched}{\abfinpath})
=
\probfin{\anosched}{\regof{\aloc,0}}(\corrtwo{\anosched}{\abfinpath})$
by induction,
we have:
\begin{eqnarray*}
\sum_{\afinpath \in \corrtwo{\anosched}{\abfinpath}} \probfin{\anosched}{\regof{\aloc,0}}(\afinpath)
\cdot
\anosched(\afinpath)((\timeelapse{\abint},\apedge),\anaof{\afinpath}{\timeelapse{\abint}}{\apedge})
& = &
\probfin{\ansched}{(\aloc,0)}(\corrtwo{\ansched}{\abfinpath})
\cdot
\ansched(\afinpathspecial)(\timeelapsespecial{\abfinpath}{\timeelapse{\abint}}{\apedge},\apedge)  \; .
\end{eqnarray*}
We then obtain:
\begin{eqnarray*}
& &
\fcno
\cdot
\sum_{\afinpath \in \corrtwo{\anosched}{\abfinpath}}
\probfin{\anosched}{\regof{\aloc,0}}(\afinpath)
\cdot
\anosched(\afinpath)((\timeelapse{\abint},\apedge),\anaof{\afinpath}{\timeelapse{\abint}}{\apedge})
\\
& &
+
\,
\fdno
\cdot
\timeelapsespecial{\abfinpath}{\timeelapse{\abint}}{\apedge}
\cdot
\sum_{\afinpath \in \corrtwo{\anosched}{\abfinpath}}
\probfin{\anosched}{(\regof{\aloc,0})}(\afinpath)
\cdot
\anosched(\afinpath)((\timeelapse{\abint},\apedge),\anaof{\afinpath}{\timeelapse{\abint}}{\apedge})
\\
& = &
\fcno
\cdot
\probfin{\ansched}{(\aloc,0)}(\corrtwo{\ansched}{\abfinpath})
\cdot
\ansched(\afinpathspecial)(\timeelapsespecial{\abfinpath}{\timeelapse{\abint}}{\apedge},\apedge)
\\
& &
+
\,
\fdno
\cdot
\timeelapsespecial{\abfinpath}{\timeelapse{\abint}}{\apedge}
\cdot
\probfin{\ansched}{(\aloc,0)}(\corrtwo{\ansched}{\abfinpath})
\cdot
\ansched(\afinpathspecial)(\timeelapsespecial{\abfinpath}{\timeelapse{\abint}}{\apedge},\apedge)
\\
& = &
\probfin{\ansched}{(\aloc,0)}(\corrtwo{\ansched}{\abfinpath})
\cdot
\ansched(\afinpathspecial)(\timeelapsespecial{\abfinpath}{\timeelapse{\abint}}{\apedge},\apedge)
\cdot
(\fcno + \fdno
\cdot
\timeelapsespecial{\abfinpath}{\timeelapse{\abint}}{\apedge}
)
\\
& = &
\probfin{\ansched}{(\aloc,0)}(\corrtwo{\ansched}{\abfinpath})
\cdot
\ansched(\afinpathspecial)(\timeelapsespecial{\abfinpath}{\timeelapse{\abint}}{\apedge},\apedge)
\cdot
\adtpara{\pd}{\timeelapsespecial{\abfinpath}{\timeelapse{\abint}}{\apedge}}(\emptyset,\aloc')
\\
& = &
\probfin{\ansched}{(\aloc,0)}(\corrtwo{\ansched}{\abfinpath})
\cdot
\ansched(\afinpathspecial)(\timeelapsespecial{\abfinpath}{\timeelapse{\abint}}{\apedge},\apedge)
\cdot
\mfromtrans{\acdpta}(\last(\afinpathspecial),
(\timeelapsespecial{\abfinpath}{\timeelapse{\abint}}{\apedge},\apedge))(\aloc',\abint')
\\
& = &
\probfin{\ansched}{(\aloc,0)}(\corrtwo{\ansched}{\abfinpath (\timeelapse{\abint},\apedge) (\aloc',\abint')}) \; .
\\
\end{eqnarray*}
This establishes that
$\probfin{\ansched}{(\aloc,0)}(\corrtwo{\ansched}{\abfinpath (\timeelapse{\abint},\apedge) (\aloc',\abint')}) =
\probfin{\anosched}{\regof{\aloc,0}}(\corrtwo{\anosched}{\abfinpath (\timeelapse{\abint},\apedge) (\aloc',\abint')})$.
We note that the cases for other outcomes are similar,
and for constant probabilistic edges and/or closed intervals are more straightforward.
\end{proof}

We characterise the size of a 1c-cdPTA as the sum of the number of its locations,
the size of the binary encoding of the clock constraints used in invariant conditions and guards,
and the size of the binary encoding of the constants used in the distribution templates of the probabilistic edges
(i.e., $\fc{\apedge}{\outcome}$ and $\fd{\apedge}{\outcome}$
for each $\apedge \in \pedges$ and $\outcome \in 2^{\set{\aclock}} \times \locs$).

\begin{thm}
Quantitative and qualitative problems for 1c-cdPTA can be solved in polynomial time.
\end{thm}

The theorem follows from
\lemref{lem:bminimal},
\lemref{lem:sched_cdPTA2IMDP}, \lemref{lem:sched_IMDP2cdPTA},
\propref{prop:imdp2imc_props},
the fact that the IMDP defined in this section can be constructed in polynomial time,
and the fact that
quantitative and qualitative problems for IMDPs
can be solved in polynomial time,
given that there exist polynomial-time algorithms for analogous problems on IMCs
with the semantics adopted in this paper~\cite{PLSS13,CHK13,CK15,Spr18}.
%
We add that the quantitative and qualitative problems for 1c-cdPTAs are PTIME-hard,
following from the PTIME-hardness of the corresponding problems for MDPs~\cite{PT87,CDH10}.

\section{Conclusion}
We have presented a method for the transformation of a class of 1c-cdPTAs to IMDPs
such that there is a precise relationship between the schedulers of the 1c-cdPTA and the IMDP,
allowing us to use established polynomial-time algorithms for IMDPs to decide
quantitative and qualitative reachability problems on the 1c-cdPTA\@.
Overall, the results establish that such problems are in PTIME\@.
The techniques rely on the initialisation requirement,
which ensures that optimal choices for non-constant probabilistic edges
correspond to the left or right endpoints of intervals
that are derived from the syntactic description of the 1c-cdPTA\@.
The initialisation requirement restricts dependencies between non-constant probabilistic edges:
while this necessarily restricts the expressiveness of the formalism,
the resulting model nevertheless retains the expressive power to represent basic
situations in which the probability of certain events depends on the exact amount of time elapsed,
such as those described in the introduction.

The IMDP construction can be simplified in a number of cases:
for example, in the case in which at most two outcomes $\outcome_1, \outcome_2$
of every probabilistic edge $\apedge$ are non-constant,
i.e., for which $\fd{\apedge}{\outcome_1} \neq 0$ and $\fd{\apedge}{\outcome_2} \neq 0$,
endpoint indicators are unnecessary;
instead, when a probabilistic edge is taken from an open interval $\timeelapse{\abint}$,
each of $\outcome_1$ and $\outcome_2$ are associated with (non-singleton) intervals
(other outcomes are associated with singleton intervals),
and the choice of probability to assign between the two intervals
represents the choice of clock valuation in $\timeelapse{\abint}$.
This construction is also polynomial in the size of the 1c-cdPTA\@.

Future work could consider
time-bounded reachability problems for 1c-cdPTAs,
or lifting one of the two restrictions that
we have applied on the original clock-dependent probabilistic timed automata formalism
as presented in~\cite{Spr21},
namely initialisation and the restriction to one clock.
We discuss these restrictions in turn.
Consider first the case of initialisation.
Observe that initialisation enforces that the choice of clock value
at the point at which a non-constant probabilistic edge is taken
is independent of the choice of clock value at the point
at which any other non-constant probabilistic edge is taken.
For example, for the 1c-cdPTA of \figref{fig:one_clock},
the clock values chosen when leaving location $\mathrm{W}$
and (in the case in which the outcome to location $\mathrm{F}$ is taken)
when leaving location $\mathrm{F}$
are independent from each other,
because they belong to different intervals of $\intervalsof{\borders}$.
Instead, in the example of \figref{fig:notinit} (taken from~\cite{Spr21}),
in which initialisation does not hold,
the choice of clock value when leaving location $\mathrm{B}$ ($\mathrm{C}$, respectively)
depends on (i.e., cannot be less than)
the clock value when leaving location $\mathrm{A}$ ($\mathrm{B}$, respectively).
The characteristic of independence of the choice of clock values
for different non-constant probabilistic edges
permits the reduction from 1c-cdPTAs to IMDPs presented in this paper:
the IMDP construction encodes the choice of clock value of the 1c-cdPTA
by the chosen assignment of an interval distribution
(i.e., the interval distributions with intervals $(0,1)$ available from regions).
If the choice of clock values for different non-constant probabilistic edges are \emph{not} independent,
as is the case when the 1c-cdPTA is not initialised,
then we would require some mechanism in the IMDP to enforce such dependencies
between the choices of assignment in different IMDP states,
which does not exist in the classical IMDP formalism.
Solutions to the qualitative problem for non-initialised 1c-cdPTAs
could potentially utilise connections with parametric MDPs~\cite{HHZ11,WJPK21},
in which dependence between parameters on probabilities of transitions
from different states is an inherent part of the formalism.

\begin{figure}[t]
{

\centering
\scriptsize

\begin{tikzpicture}[->,>=stealth',shorten >=1pt,auto, thin] 

  \tikzstyle{state}=[draw=black, text=black, shape=circle, inner sep=3pt, outer sep=0pt, circle, rounded corners] 


	\node[state, node distance=1.2cm](A){$\mathrm{A}$};
	\node[yshift=-0.5cm,above of = A](inv_A){$\aclock < 1$};

	\node[fill=black, node distance=1.7cm] (nail_A) [right of=A] {};

	\node[state, node distance=1.7cm](B)[right of=nail_A]{$\mathrm{B}$};
	\node[yshift=-0.5cm,above of = B](inv_B){$\aclock < 1$};

	\node[fill=black, node distance=1.7cm] (nail_B) [right of=B] {};

	\node[state, node distance=1.7cm](C)[right of=nail_B]{$\mathrm{C}$};
	\node[yshift=-0.5cm,above of = C](inv_C){$\aclock < 1$};

	\node[fill=black, node distance=1.7cm] (nail_C) [right of=C] {};

	\node[state, node distance=1.7cm](D)[right of=nail_C]{$\mathrm{D}$};

	\node[state, node distance=1.7cm](E)[below of=nail_B]{$\mathrm{E}$};

	\path[rounded corners]

	(A)
		edge [above] node {$\aclock>0$} (nail_A)

	(nail_A)
		edge [above] node {\colorbox{gray!15}{$\aclock$}} (B)

	(B)
		edge [above] node {$\aclock>0$} (nail_B)

	(nail_B)
		edge [above] node {\colorbox{gray!15}{$1-\aclock$}} (C)
		edge [right] node {\colorbox{gray!15}{$\aclock$}} (E)

	(C)
		edge [above] node {$\aclock>0$} (nail_C)

	(nail_C)
		edge [above] node {\colorbox{gray!15}{$1-\frac{\aclock}{2}$}} (D);

\draw[->,rounded corners] (nail_A.south) |- node[above,text=black,pos=.7]{\colorbox{gray!15}{$1-\aclock$}} (E.west);
\draw[->,rounded corners] (nail_C.south) |- node[above,text=black,pos=.7]{\colorbox{gray!15}{$\frac{\aclock}{2}$}} (E.east);

\end{tikzpicture}

}

\caption{A non-initialised 1c-cdPTA (taken from~\cite{Spr21}).}%
\label{fig:notinit}
\end{figure}

For the case of lifting the restriction to one clock,
while retaining initialisation
(i.e., a clock has a natural-numbered value between any two non-constant probabilistic edges
with a dependence on that clock),
a natural way of generalising the results of this paper would be to consider
classical multidimensional regions (as presented for timed automata in~\cite{AD94}
and used for approximate analysis of clock-dependent probabilistic timed automata in~\cite{Spr21})
as states of a finite-state MDP (whether an IMDP or a standard MDP).
Now consider a clock-dependent probabilistic timed automaton
that has the same overall structure as the 1c-cdPTA of \figref{fig:notinit},
but which has three clocks, $\aclock$, $\anotherclock$ and $\yanotherclock$,
and for which the clock dependencies of the three probabilistic edges
are as follows:
for the probabilistic edge from location $\mathrm{A}$ ($\mathrm{B}$, $\mathrm{C}$, respectively),
the outcome to location $\mathrm{B}$ ($\mathrm{C}$, $\mathrm{D}$, respectively)
has clock dependence $\aclock$ ($1 - \anotherclock$, $1 - \frac{\yanotherclock}{2}$, respectively).
This example satisfies the notion of initialisation in the case of multiple clocks,
but features the same dependence between the
times at which the probabilistic edges can be taken
as in the case of the previous paragraph,
and hence a reduction to IMDPs (or even standard MDPs) is likely to be challenging if not impossible.
Future work could explore the applications of a subclass of this formalism,
for example with multiple clocks but only one clock used for clock dependencies,
or a more strict definition of initialisation than that described above:
such subclasses may yield models that have more significance in practice than 1c-cdPTAs
while offering the possibility of exact rather than approximate analysis.

Finally, we also discuss lifting the restriction to affine clock dependencies
in the case of initialised 1c-cdPTAs.
We recall that the IMDP construction presented in this paper
depends on the fact that the distributions corresponding to
the endpoints of intervals in $\intervalsof{\borders}$ are ``extremal'',
in the sense that, for a given interval, non-constant probabilistic edge and outcome,
the maximum or minimum probability assigned to that outcome according to the
probabilistic edge's distribution template 
will be assigned at one of the endpoints of the interval.
This property allows any scheduler of the 1c-cdPTA to be mimicked by a IMDP scheduler
by using probabilistic edges available from regions to obtain linear combinations
of the distributions corresponding to an interval's endpoints.
In the case of non-affine clock dependencies, this property does not necessarily hold,
i.e., the maximum or minimum probability assigned to an outcome
may fall \emph{within} an interval,
rather than at one of its endpoints.
Hence, dealing with non-affine clock dependencies requires non-trivial developments
beyond those featured in this paper,
or requires extra assumptions in order to handle the greater expressiveness of the clock dependencies,
for example restrictions on the non-affine functions used
(for example, so that their maximum or minimum
is always at the endpoint of any interval in $\intervalsof{\borders}$)
or, for any probabilistic edge, on the number of outcomes that can have non-affine clock dependencies.
\bibliographystyle{alpha}
\bibliography{onec-cdpta}

\newcommand{\etalchar}[1]{$^{#1}$}
\begin{thebibliography}{FWHT16}

\bibitem[ABK{\etalchar{+}}16]{ABKMT16}
S.~Akshay, P.~Bouyer, S.~N. Krishna, L.~Manasa, and A.~Trivedi.
\newblock Stochastic timed games revisited.
\newblock In P.~Faliszewski, A.~Muscholl, and R.~Niedermeier, editors, {\em
  {P}roc. MFCS 2016}, volume~58 of {\em LIPIcs}, pages 8:1--8:14.
  Leibniz-Zentrum f{\"u}r Informatik, 2016.

\bibitem[AD94]{AD94}
R.~Alur and D.~L. Dill.
\newblock A theory of timed automata.
\newblock {\em Theoretical Computer Science}, 126(2):183--235, 1994.

\bibitem[BBB{\etalchar{+}}14]{BBBMBGJ14}
N.~Bertrand, P.~Bouyer, T.~Brihaye, Q.~Menet, C.~Baier, M.~Gr{\"{o}}{\ss}er,
  and M.~Jurdzinski.
\newblock Stochastic timed automata.
\newblock {\em Logical Methods in Computer Science}, 10(4), 2014.

\bibitem[BBG14]{BBG14}
N.~Bertrand, T.~Brihaye, and B.~Genest.
\newblock Deciding the value 1 problem for reachability in 1-clock decision
  stochastic timed automata.
\newblock In G.~Norman and W.~H. Sanders, editors, {\em Proc. QEST 2004},
  volume 8657 of {\em LNCS}, pages 313--328. Springer, 2014.

\bibitem[BBH{\etalchar{+}}19]{BBHMPS19}
N.~Bertrand, B.~Bordais, L.~H{\'{e}}lou{\"{e}}t, T.~Mari, J.~Parreaux, and
  O.~Sankur.
\newblock Performance evaluation of metro regulations using probabilistic
  model-checking.
\newblock In S.~Collart Dutilleul, T.~Lecomte, and A.~B. Romanovsky, editors,
  {\em Proc. RSSRail 2019}, volume 11495 of {\em LNCS}, pages 59--76. Springer,
  2019.

\bibitem[BFL{\etalchar{+}}18]{BFLMOW18}
P.~Bouyer, U.~Fahrenberg, K.~G. Larsen, N.~Markey, J.~Ouaknine, and J.~Worrell.
\newblock Model checking real-time systems.
\newblock In E.~M. Clarke, T.~A. Henzinger, H.~Veith, and R.~Bloem, editors,
  {\em Handbook of Model Checking}, pages 1001--1046. Springer, 2018.

\bibitem[BK08]{BK08}
C.~Baier and J.-P. Katoen.
\newblock {\em Principles of model checking}.
\newblock MIT Press, 2008.

\bibitem[BLM08]{BLM08}
P.~Bouyer, K.~G. Larsen, and N.~Markey.
\newblock Model checking one-clock priced timed automata.
\newblock {\em Logical Methods in Computer Science}, 4(2), 2008.

\bibitem[CDH10]{CDH10}
K.~Chatterjee, L.~Doyen, and T.~A. Henzinger.
\newblock Qualitative analysis of partially-observable {M}arkov decision
  processes.
\newblock In Petr Hlinen{\'{y}} and Anton{\'{\i}}n Ku\v{c}era, editors, {\em
  Proc. MFCS 2010}, volume 6281 of {\em LNCS}, pages 258--269. Springer, 2010.

\bibitem[CGP01]{CGP01}
E.~M. Clarke, O.~Grumberg, and D.~A. Peled.
\newblock {\em Model checking}.
\newblock {MIT} Press, 2001.

\bibitem[CHK13]{CHK13}
T.~Chen, T.~Han, and M.~Kwiatkowska.
\newblock On the complexity of model checking interval-valued discrete time
  {M}arkov chains.
\newblock {\em Information Processing Letters}, 113(7):210--216, 2013.

\bibitem[CK15]{CK15}
S.~Chakraborty and J.-P. Katoen.
\newblock Model checking of open interval {M}arkov chains.
\newblock In M.~Gribaudo, D.~Manini, and A.~Remke, editors, {\em Proc. ASMTA
  2015}, volume 9081 of {\em LNCS}, pages 30--42. Springer, 2015.

\bibitem[CSH08]{CSH08}
K.~Chatterjee, K.~Sen, and T.~A. Henzinger.
\newblock Model-checking $\omega$-regular properties of interval {M}arkov
  chains.
\newblock In R.~Amadio, editor, {\em Proc. FOSSACS 2008}, volume 4962 of {\em
  LNCS}, pages 302--317. Springer, 2008.

\bibitem[FKNP11]{FKNP11}
V.~Forejt, M.~Kwiatkowska, G.~Norman, and D.~Parker.
\newblock Automated verification techniques for probabilistic systems.
\newblock In M.~Bernardo and V.~Issarny, editors, {\em Formal Methods for
  Eternal Networked Software Systems (SFM 2011)}, volume 6659 of {\em LNCS},
  pages 53--113. Springer, 2011.

\bibitem[FWHT16]{FWHT16}
L.~Feng, C.~Wiltsche, L.~R. Humphrey, and U.~Topcu.
\newblock Synthesis of human-in-the-loop control protocols for autonomous
  systems.
\newblock {\em {IEEE} Trans. Automation Science and Engineering},
  13(2):450--462, 2016.

\bibitem[GJ95]{GJ95}
H.~Gregersen and H.~E. Jensen.
\newblock Formal design of reliable real time systems.
\newblock Master's thesis, Department of Mathematics and Computer Science,
  Aalborg University, 1995.

\bibitem[GLD00]{GLD00}
R.~Givan, S.~M. Leach, and T.~L. Dean.
\newblock Bounded-parameter {M}arkov decision processes.
\newblock {\em Artificial Intelligence}, 122(1-2):71--109, 2000.

\bibitem[Hah13]{Hah13}
E.~M. Hahn.
\newblock {\em Model checking stochastic hybrid systems}.
\newblock PhD thesis, Universit{\"a}t des Saarlandes, 2013.

\bibitem[HHK14]{HHK14}
V.~Hashemi, H.~Hatefi, and J.~Krc{\'{a}}l.
\newblock Probabilistic bisimulations for {PCTL} model checking of interval
  {MDPs}.
\newblock In {\'{E}}.~Andr{\'{e}} and G.~Frehse, editors, {\em Proc. SynCoP
  2014}, volume 145 of {\em EPTCS}, pages 19--33, 2014.

\bibitem[HHZ11]{HHZ11}
E.~M. Hahn, T.~Han, and L.~Zhang.
\newblock Synthesis for {PCTL} in parametric {M}arkov decision processes.
\newblock In M.~G. Bobaru, K.~Havelund, G.~J. Holzmann, and R.~Joshi, editors,
  {\em Proc. NFM 2011}, volume 6617 of {\em LNCS}, pages 146--161. Springer,
  2011.

\bibitem[HKPV98]{HKPV98}
T.~A. Henzinger, P.~W. Kopke, A.~Puri, and P.~Varaiya.
\newblock What's decidable about hybrid automata?
\newblock {\em Journal of Computer and System Sciences}, 57(1):94--124, 1998.

\bibitem[HM18]{HM18}
S.~Haddad and B.~Monmege.
\newblock Interval iteration algorithm for {MDPs} and {IMDPs}.
\newblock {\em Theoretical Computer Science}, 735:111--131, 2018.

\bibitem[JL91]{JL91}
B.~Jonsson and K.~G. Larsen.
\newblock Specification and refinement of probabilistic processes.
\newblock In {\em Proc. LICS 1991}, pages 266--277. {IEEE} Computer Society,
  1991.

\bibitem[JLS08]{JLS08}
M.~Jurdzi\'{n}ski, F.~Laroussinie, and J.~Sproston.
\newblock Model checking probabilistic timed automata with one or two clocks.
\newblock {\em Logical Methods in Computer Science}, 4(3):1--28, 2008.

\bibitem[KNSS02]{KNSS02}
M.~Kwiatkowska, G.~Norman, R.~Segala, and J.~Sproston.
\newblock Automatic verification of real-time systems with discrete probability
  distributions.
\newblock {\em Theoretical Computer Science}, 286:101--150, 2002.

\bibitem[KU02]{KU02}
I.~O. Kozine and L.~V. Utkin.
\newblock Interval-valued finite {M}arkov chains.
\newblock {\em Reliable Computing}, 8(2):97--113, 2002.

\bibitem[LMS04]{LMS04}
F.~Laroussinie, N.~Markey, and P.~Schnoebelen.
\newblock Model checking timed automata with one or two clocks.
\newblock In P.~Gardner and N.~Yoshida, editors, {\em Proc. CONCUR 2004},
  volume 3170 of {\em LNCS}, pages 387--401. Springer, 2004.

\bibitem[NE05]{NG05}
A.~Nilim and L.~{El Ghaoui}.
\newblock Robust control of {M}arkov decision processes with uncertain
  transition matrices.
\newblock {\em Operations Research}, 53(5):780--798, 2005.

\bibitem[NPS13]{NPS13}
G.~Norman, D.~Parker, and J.~Sproston.
\newblock Model checking for probabilistic timed automata.
\newblock {\em Formal Methods in System Design}, 43(2):164--190, 2013.

\bibitem[PLSS13]{PLSS13}
A.~Puggelli, W.~Li, A.~L. Sangiovanni{-}Vincentelli, and S.~A. Seshia.
\newblock Polynomial-time verification of {PCTL} properties of {MDPs} with
  convex uncertainties.
\newblock In N.~Sharygina and H.~Veith, editors, {\em Proc. CAV 2013}, volume
  8044 of {\em LNCS}, pages 527--542. Springer, 2013.

\bibitem[PT87]{PT87}
C.~H. Papadimitriou and J.~N. Tsitsiklis.
\newblock The complexity of {M}arkov decision processes.
\newblock {\em Mathematics of Operations Research}, 12(3):441--450, 1987.

\bibitem[Spr18]{Spr18}
J.~Sproston.
\newblock Qualitative reachability for open interval {M}arkov chains.
\newblock In I.~Potapov and P.-A. Reynier, editors, {\em Proc. RP 2018}, volume
  11123 of {\em LNCS}, pages 146--160. Springer, 2018.

\bibitem[Spr20]{Spr20}
J.~Sproston.
\newblock Probabilistic timed automata with one clock and initialised
  clock-dependent probabilities.
\newblock In A.~Gotsman and A.~Sokolova, editors, {\em Proc. FORTE 2020},
  volume 12136 of {\em LNCS}, pages 150--168. Springer, 2020.

\bibitem[Spr21]{Spr21}
J.~Sproston.
\newblock Probabilistic timed automata with clock-dependent probabilities.
\newblock {\em Fundamenta Informaticae}, 178(1--2):101--138, 2021.

\bibitem[SVA06]{SVA06}
K.~Sen, M.~Viswanathan, and G.~Agha.
\newblock Model-checking {M}arkov chains in the presence of uncertainties.
\newblock In H.~Hermanns and J.~Palsberg, editors, {\em Proc. TACAS 2006},
  volume 3920 of {\em LNCS}, pages 394--410, 2006.

\bibitem[WJPK21]{WJPK21}
T.~Winkler, S.~Junges, G.~A. P{\'{e}}rez, and J.-P. Katoen.
\newblock On the complexity of reachability in parametric {M}arkov decision
  processes.
\newblock {\em Journal of Computer and System Sciences}, 119:183--210, 2021.

\bibitem[WK08]{WK08}
D.~Wu and X.~D. Koutsoukos.
\newblock Reachability analysis of uncertain systems using bounded-parameter
  {M}arkov decision processes.
\newblock {\em Artificial Intelligence}, 172(8-9):945--954, 2008.

\end{thebibliography}

\end{document}